\RequirePackage{rotating}
\RequirePackage{amsmath}
\documentclass[smallcondensed]{svjour3}
\usepackage{mathptmx}
\usepackage{stmaryrd}
\usepackage{amssymb}
\usepackage{listings}
\usepackage{url}
\usepackage{framed}
\usepackage{array}
\usepackage{threeparttable}
\usepackage{verbatim}
\usepackage{ifthen}
\usepackage{mathtools}
\usepackage{algorithm}
\usepackage{algorithmic}
\usepackage{xcolor}
\usepackage{afterpage}
\usepackage[bookmarksdepth=3]{hyperref}
\usepackage{doi}

\newcommand*{\uns}{\text{unsat.}}

\spnewtheorem*{delayedproof}{Proof}{\itshape}{\rmfamily}

\newcommand*{\reld}[3]{\mathord{#1}\subseteq#2\times#3}
\newcommand*{\fund}[3]{\mathord{#1}\colon#2\to#3}
\newcommand*{\pard}[3]{\mathord{#1}\colon#2\rightarrowtail#3}

\providecommand*{\Nset}{\mathbb{N}}            
\providecommand*{\Rset}{\mathbb{R}}            
\providecommand*{\extRset}{\overline{\mathbb{R}}} 
\providecommand*{\Fset}{\mathbb{F}}            

\newcommand*{\cI}{\ensuremath{\mathcal{I}}}

\newcommand*{\cN}{\ensuremath{\mathcal{N}}}

\newcommand*{\nan}{\ensuremath{\mathrm{NaN}}}

\newcommand*{\fpred}{\mathop{\mathrm{pred}}\nolimits}
\newcommand*{\fsucc}{\mathop{\mathrm{succ}}\nolimits}

\newcommand*{\feven}{\mathop{\mathrm{even}}\nolimits}
\newcommand*{\fodd}{\mathop{\mathrm{odd}}\nolimits}

\newcommand*{\ferrdown}{\mathop{\nabla^\downarrow}\nolimits}
\newcommand*{\ferrup}{\mathop{\nabla^\uparrow}\nolimits}
\newcommand*{\ftwiceerrnearneg}{\mathop{\nabla_2^{\mathrm{n}-}}\nolimits}
\newcommand*{\ftwiceerrnearpos}{\mathop{\nabla_2^{\mathrm{n}+}}\nolimits}

\newcommand*{\sgn}{\mathop{\mathrm{sgn}}\nolimits}

\newcommand*{\rup}{{\mathord{\uparrow}}}
\newcommand*{\rdown}{{\mathord{\downarrow}}}
\newcommand*{\rzero}{{0}}
\newcommand*{\rnear}{{\mathrm{n}}}

\newcommand*{\round}[2]{[#2]_{#1}}
\newcommand*{\roundup}[1]{\round{\rup}{#1}}
\newcommand*{\rounddown}[1]{\round{\rdown}{#1}}
\newcommand*{\roundzero}[1]{\round{\rzero}{#1}}
\newcommand*{\roundnear}[1]{\round{\rnear}{#1}}
\newcommand*{\biground}[2]{\bigl[#2\bigr]_{#1}}
\newcommand*{\bigroundup}[1]{\biground{\rup}{#1}}
\newcommand*{\bigrounddown}[1]{\biground{\rdown}{#1}}

\newcommand*{\fmin}{{f_\mathrm{min}}}
\newcommand*{\fmax}{{f_\mathrm{max}}}
\newcommand*{\fnormin}{{f^\mathrm{nor}_\mathrm{min}}}
\newcommand*{\emin}{{e_\mathrm{min}}}
\newcommand*{\emax}{{e_\mathrm{max}}}
\newcommand*{\float}[2]{\ensuremath{\mathalpha{#1} \times \mathalpha{2^{#2}}}}

\newcommand*{\slt}{\prec}
\newcommand*{\sgt}{\succ}
\newcommand*{\sleq}{\preccurlyeq}
\newcommand*{\sgeq}{\succcurlyeq}

\newcommand*{\eval}[2]{\llbracket #2 \rrbracket_{#1}}
\newcommand*{\evalup}[1]{\eval{\rup}{#1}}
\newcommand*{\evaldown}[1]{\eval{\rdown}{#1}}

\newcommand{\mop}{\boxcircle}
\newcommand{\madd}{\boxplus}
\newcommand{\msub}{\boxminus}
\newcommand{\mmul}{\boxdot}
\newcommand{\mdiv}{\boxslash}

\newcommand{\defrel}[1]{\mathrel{\buildrel \mathrm{def} \over {#1}}}
\newcommand{\defeq}{\defrel{=}}

\newcommand{\summary}[1]{\textrm{\textbf{\textup{#1}}}}

\newcommand*{\Cplusplus}{{C\nolinebreak[4]\hspace{-.05em}\raisebox{.4ex}{\tiny\bf ++}}}

\renewcommand{\emptyset}{\varnothing}

\newcommand*{\sseq}{\subseteq}

\newcommand*{\sslt}{\subset}

\newcommand*{\union}{\cup}

\newcommand*{\inters}{\cap}
\newcommand*{\setdiff}{\setminus}

\newcommand{\st}{\mathrel{.}}
\newcommand{\itc}{\mathrel{:}}

\newcommand*{\assign}{\mathrel{\mathord{:}\mathord{=}}}

\newcommand*{\diraddl}{\mathrm{da}_\ell}
\newcommand*{\diraddu}{\mathrm{da}_u}
\newcommand*{\invaddl}{\mathrm{ia}_\ell}
\newcommand*{\invaddu}{\mathrm{ia}_u}

\newcommand*{\dirsubl}{\mathrm{ds}_\ell}
\newcommand*{\dirsubu}{\mathrm{ds}_u}
\newcommand*{\invfirstsubl}{\mathrm{is}^f_\ell}
\newcommand*{\invfirstsubu}{\mathrm{is}^f_u}
\newcommand*{\invsecsubl}{\mathrm{is}^s_\ell}
\newcommand*{\invsecsubu}{\mathrm{is}^s_u}

\newcommand*{\dirmull}{\mathrm{dm}_\ell}
\newcommand*{\dirmulu}{\mathrm{dm}_u}
\newcommand*{\invmull}{\mathrm{im}_\ell}
\newcommand*{\invmulu}{\mathrm{im}_u}

\newcommand*{\dirdivl}{\mathrm{dd}_\ell}
\newcommand*{\dirdivu}{\mathrm{dd}_u}
\newcommand*{\invfirstdivl}{\mathrm{id}^f_\ell}
\newcommand*{\invfirstdivu}{\mathrm{id}^f_u}
\newcommand*{\invsecdivl}{\mathrm{id}^s_\ell}
\newcommand*{\invsecdivu}{\mathrm{id}^s_u}

\newcolumntype{L}{>{\quad$}l<{$\quad}}
\newcolumntype{C}{>{\quad$}c<{$\quad}}

\definecolor{deepblue}{rgb}{0,0,0.5}
\definecolor{deepred}{rgb}{0.6,0,0}
\definecolor{deepgreen}{rgb}{0,0.5,0}
\definecolor{Keywords}{rgb}{0,0,1}
\definecolor{functions}{rgb}{0.75,0,0.5}
\definecolor{boolean}{rgb}{1,0.4,0}

\begin{document}

\title{Correct Approximation of IEEE~754
       Floating-Point Arithmetic
       for Program Verification
}
\author{Roberto~Bagnara
\and Abramo~Bagnara
\and Fabio~Biselli
\and Michele~Chiari
\and Roberta~Gori
}
\authorrunning{R.~Bagnara,
               A.~Bagnara,
               F.~Biselli,
               M.~Chiari,
               R.~Gori
}
\titlerunning{Correct Approximation of IEEE~754 Floating-Point Arithmetic}

\institute{R.~Bagnara
    \at Dipartimento di Scienze Matematiche, Fisiche e Informatiche,
    Universit\`a di Parma,
    Italy \\
    \email{bagnara@cs.unipr.it}
\and
    R.~Bagnara, A.~Bagnara, F.~Biselli,  M.~Chiari
    \at BUGSENG srl,
    \url{http://bugseng.com},
    Italy \\
    \email{{\it name.surname}@bugseng.com}
\and
    F.~Biselli
    \at Certus Software V\&V Center, SIMULA Research Laboratory,
    Norway
\and
    M.~Chiari
    \at Dipartimento di Elettronica, Informazione e Bioingegneria, Politecnico di Milano,
    Italy \\
    \email{michele.chiari@polimi.it}
\and
    R.~Gori
    \at Dipartimento di Informatica,
    Universit\`a di Pisa,
    Italy \\
    \email{gori@di.unipi.it}
}

\date{}
\maketitle

\begin{abstract}
Verification of programs using floating-point arithmetic is challenging
on several accounts.  One of the difficulties of reasoning about
such programs is due to the peculiarities of floating-point arithmetic:
rounding errors, infinities, non-numeric objects (NaNs), signed zeroes,
denormal numbers, different rounding modes, etc.
One possibility to reason about floating-point
arithmetic is to model a program computation path by means of a set
of ternary constraints of the form $z = x \mop y$ and use constraint
propagation techniques to infer new information on the variables'
possible values.  In this setting, we define and prove the correctness
of algorithms to precisely bound the value of one of the variables
$x$, $y$ or $z$, starting from the bounds known for the other two.
We do this for each of the operations and for each rounding mode
defined by the IEEE~754 binary floating-point standard,
even in the case the rounding mode in effect is only partially known.
This is the first time that such so-called \emph{filtering algorithms}
are defined and their correctness is formally proved.
This is an important slab for paving the way to formal verification
of programs that use floating-point arithmetics.
\end{abstract}

\section{Introduction}
\label{sec:introduction}

Programs using floating-point numbers are notoriously difficult to reason
about \cite{Monniaux08}.  Many factors complicate the task:
\begin{enumerate}
\item
compilers may transform the code in a way that does not preserve
the semantics of floating-point computations;
\item
floating-point formats are an implementation-defined aspect
of most programming languages;
\item
there are different, incompatible implementations of the operations
for the same floating-point format;
\item
mathematical libraries often come with little or no guarantee about
what is actually computed;
\item \label{item:issue-anomalies}
programmers have a hard time predicting and avoiding phenomena caused
by the limited range and precision of floating-point
numbers (overflow, absorption, cancellation, underflow, etc.);
moreover, devices that modern floating-point formats possess in order to support
better handling of such phenomena (infinities, signed zeroes,
denormal numbers, non-numeric objects a.k.a.\ NaNs) come with
their share of issues;
\item \label{item:issue-rounding-modes}
rounding is a source of confusion in itself; moreover, there are
several possible rounding modes and programs can change the rounding
mode any time.
\end{enumerate}
As a result of these difficulties, the verification of floating-point
programs in industry relies, almost exclusively, on informal methods,
mainly testing, or on the evaluation of the numerical accuracy of computations,
which only allows to determine conservative (but often too loose)
bounds on the propagated error \cite{DelmasGPSTV09}.

The satisfactory formal treatment of programs engaging in floating-point
computations requires an equally satisfactory solution to the difficulties
summarized in the above enumeration.
Progress has been made, but more remains to be done.
Let us review each point:
\begin{enumerate}
\item
Some compilers provide options to refrain from rearranging floating-point
computations.
When these are not available or cannot be used, the only possibility
is to verify the generated machine code or some intermediate code
whose semantics is guaranteed to be preserved by the compiler back-end.
\item
Even though the used floating-point formats are implementation-defined aspects
of, say, C and \Cplusplus\footnote{This is not relevant if one
analyzes machine or sufficiently low-level intermediate code.}
the wide adoption of the IEEE~754 standard for binary floating-point
arithmetic \cite{IEEE-754-2008} has improved things considerably.
\item
The IEEE~754 standard does provide some strong guarantees, e.g.,
that the results of individual additions, subtractions, multiplications,
divisions and square roots are correctly rounded, that is, it is \emph{as if}
the results were computed in the reals and then rounded as per
the rounding mode in effect.
But it does not provide guarantees on the results of other operations
and on other aspects, such as, e.g., when the underflow exception is
signaled~\cite{CuytKVV02}.\footnote{The indeterminacy described in
\cite{CuytKVV02} is present also in the 2008 edition of IEEE~754
\cite{IEEE-754-2008}.}
\item
A pragmatic, yet effective approach to support formal reasoning
on commonly used implementation of mathematical functions has been
recently proposed in \cite{BagnaraCGB21}.
The proposed techniques exploit the fact that the floating-point
implementation of mathematical functions preserve, not completely but
to a great extent, the piecewise monotonicity nature of the
approximated functions over the reals.
\item
A static analysis for detecting floating-point exceptions
based on abstract interpretation has been presented in \cite{Mine04}.
A few attempts at this task have been made using other techniques
\cite{BarrVLS13,WuLZ17} but, as we argue in Sections~\ref{sec:related-work}
and \ref{sec:experimental-evaluation}, they present precision
and soundness issues.
\item
Most verification approaches in the literature assume the round-to-nearest
rounding mode \cite{BotellaGM06}, or over-approximate by always considering
worst-case rounding modes \cite{Mine04}.
Analyses based on SMT solvers \cite{BrainTRW15} can treat each rounding mode
precisely, but only if the rounding mode in use is known exactly.
As we show in Section~\ref{sec:experimental-evaluation}, some SMT solvers
also suffer from soundness issues.
\end{enumerate}

The contribution of this paper is in areas~\ref{item:issue-anomalies}
and \ref{item:issue-rounding-modes}.
In particular, concerning point~\ref{item:issue-anomalies}, by defining
and formally proving the correctness of constraint propagation algorithms
for IEEE~754 arithmetic constraints, we enable the use of formal methods
for a broad range of programs.  Such methods, i.e., abstract interpretation
and symbolic model checking, allow for proving that a number of generally
unwanted phenomena (e.g., generation of NaNs and infinities,
absorption, cancellation, instability, etc.) do not happen or, in case
they do happen, allow the generation of a test vector to reproduce
the issue.
Regarding point~\ref{item:issue-rounding-modes},
handling of all IEEE~754 rounding modes, and being resilient to uncertainty
about the rounding mode in effect, is another original contribution of this
paper.

While the round-to-nearest rounding mode is,
by far, the most frequently used one, it must be taken into account that:
\begin{itemize}
\item
the possibility of programmatically changing the rounding mode is granted
by IEEE~754 and is offered by most of its implementations (e.g., in the
C~programming language, via the \verb+fesetround()+ standard function);
\item
such possibility is exploited by interval libraries and by numerical calculus
algorithms (see, e.g., \cite{Rump13,RumpO07});
\item
setting the rounding mode to something different from round-to-nearest
can be done by third parties in a way that was not anticipated by
programmers: this may cause unwanted non-determinism in video games
\cite{Fiedler10} and there is nothing preventing the abuse of this
feature for more malicious ends, denial-of-service being only
the least dangerous in the range of possibilities.
Leaving malware aside, there are graphic and printer drivers
and sound libraries that are known to change the rounding mode and
may fail to set it back \cite{Watte08}.
\end{itemize}

As a possible way of tackling the difficulties described until now,
and enabling sound formal verification of floating-point computations,
this paper introduces new algorithms for the propagation
of arithmetic constraints over floating-point numbers.
These algorithms are called \emph{filtering algorithms} as their purpose
is to prune the domains of possible variable values by \emph{filtering out}
those values that cannot be part of the solution of a system of constraints.
Algorithms of this kind must be employed in \emph{constraint solvers}
that are required in several different areas,
such as automated test-case generation, exception detection or
the detection of subnormal computations.
In this paper we propose fully detailed, provably correct filtering
algorithms for floating-point constraints.
Such algorithms handle all values, including symbolic values
(NaNs, infinities and signed zeros), and rounding modes defined by IEEE~754.
Note that filtering techniques used in solvers over the reals do not
preserve all solutions of constraints over floating-point numbers
\cite{MichelRL01,Michel02}, and therefore they cannot be used to prune
floating-point variable domains reliably.
This leads to the need of filtering algorithms such as those we hereby
introduce.

The choice of the IEEE~754 Standard for floating-point numbers
as the target representation for our algorithms is due to their
ubiquity in modern computing platforms.
Indeed, although some programming languages leave the floating-point
format as an implementation-defined aspect, all widely-used
hardware platforms ---e.g., x86 \cite{Intelx8664} and ARM \cite{Armv8}---
only implement the IEEE~754 Standard, while older formats are
considered legacy.

Before defining our filtering algorithms in a detailed and formal way,
we provide a more comprehensive context on the propagation of floating-point
constraints and its practical applications
(Sections~\ref{sec:from-programs-to-floating-point-constraints}
and~\ref{sec:constraint-propagation}),
and justify their use in formal program analysis and verification
(Section~\ref{sec:applications-to-program-analysis}).
We also give a more in-depth view of related work in Section~\ref{sec:related-work},
and clarify our contribution in Section~\ref{sec:contribution}.

\subsection{From Programs to Floating-Point Constraints}
\label{sec:from-programs-to-floating-point-constraints}

Independently from the application, program analysis starts with
parsing, the generation of an \emph{abstract syntax tree} and
the generation of various kinds of intermediate program representations.
An important intermediate representation is called
\emph{three-address code} (TAC).
In this representation, complex arithmetic expressions and assignments
are decomposed into sequences of assignment instructions of the form
\[
  \mathtt{result}
    \mathrel{:=}
      \mathtt{operand}_1 \,\mathbin{\mathtt{operator}} \,\mathtt{operand}_2.
\]
A further refinement is the computation of the
\emph{static single assignment form} (SSA) \cite{AhoLSU06} whereby,
labeling each assigned variable with a fresh name,
assignments can be considered as if they were equality constraints.
For example, the TAC form of the floating-point assignment
$\mathtt{z} \mathrel{:=} \mathtt{z}*\mathtt{z} + \mathtt{z}$
is
\(
  \mathtt{t} \mathrel{:=} \mathtt{z}*\mathtt{z}; \;
  \mathtt{z} \mathrel{:=} \mathtt{t} + \mathtt{z}
\),
which in an SSA form becomes
\(
  \mathtt{t}_1 \mathrel{:=} \mathtt{z}_1*\mathtt{z}_1; \;
  \mathtt{z}_2 \mathrel{:=} \mathtt{t}_1 + \mathtt{z}_1
\).
These, in turn, can be regarded as the conjunction
of the constraints $t_1 = z_1 \mmul  z_1$ and $z_2 = t_1 \madd z_1$,
where by $\mmul$ and $\madd$ we denote the multiplication and
addition operations on floating-point numbers, respectively.
The Boolean comparison expressions that appear in the guards of if statements
and loops can be translated into constraints similarly.
This way, a C/\Cplusplus{} program translated into an SSA-based intermediate
representation can be represented as a set of constraints on its variables.
In particular, a constraint set arises form each execution path in the program.
For this reason, this approach to program modeling can be viewed
as \emph{symbolic execution} \cite{King76,ClarkeR85}.
Constraints can be added or removed from such a set in order to obtain
a constraint system that describes a particular behavior of the program
(e.g., the execution of a certain instruction, the occurrence of an overflow in
a computation, etc.). Once such a constraint system has been solved,
the variable domains only contain values that cause the desired behavior.
If one of the domains is empty, then that behavior can be ruled out.
For more details on the symbolic execution of floating-point computations,
we refer the reader to \cite{BotellaGM06,BagnaraCGG13ICST}.

\subsection{Constraint Propagation}
\label{sec:constraint-propagation}

Once constraints have been generated, they are amenable to
\emph{constraint propagation}:
under this name goes any technique that entails considering
a subset of the constraints at a time, explicitly removing
elements from the set of values that are candidate to be assigned
to the constrained variables. The values that can be removed are
those that cannot possibly participate in a solution for the selected
set of constraints.
For instance, if a set of floating-point constraints contains
the constraint
$x \mmul x = x$, then any value outside the set $\{ \nan, +0, 1, +\infty \}$
can be removed from further consideration.
The degree up to which this removal can actually take place depends on
the data-structure used to record the possible values for $x$,
intervals and multi-intervals being typical choices for numerical
constraints.
For the example above, if intervals are used, the removal can only be partial
(negative floating-point numbers are removed from the domain of $x$).
With multi-intervals more precision is possible, but any approach
based on multi-intervals must take measures to avoid combinatorial
explosion.

In this paper, we only focus on interval-based constraint propagation:
the algorithms we present for intervals can be rather easily
generalized to the case of multi-intervals.
We make the further assumption that the floating-point formats available
to the analyzed program are also available to the analyzer:
this is indeed quite common due to the wide adoption
of the IEEE~754 formats.

Interval-based floating-point constraint propagation consists of
iteratively narrowing the intervals associated to each variable:
this process is called \emph{filtering}.
A \emph{projection} is a function that,
given a constraint and the intervals associated to two of the variables
occurring in it, computes a possibly refined interval for the third
variable (the projection is said to be \emph{over} the third variable).
Taking $z_2 = t_1 \madd z_1$ as an example, the projection over $z_2$
is called \emph{direct projection} (it goes in the same sense of the
TAC assignment it comes from), while the projections over $t_1$ and $z_1$
are called \emph{indirect projections}.

\subsection{Applications of Constraint Propagation to Program Analysis}
\label{sec:applications-to-program-analysis}

When integrated in a complete program verification framework,
the constraint propagation techniques presented in this paper enable
activities such as abstract interpretation,
automatic test-input generation and symbolic model checking.
In particular, symbolic model checking means exhaustively proving
that a certain property, called specification, is satisfied by the system in exam,
which in this case is a computer program.
A model checker can either prove that the given specification is satisfied,
or provide a useful counterexample whenever it is not.

For programs involving floating-point computations,
some of the most significant properties that can be checked
consist of ruling out certain undesired exceptional behaviors
such as overflows, underflows and the generation of NaNs,
and numerical pitfalls such as absorption and cancellation.
In more detail, we call a \emph{numeric-to-NaN} transition
a floating-point arithmetic computation that returns a NaN
despite its operands being non-NaN.
We call a \emph{finite-to-infinite} transition
the event of a floating-point operation returning
an infinity when executed on finite operands,
which occurs if the operation overflows.
An \emph{underflow} occurs when the output of a computation is too small
to be represented in the machine floating-point format
without a significant loss in accuracy.
Specifically, we divide underflows into three categories,
depending on their severity:
\begin{description}
\item[Gradual underflow:]
an operation performed on normalized numbers results in a subnormal
number.  In other words, a subnormal has been generated out of normalized
numbers: enabling gradual underflow is indeed the very reason for the
existence of subnormals in IEEE~754.
However, as subnormals come with their share of problems, generating
them is better avoided.
\item[Hard underflow:]
an operation performed on normalized numbers results in a zero,
whereas the result computed on the reals is nonzero.
This is called \emph{hard} because the relative error is 100\%,
gradual overflow does not help (the output is zero, not a subnormal),
and, as neither input is a subnormal, this operation may constitute
a problem per se.
\item[Soft underflow:]
an operation with at least one subnormal operand results in a zero,
whereas the result computed on the reals is nonzero.
The relative error is still 100\% but, as one of the operands is a
subnormal, this operation may not be the root cause of the problem.
\end{description}
\emph{Absorption} occurs when the result of an arithmetic operation
is equal to one of the operands, even if the other one is not the neutral
element of that operation. For example, absorption occurs when summing
a number with another one that has a relatively very small exponent.
If the precision of the floating-point format in use is not enough to
represent them, the additional digits that would appear in the
mantissa of the result are rounded out.
\begin{definition} \summary{(Absorption.)}
Let $x, y, z \in \Fset$ with $y, z \in \Rset$,
let $\mathord{\mop}$ be any IEEE~754 floating-point operator,
and let $x = y \mop z$.
Then $y \mop z$ gives rise to \emph{absorption} if
\begin{itemize}
\item
$\mathord{\mop} = \mathord{\madd}$
and either $x = y$ and $z \neq 0$, or $x = z$ and $y \neq 0$;
\item
$\mathord{\mop} = \mathord{\msub}$
and either $x = y$ and $z \neq 0$, or $x = -z$ and $y \neq 0$;
\item
$\mathord{\mop} = \mathord{\mmul}$
and either $x = \pm y$ and $z \neq \pm 1$, or $x = \pm z$ and $y \neq \pm 1$;
\item
$\mathord{\mop} = \mathord{\mdiv}$,
$x = \pm y$ and $z \neq \pm 1$.
\end{itemize}
\end{definition}

In this section, we show how symbolic model checking
can be used to either rule out or pinpoint the presence of these run-time anomalies
in a software program by means of a simple but meaningful practical example.
Floating-point constraint propagation has been fully implemented with the techniques
presented in this paper in the commercial tool ECLAIR,%
\footnote{\url{https://bugseng.com/eclair}, last accessed on October~28th, 2021.}
developed and commercialized by BUGSENG.
ECLAIR is a generic platform for the formal verification of
C/\Cplusplus{} and Java source code, as well as Java bytecode.
The filtering algorithms described in the present paper are used in
the C/\Cplusplus{} modules of ECLAIR that are responsible for semantic
analysis based on abstract interpretation \cite{CousotC77}, automatic
generation of test-cases, and symbolic model checking. The latter two
are based on symbolic execution and constraint satisfaction problems
\cite{GotliebBR98,GotliebBR00}, whose solution is based on
multi-interval refinement and is driven by labeling and backtracking
search.
Indeed, the choice of ECLAIR as our target verification platform
is mainly due to its use of constraint propagation for solving
constraints generated by symbolic execution, which makes it easier
to integrate the algorithms presented in this paper.
However, such techniques are general, and could be used to solve
the constraints generated by any symbolic execution engine.

Constraints arising from the use of mathematical functions provided
by C/\Cplusplus{} standard libraries are also supported.
Unfortunately, most implementations of such libraries are not correctly rounded,
which makes the realization of filtering algorithms for them rather challenging.
In ECLAIR, propagation for such constraints is performed by exploiting
the piecewise monotonicity properties of those functions, which are partially
retained by all implementations we know of \cite{BagnaraCGB21}.

\begin{figure}
\centering
\begin{lstlisting}[mathescape,emph={exp,fabs},numbers=left,morekeywords={gsl_sf_result}]
int gsl_sf_bessel_i1_scaled_e(const double x, gsl_sf_result * result)
{
  double ax = fabs(x);

  /* CHECK_POINTER(result) */

  if(x == 0.0) {
    result->val = 0.0;
    result->err = 0.0;
    return GSL_SUCCESS;
  }
  else if(ax < 3.0*GSL_DBL_MIN) { |\label{line:bessel-if-underflow}|
    UNDERFLOW_ERROR(result);
  }
  else if(ax < 0.25) {
    const double eax = exp(-ax);
    const double y  = x*$^\mathrm{h}$x; |\label{line:bessel-hard1}|
    const double c1 = 1.0/10.0;
    const double c2 = 1.0/280.0;
    const double c3 = 1.0/15120.0;
    const double c4 = 1.0/1330560.0;
    const double c5 = 1.0/172972800.0;
    const double sum = 1.0 +$^\mathrm{a}$ y*$^\mathrm{sg}$(c1 +$^\mathrm{a}$ y*$^\mathrm{sg}$(c2 +$^\mathrm{a}$ y*$^\mathrm{sg}$(c3 +$^\mathrm{a}$ y*$^\mathrm{sg}$(c4 +$^\mathrm{a}$ y*$^\mathrm{sg}$c5))));
    result->val = eax * x/3.0 * sum;
    result->err = 2.0 * GSL_DBL_EPSILON * fabs(result->val);
    return GSL_SUCCESS;
  }
  else {
    double ex = exp$^\mathrm{hs}$(-2.0*$^\mathrm{i}$ax); |\label{line:bessel-minf}|
    result->val = 0.5 * (ax*(1.0+$^\mathrm{a}$ex) -$^\mathrm{a}$ (1.0-$^\mathrm{a}$ex)) /$^\mathrm{n}$ (ax*$^\mathrm{i}$ax); |\label{line:bessel-pinf-nan}|
    result->err = 2.0 * GSL_DBL_EPSILON * fabs(result->val);
    if(x < 0.0) result->val = -result->val;
    return GSL_SUCCESS;
  }
}
\end{lstlisting}
\caption{Function extracted from the GNU Scientific Library (GSL), version 2.5.
  The possible numerical exceptions detected by ECLAIR are marked
  by the raised letters next to the operators causing them.
  h, s and g stand for hard, soft and gradual underflow, respectively;
  a for absorption; i for finite-to-infinity; n for numeric-to-NaN.}
\label{fig:bessel}
\end{figure}

To demonstrate the capabilities of the techniques presented in this paper,
we applied them to the C code excerpt of Figure~\ref{fig:bessel}.
It is part of the implementation of the Bessel functions in the GNU Scientific Library,%
\footnote{\url{https://www.gnu.org/software/gsl/}, last accessed on October~28th, 2021.}
a widely adopted library for numerical computations.
In particular, it computes the scaled regular modified cylindrical Bessel function of first order,
$\exp(-|x|)I_1(x)$, where $x$ is a purely imaginary argument.
The function stores the computed result in the \texttt{val} field of the data structure
\texttt{result}, together with an estimate of the absolute error (\texttt{result->err}).
Additionally, the function returns an \texttt{int} status code,
which reports to the user the occurrence of certain exceptional conditions,
such as overflows and underflows.
In particular, this function only reports an underflow when the argument is smaller than a constant.
We analyzed this program fragment with ECLAIR's symbolic model checking engine,
setting it up to detect overflow (finite-to-infinite transitions),
underflow and absorption events, and NaN generation (numeric-to-NaN transitions).
Thus, we found out the underflow guarded against by the if statement
of line~\ref{line:bessel-if-underflow} is by far not the only numerical anomaly
affecting this function. In total, we found a numeric-to-NaN transition,
two possible finite-to-infinite transitions, two hard underflows, 5 gradual underflows
and 6 soft underflows. The code locations in which they occur are
all reported in Figure~\ref{fig:bessel}.

For each one of these events, ECLAIR yields an input value causing it.
Also, it optionally produces an instrumented version of the original code,
and runs it on every input it reports, checking whether it actually triggers
the expected behavior or not.
Hence, the produced input values are validated automatically.
For example, the hard underflow of line~\ref{line:bessel-hard1}
is triggered by the input $\mathtt{x} = \texttt{-0x1.8p-1021} \approx -6.6752 \times 10^{-308} $.
If the function is executed with $\mathtt{x} = \texttt{-0x1p+1023} \approx -8.9885 \times 10^{307}$,
the multiplication of line~\ref{line:bessel-minf} yields a negative infinity.
Since $\mathtt{ax} = |\mathtt{x}|$, we know $\mathtt{x} = \texttt{0x1p+1023}$
would also cause the overflow.
The same value of $\mathtt{x}$ causes an overflow in line~\ref{line:bessel-pinf-nan} as well.
The division in the same line produces a NaN if the function is executed
with $\mathtt{x} = -\infty$.

The context in which the events we found occur determines
whether they could cause significant issues.
For example, even in the event of absorption, the output of the overall computation
could be correctly rounded. Whether or not this is acceptable must be assessed depending
on the application. Indeed, the capability of ECLAIR of detecting absorption
can be a valuable tool to decide if a floating-point format with a higher precision is needed.
Nevertheless, some of such events are certainly problematic.
The structure of the function suggests that no underflow should occur if control flow
reaches past the if guard of line~\ref{line:bessel-if-underflow}.
On the contrary, several underflows may occur afterwards, some of which are even \emph{hard}.
Moreover, the generation of infinities or NaNs should certainly either be avoided,
or signaled by returning a suitable error code (and not \texttt{GSL\_SUCCESS}).
The input values reported by ECLAIR could be helpful for the developer in fixing
the problems detected in the function of Figure~\ref{fig:bessel}.
Furthermore, the algorithms presented in this paper are provably correct.
For this reason, it is possible to state that this code excerpt presents no other issues
besides those we reported above.
Notice, however, that due to the way the standard C mathematical library functions
are treated, the results above only hold with respect to the implementation of the $\mathtt{exp}$
function in use. In particular, the machine we used for the analysis is equipped
with the x86\_64 version of EGLIBC~2.19, running on Ubuntu 14.04.1.

\subsection{Related Work}
\label{sec:related-work}

\subsubsection{Filtering Algorithms}
In \cite{Michel02} C.\ Michel proposed a framework for filtering
constraints over floating-point numbers.
He considered monotonic functions over one argument and devised exact
direct and correct indirect projections for each possible rounding mode.
Extending this approach to binary arithmetic operators is not an easy task.
In \cite{BotellaGM06}, the authors extended the approach of \cite{Michel02}
by proposing filtering algorithms for the four basic binary arithmetic operators
when only the round-to-nearest tails-to-even rounding mode is available.
They also provided tables for indirect function projections
when zeros and infinities are considered with this rounding mode.
In our approach, we generalize the initial work of \cite{BotellaGM06}
by providing extended interval reasoning.
The algorithms and tables we present in this paper
consider all rounding modes, and contain all details and special cases,
allowing the interested reader to write an implementation of
interval-based filtering code.

Recently, \cite{Gallois-WongBC20} presented optimal inverse
projections for addition under the round-to-nearest rounding mode.
The proposed algorithms combine classical filtering based on the
properties of addition with filtering based on the properties of
subtraction constraints on floating-points as introduced by Marre and
Michel \cite{MarreM10}.
The authors are able to prove the optimality of the lower bounds
computed by their algorithms.
However, \cite{Gallois-WongBC20} only covers addition in the round-to-nearest
rounding mode, leaving other arithmetic operations (subtraction, multiplication and division)
and rounding modes to future work.
Special values (infinities and NaNs) are also not handled.
Conversely, this paper presents filtering algorithms covering all such cases.

It is worth noting that the filtering algorithms on intervals
presented in \cite{MarreM10} have been corrected for
addition/subtraction constraints and extended to multiplication and
division under the round-to-nearest rounding mode by some of these
authors (see \cite{BagnaraCGG13ICST,BagnaraCGG16IJOC}).
In this paper we discuss the cases in which the filtering algorithms
in \cite{BagnaraCGG13ICST,BagnaraCGG16IJOC,MarreM10} should be used in
combination with our filters for arithmetic constraints.
However, the main aim of this paper is to provide an exhaustive and
provably correct treatment of filtering algorithms supporting all
special cases for all arithmetic constraints under all rounding modes.

\subsubsection{SMT Solvers}
Satisfiability Modulo Theories (SMT) is the problem of deciding satisfiability of
first-order logic formulas containing terms from different, pre-defined theories.
Examples of such theories are integer or real arithmetic, bit-vectors,
arrays and uninterpreted functions.
Recently, SMT solvers have been widely employed as backends
for different software verification techniques, such as model checking
and symbolic execution \cite{BarrettT18}.
The need for verifying floating-point programs lead to the introduction
of a floating-point theory \cite{BrainTRW15} in SMT-LIB,
a library defining a common input language for SMT solvers.
Since then, the theory has been implemented in different ways into several solvers.
CVC4 \cite{BarrettCDHJKRT11,BrainSS19}, MathSAT \cite{CimattiGSS13} and Z3 \cite{MouraB08}
use \emph{bit-blasting}, i.e., they convert floating-point constraints to bit-vector formulae,
which are then solved as Boolean SAT problems.
Some tools, instead, use methods based on interval reasoning.
MathSAT also supports Abstract Conflict Driven Learning (ACDL)
for solving floating-point constraints based on interval domains \cite{BrainDGHK14}.
Colibri \cite{MarreBC17} uses constraint programming techniques,
with filtering algorithms such as those in \cite{BotellaGM06,BagnaraCGG16IJOC}
and those presented in this paper.
However, \cite{MarreBC17} does not report such filters in detail,
nor proves their correctness.
This leads to serious soundness issues, as we shall see in Section~\ref{sec:smt-solvers}.
An experimental comparison of such tools can be found in \cite{BrainSS19}.

Note that SMT-LIB, the input language of all such tools, only allows to specify
one single rounding mode for each floating-point operation.
Thus, the only way of dealing with uncertainty of the rounding mode in use
is to solve the same constraint system with all possible rounding mode combinations,
which is quite unpractical.
Our filtering algorithms are instead capable of working with a \emph{set} of possible
rounding modes, and retain soundness by always choosing the worst-case one.

\subsubsection{Floating-Point Program Verification}
Several program analyses for automatic detection of floating-point exceptions
were proposed in the literature.

Relational abstract domains for the analysis of floating-point
computations through abstract interpretation have been presented
in \cite{Mine04} and implemented in the tool Astr\'ee.%
\footnote{\url{https://www.absint.com/astree/index.htm}, last accessed on October 28th, 2021.}
Such domains, however, over-ap\-prox\-i\-mate rounding operations by always
assuming worst cases, namely rounding toward plus and minus infinity,
which may cause precision issues (i.e., false positives) if only
round-to-nearest is used.
Also, \cite{Mine04} does not offer a treatment of symbolic values
(e.g., infinities) as exhaustive as the one we offer in this paper.

In \cite{BarrVLS13} the authors proposed a symbolic execution system
for detecting floating-point exceptions.
It is based on the following steps: each numerical program is
transformed to directly check each exception-triggering condition, the
transformed program is symbolically-executed in real arithmetic to
find a (real) candidate input that triggers the exception, the real
candidate is converted into a floating-point number, which is finally
tested against the original program.
Since approximating floating-point arithmetic with real arithmetic
does not preserve the feasibility of execution paths and outputs in
any sense, they cannot guarantee that once a real candidate has been
selected, a floating-point number raising the same exception can be
found.
Even more importantly, even if the transformed program over the reals
is exception-free, the original program using floating-point
arithmetic may not be actually exception-free.

Symbolic execution is also the basis of the analysis proposed in
\cite{WuLZ17}, that aims at detecting floating-point exceptions by
combining it with value range analysis.
The value range of each variable is updated with the appropriate path
conditions by leveraging interval constraint-propagation techniques.
Since the projections used in that paper have not been proved to yield
correct approximations, it can be the case that the obtained value
ranges do not contain all possible floating-point values for each
variable.
Indeed, valid values may be removed from value ranges, which leads to
false negatives.
In Section~\ref{sec:discussion-and-conclusion},
the tool for floating-point exception
detection presented in \cite{WuLZ17} is compared with the same
analysis based on our propagation algorithms.
As expected, no false positives were detected among the results of our analysis.

\subsection{Contribution}
\label{sec:contribution}

This paper improves the state of the art in several directions:
\begin{enumerate}
\item
all rounding modes are treated and there is no assumption that the
rounding mode in effect is known and unchangeable (increased generality);
\item
utilization, to a large extent, of machine floating-point arithmetic
in the analyzer with few rounding mode changes (increased performance);
\item
accurate treatment of \emph{round half to even} ---the default rounding
mode of IEEE~754--- (increased precision);
\item
explicit and complete treatment of intervals containing symbolic values
(i.e., infinities and signed zeros);
\item
application of floating-point constraint propagation techniques
to enable detection of program anomalies such as overflows, underflows,
absorption, generation of NaNs.
\end{enumerate}

\subsection{Plan of the Paper}

The rest of the paper is structured as follows:
Section~\ref{sec:preliminaries} recalls the required notions
and introduces the notation used throughout the paper;
Section~\ref{sec:rounding-modes-and-rounding-errors} presents
some results on the treatment of uncertainty on the rounding
mode in effect and on the quantification of the rounding errors
committed in floating-point arithmetic operations;
Section~\ref{sec:propagation-arithmetic-constraints} contains the complete treatment
of addition and division constraints on intervals,
by showing detailed special values tables and the refinement algorithms;
Section~\ref{sec:experimental-evaluation} reports the results of experiments
aimed at evaluating the soundness of existing tools.
Section~\ref{sec:discussion-and-conclusion} concludes the main part of the paper.
Appendix~\ref{se:Subtraction-Multiplication} contains the complete
treatment of subtraction and multiplication constraints.
The proofs of results not reported in the main text of the paper
can be found in Appendix~\ref{proofs}.


\section{Preliminaries}
\label{sec:preliminaries}

We will denote by $\Rset_+$ and $\Rset_-$ the sets of strictly positive
and strictly negative real numbers, respectively.
The set of \emph{affinely extended reals},
$\Rset \union \{ -\infty, +\infty \}$, is denoted by $\extRset$.

\begin{definition} \summary{(IEEE~754 binary floating-point numbers.)}
\label{def:binary-floating-point-number}
A set of IEEE~754 binary floating-point numbers \textup{\cite{IEEE-754-2008}}
is uniquely identified by:
$p \in \Nset$, the number of significant digits (precision);
$\emax \in \Nset$, the maximum exponent, the minimum exponent being
$\emin \defeq 1 - \emax$.
The set of binary floating-point numbers
$\Fset(p, \emax, \emin)$ includes:
\begin{itemize}
\item
all signed zero and non-zero numbers of the form
$(-1)^s \cdot 2^e \cdot m$, where
  \begin{itemize}
  \item
   $s$ is the \emph{sign bit};
   \item
   the \emph{exponent} $e$ is any integer
    such that $\emin \leq e \leq \emax$;
   \item
   the \emph{mantissa} $m$, with $0 \leq m < 2$, is a number
   represented by a string of $p$ binary
   digits with a ``binary point'' after the first digit:
   \[
     m = (d_0 \;.\; d_1 d_2 \dots d_{p-1})_2 = \sum_{i = 0}^{p-1} d_i 2^{-i};
   \]
   \end{itemize}
\item
the \emph{infinities} $+\infty$ and $-\infty$;
the \emph{NaNs}: $\mathrm{qNaN}$ (\emph{quiet NaN})
and $\mathrm{sNaN}$ (\emph{signaling NaN}).
\end{itemize}
Numbers such that $d_0 = 1$ are called \emph{normal}.
The smallest positive normal floating-point number is
$\fnormin \defeq 2^{\emin}$
and the largest is $\fmax \defeq 2^{\emax}(2 - 2^{1-p})$.
The non-zero floating-point numbers such that $d_0 = 0$
are called \emph{subnormal}:
their absolute value is less than $2^{\emin}$, and
they always have fewer than $p$ significant digits.
Every finite floating-point number is an integral multiple of
the smallest subnormal magnitude
$\fmin \defeq 2^{\emin + 1 - p}$.
Note that the \emph{signed zeroes} $+0$ and $-0$ are distinct
floating-point numbers.
For a non-zero number $x$, we will write $\feven(x)$ (resp., $\fodd(x)$)
to signify that the least significant digit of $x$'s mantissa,
$d_{p-1}$, is~$0$ (resp.,~$1$).
\end{definition}

In the sequel we will only be concerned with IEEE~754 binary floating-point
numbers and we will write simply $\Fset$ for
$\Fset(p, \emax, \emin)$ when there is no risk of confusion.

\begin{definition} \summary{(Floating-point symbolic order.)}
\label{def:fp-symbolic-order}
Let $\Fset$ be any IEEE~754 floating-point format.
The relation $\reld{\slt}{\Fset}{\Fset}$ is such that,
for each $x, y \in \Fset$, $x \slt y$ if and only if
both $x$ and $y$ are not NaNs and either:
$x = -\infty$ and $y \neq -\infty$, or
$x \neq +\infty$ and $y = +\infty$, or
$x = -0$ and $y \in \{ +0 \} \union \Rset_+$, or
$x \in \Rset_- \union \{ -0 \}$ and $y = +0$, or
$x, y \in \Rset$ and $x < y$.
The partial order $\reld{\sleq}{\Fset}{\Fset}$ is such that,
for each $x, y \in \Fset$, $x \sleq y$ if and only if
both $x$ and $y$ are not NaNs and either $x \slt y$ or $x = y$.
\end{definition}
Note that $\Fset$ without the NaNs is linearly ordered with
respect to `$\slt$'.

For $x \in \Fset$ that is not a NaN, we will often abuse the notation
by interchangeably using the floating-point number or the extended
real number it represents.
The floats $-0$ and $+0$ both correspond to the real number $0$.
Thus, when we write, e.g., $x < y$ we mean that $x$ is numerically
less than $y$ (for example, we have $-0 \slt +0$ though $-0 \nless +0$,
but note that $x \sleq y$ implies $x \leq y$).
Numerical equivalence will be denoted by `$\equiv$' so that
$x \equiv 0$, $x \equiv +0$ and $x \equiv -0$ all denote
$(x = +0) \lor (x = -0)$.

\begin{definition} \summary{(Floating-point predecessors and successors.)}
\label{def:floating-point-predecessors-and-successors}
The partial function $\pard{\fsucc}{\Fset}{\Fset}$ is such that,
for each $x \in \Fset$,
\[
  \fsucc(x)
    \defeq
      \begin{cases}
        +\infty,
          &\text{if $x = \fmax$;} \\
        \min \{\, y \in \Fset \mid y > x \,\},
          &\text{if  $-\fmax \leq x < -\fmin $}\\
&\text{ or $\fmin \leq x < \fmax$;} \\
        \fmin,
          &\text{if $x \equiv 0$;} \\
        -0,
          &\text{if $x = -\fmin$;} \\
        -\fmax,
          &\text{if $x = -\infty$;} \\
        \text{undefined},
          &\text{otherwise.} \\
      \end{cases}
\]
The partial function $\pard{\fpred}{\Fset}{\Fset}$ is defined by
reversing the ordering, so that, for each $x \in \Fset$,
$\fpred(x) = -\fsucc(-x)$ whenever $\fsucc(x)$ is defined.
\end{definition}

Let $\circ \in \{ \mathord{+}, \mathord{-}, \mathord{\cdot}, \mathord{/} \}$
denote the usual arithmetic operations over the reals.
Let $R \defeq \{ \rdown, \rzero, \rup, \rnear \}$ denote
the set of IEEE~754 rounding modes (\emph{rounding-direction attributes}):
round towards minus infinity (\emph{roundTowardNegative}, $\rdown$),
round towards zero (\emph{round\-Toward\-Zero}, $\rzero$),
round towards plus infinity (\emph{roundTowardPositive}, $\rup$), and
round to nearest (\emph{round\-Ties\-To\-Even}, $\rnear$).
We will use the notation $\mathord{\mop_r}$, where
\(
  \mathord{\mop}
    \in
      \{ \mathord{\madd}, \mathord{\msub}, \mathord{\mmul}, \mathord{\mdiv} \}
\)
and $r \in R$,
to denote an IEEE~754 floating-point operation with rounding $r$.

The rounding functions are defined as follows.
Note that they are not defined for $0$: the IEEE~754 standard, in
fact, for operations whose exact result is $0$, bases the choice
between $+0$ and $-0$ on the operation itself and on the sign of the
arguments \cite[Section~6.3]{IEEE-754-2008}.

\begin{definition} \summary{(Rounding functions.)}
\label{def:rounding-functions}
The rounding functions defined by IEEE~754,
$\fund{\roundup{\cdot}}{\Rset \setdiff \{ 0 \}}{\Fset}$,
$\fund{\rounddown{\cdot}}{\Rset \setdiff \{ 0 \}}{\Fset}$,
$\fund{\roundzero{\cdot}}{\Rset \setdiff \{ 0 \}}{\Fset}$ and
$\fund{\roundnear{\cdot}}{\Rset \setdiff \{ 0 \}}{\Fset}$,
are such that, for each $x \in \Rset \setdiff \{ 0 \}$,
\begin{align}
\label{eq:round-towards-plus-infinity}
  \roundup{x}
    &\defeq
      \begin{cases}
        +\infty,
          &\text{if $x > \fmax$;} \\
        \min \{\, z \in \Fset \mid z \geq x \,\},
          &\text{if $x \leq -\fmin$ or $0 < x \leq \fmax$;} \\
        -0,
          &\text{if $-\fmin < x <0 $;}
      \end{cases} \\
\label{eq:round-towards-minus-infinity}
  \rounddown{x}
    &\defeq
      \begin{cases}
        \max \{\, z \in \Fset \mid z \leq x \,\},
          &\text{if $-\fmax \leq x < 0$ or $\fmin \leq x$;} \\
          +0,
          &\text{if $0 < x < \fmin $;} \\
          -\infty,
          &\text{if $x < -\fmax$;}
      \end{cases} \\
\label{eq:round-towards-zero}
  \roundzero{x}
    &\defeq
      \begin{cases}
        \rounddown{x},
          &\text{if $x > 0$;} \\
        \roundup{x},
          &\text{if $x < 0$;}
      \end{cases} \\
\label{eq:round-to-nearest}
  \roundnear{x}
    &\defeq
      \begin{cases}
        \rounddown{x},
          &\text{if $ -\fmax\leq x\leq  \fmax$ and either} \\
          &\text{ $\bigl|\rounddown{x} - x\bigr|
                 < \bigl|\roundup{x}   - x\bigr|$ or} \\
          &\text{ $\bigl|\rounddown{x} - x\bigr|
                 = \bigl|\roundup{x}   - x\bigr|$
                 and $\feven\bigl(\rounddown{x}\bigr)$;} \\
        \rounddown{x},
          &\text{if $\fmax < x <  2^{\emax}(2 - 2^{-p})$ 
                 or $x \leq -2^{\emax}(2 - 2^{-p})$;} \\
        \roundup{x},
          &\text{otherwise.}
      \end{cases}
\end{align}
\end{definition}
Note that, when the result of an operation has magnitude lower than $\fnormin$,
it is rounded to a subnormal number, by adjusting it to the form
$(-1)^s \cdot 2^\emin \cdot m$, and truncating its mantissa $m$,
which now starts with at least one $0$, to the first $p$ digits.
This phenomenon is called \emph{gradual underflow}, and while it is preferred
to \emph{hard underflow}, which truncates a number to $0$,
it still may cause precision issues due to the reduced number of significant
digits of subnormal numbers.

The rounding modes $\rdown$ and $\rup$ are the most ``extreme'',
while $\rnear$ and $\rzero$ are always contained between them.
We formalize this observation as follows:
\begin{proposition} \summary{(Properties of rounding functions.)}
\label{prop:round-properties}
Let $x \in \Rset \setdiff \{ 0 \}$. Then
\begin{align}
\label{round-properties:first}
  \rounddown{x} \leq x \leq  \roundup{x}, \\
\label{round-properties:second}
  \rounddown{x} \leq \roundzero{x} \leq \roundup{x}, \\
\label{round-properties:third}
  \rounddown{x} \leq \roundnear{x} \leq \roundup{x}.
\end{align}
 Moreover,
\begin{equation}
\label{round-properties:fourth}
  \rounddown{x} = -\roundup{-x}.
\end{equation}
\end{proposition}

In this paper, we use intervals of floating-point numbers in
$\Fset$ that are not NaNs.
\begin{definition} \summary{(Floating-point intervals.)}
Let $\Fset$ be any IEEE~754 floating-point format.
The set $\cI_\Fset$ of floating-point intervals with boundaries in $\Fset$
is given by
\[
  \cI_\Fset
    \defeq
      \{ \emptyset \}
        \union
          \bigl\{\,
            [\ell, u]
          \bigm|
            \ell, u \in \Fset, \ell \sleq u
          \,\bigr\}.
\]
By $[\ell, u]$ we denote the set
$\{\, x \in \Fset \mid \ell \sleq x \sleq u \,\}$.
The set $\cI_\Fset$ is a bounded meet-semilattice with least element $\emptyset$,
greatest element $[-\infty, +\infty]$, and the meet operation, which is
induced by set-intersection, will be simply denoted by $\mathord{\inters}$.
\end{definition}

Floating-point intervals with boundaries in $\Fset$ allow to capture
the extended numbers in $\Fset$: NaNs should be tracked separately.


\section{Rounding Modes and Rounding Errors}
\label{sec:rounding-modes-and-rounding-errors}

The IEEE~754 standard for floating-point arithmetic
introduces different rounding operators, among which
the user can choose on compliant platforms.
The rounding mode in use affects the results of the
floating-point computations performed, and it must be therefore
taken into account during constraint propagation.
In this section, we present some abstractions aimed at
facilitating the treatment of rounding modes
in our constraint projection algorithms.

\subsection{Dealing with Uncertainty on the Rounding Mode in Effect}

Even if programs that change the rounding mode in effect are quite rare,
whenever this happens,  the rounding mode in effect at each program point
cannot be known precisely.  So, for a completely general treatment
of the problem, such as the one we are proposing, our choice is to consider
a \emph{set} of possible rounding modes.  To this aim, in this section we define
two auxilliary functions that, given a set of rounding modes possibly
in effect, select a worst-case rounding mode that ensures soundness
of interval propagation.
Soundness is guaranteed even if the rounding mode used in the
actual computation differs from the one selected,
as far as the former is contained in the set.
Of course, if a program never changes the rounding mode,
the set of possible rounding modes boils down to be a singleton.

The functions presented in the first definition select the rounding modes that
can be used to compute the lower (function $r_\ell$) and upper (function $r_u$) bounds
of an operation in case of direct projections.

\begin{definition} \summary{(Rounding mode selectors for direct projections.)}
\label{def:rounding-mode-selectors}
Let $\Fset$ be any IEEE~754 floating-point format and
$S \sseq R$ be a set of rounding modes.
Let also $y, z \in \Fset$
and
\(
  \mathord{\mop}
   \in
     \{ \mathord{\madd}, \mathord{\msub}, \mathord{\mmul}, \mathord{\mdiv} \}
\)
be such that either $\mathord{\mop} \neq \mathord{\mdiv}$ or $z \neq 0$.
Then
\begin{align*}
   r_\ell(S, y, \mathord{\mop}, z)
     &\defeq
       \begin{cases}
         \rdown, &\text{if $\rdown \in S$;} \\
         \rdown, &\text{if $\rzero \in S$
                        and $y \circ z > 0$;} \\
         \rnear, &\text{if $\rnear \in S$;} \\
         \rup,   &\text{otherwise;}
      \end{cases} \\
   r_u(S, y, \mathord{\mop}, z)
     &\defeq
       \begin{cases}
         \rup,   &\text{if $\rup \in S$;} \\
         \rup,   &\text{if $\rzero \in S$
                        and $y \circ z \leq 0$;} \\
         \rnear, &\text{if $\rnear \in S$;} \\
         \rdown, &\text{otherwise.}
      \end{cases} \\
\end{align*}
\end{definition}

The following functions select the rounding modes that will be used for
the lower (functions $\bar{r}_\ell^r$ and $\bar{r}_\ell^\ell$)
and upper (functions $\bar{r}_u^r$ and $\bar{r}_u^\ell$)
bounds of an operation when computing inverse projections.
Note that there are different functions depending on which one of the
two operands is being projected:
$\bar{r}_\ell^r$ and $\bar{r}_u^r$ for the right one,
$\bar{r}_\ell^\ell$ and $\bar{r}_u^\ell$ for the left one.

\begin{definition} \summary{(Rounding mode selectors for inverse projections.)}
\label{def:rounding-mode-selectors-inverse}
Let $\Fset$ be any IEEE~754 floating-point format and
$S \sseq R$ be a set of rounding modes.
Let also $a, b \in \Fset$
and
\(
  \mathord{\mop}
   \in
     \{ \mathord{\madd}, \mathord{\msub}, \mathord{\mmul}, \mathord{\mdiv} \}.
\)
First, we define
\begin{align*}
   \hat{r}_\ell(S, \mop, b)
     &\defeq
       \begin{cases}
         \rup,   &\text{if $\rup \in S$;} \\
         \rup,   &\text{if $\rzero \in S$
                        and $b \sleq -0$,
                        or $b = +0$ and
                        $\mathord{\mop} \in \{\mathord{\madd}, \mathord{\msub}\}$;} \\
         \rnear, &\text{if $\rnear \in S$;} \\
         \rdown, &\text{otherwise;}
      \end{cases} \\
   \hat{r}_u(S, b)
     &\defeq
       \begin{cases}
         \rdown, &\text{if $\rdown \in S$;} \\
         \rdown, &\text{if $\rzero \in S$
                        and $b \sgeq +0$;} \\
         \rnear, &\text{if $\rnear \in S$;} \\
         \rup,   &\text{otherwise.}
       \end{cases}
\end{align*}
Secondly, we define the following selectors:
\begin{align*}
   \bigl(
    \bar{r}_\ell^\ell(S, b, \mop, a),
    \bar{r}_u^\ell(S, b, \mop, a)
   \bigr)
   &\defeq
     \begin{cases}
       \bigl(
        \hat{r}_\ell(S, \mop, b),
        \hat{r}_u(S, b)
       \bigr),
       &\text{if $\mop \in \{\madd, \msub\}$} \\
       &\text{ or $\mop \in \{\mmul, \mdiv\} \land a \sgeq +0$;} \\
       \bigl(
        \hat{r}_u(S, b),
        \hat{r}_\ell(S, \mop, b)
       \bigr),
       &\text{if \(
                   \mop \in \{\mmul, \mdiv\}
                   \land a \sleq -0
                 \);} \\
     \end{cases} \\
   \bigl(
    \bar{r}_\ell^r(S, b, \mop, a),
    \bar{r}_u^r(S, b, \mop, a)
   \bigr)
   &\defeq
     \begin{cases}
       \bigl(
        \hat{r}_\ell(S, \mop, b),
        \hat{r}_u(S, b)
       \bigr),
       &\text{if $\mop = \madd$,} \\
       &\text{ or $\mop = \mmul \land a \sgeq +0$,} \\
       &\text{ or $\mop = \mdiv \land a \sleq -0$;} \\
       \bigl(
        \hat{r}_u(S, b),
        \hat{r}_\ell(S, \mop, b)
       \bigr),
       &\text{if $\mop = \msub$,} \\
       &\text{ or $\mop = \mmul \land a \sleq -0$,} \\
       &\text{ or $\mop = \mdiv \land a \sgeq +0$.}
     \end{cases}
\end{align*}
\end{definition}

The usefulness in interval propagation of the functions presented above
will be clearer after considering Proposition~\ref{prop:worst-case-rounding-modes}.
Moreover, it is worth noting that, if the set of possible rounding modes
is composed by a unique rounding mode,
then all the previously defined functions return such rounding mode itself.
In that case, the claims of Proposition~\ref{prop:worst-case-rounding-modes}
trivially hold.

\begin{proposition}
\label{prop:worst-case-rounding-modes}
Let $\Fset$, $S$, $y$, $z$ and `$\mathord{\mop}$' be
as in \textup{Definition~\ref{def:rounding-mode-selectors}}.
Let also
$r_\ell = r_\ell(S, y, \mathord{\mop}, z)$
and
$r_u = r_u(S, y, \mathord{\mop}, z)$.
Then, for each $r \in S$
\begin{equation}
\label{eq:worst-case-rm-direct-between}
  y \mop_\mathrm{r_\ell} z
    \sleq
  y \mop_\mathrm{r} z
    \sleq
  y \mop_\mathrm{r_u} z.
\end{equation}
Moreover, there exist $r', r'' \in S$ such that
\begin{equation}
\label{eq:worst-case-rm-direct-equal}
  y \mop_\mathrm{r_\ell} z = y \mop_\mathrm{r'} z
  \quad\text{and}\quad
  y \mop_\mathrm{r_u} z = y \mop_\mathrm{r''} z.
\end{equation}

Now, consider $x = y \mop_\mathrm{r} z$
with $x, z \in \Fset$ and $r \in S$.
Let
$\bar{r}_\ell = \bar{r}_\ell^\ell(S, x_u, \mop, z)$
and
$\bar{r}_u = \bar{r}_u^\ell(S, x_\ell, \mop, z)$,
according to
\textup{Definition~\ref{def:rounding-mode-selectors-inverse}}.
Moreover, let
$\hat{y}'$
be the minimum $y' \in \Fset$ such that
$x = y' \mop_\mathrm{\bar{r}_\ell} z$,
and let
$\tilde{y}''$
be the maximum $y'' \in \Fset$ such that
$x = y'' \mop_\mathrm{\bar{r}_u} z$.
Then, the following inequalities hold:
\[
  \hat{y}' \sleq y \sleq \tilde{y}''.
\]
The same result holds if
$x = z \mop_\mathrm{r} y$,
with
$\bar{r}_\ell = \bar{r}_\ell^r(S, x_u, \mop, z)$
and
$\bar{r}_u = \bar{r}_u^r(S, x_\ell, \mop, z)$.
\end{proposition}
\begin{proof}
Here we only prove the claims for direct projections
(namely, \eqref{eq:worst-case-rm-direct-between} and
\eqref{eq:worst-case-rm-direct-equal}), leaving those concerning indirect projections,
which are analogous, to Appendix~\ref{sec:proofs-rounding-modes-and-rounding-errors}.

First, we observe that, for each $x, y, z \in \Fset$,
we have $\roundnear{y \circ z} = \rounddown{y \circ z}$,
or $\roundnear{y \circ z} = \roundup{y \circ z}$ or both.
Then we prove that, for each $x, y, z \in \Fset$, we have
$y \mop_\rdown z \sleq y \mop_\rnear z \sleq y \mop_\rup z$.
We distinguish between the following cases, depending on $y \circ z$:
\begin{description}

\item[$y \circ z = +\infty \;\lor\; y \circ z = -\infty:$]
in this case we have $y \mop_\rdown z = y \mop_\rnear z = y \mop_\rup z$
and thus $y \mop_\rdown z\sleq y \mop_\rnear z \sleq y \mop_\rup z$
holds.

\item[$y \circ z\leq -\fmin \;\lor\; y \circ z\geq \fmin:$]
in this case we have, by Proposition~\ref{prop:round-properties},
\(
       y \mop_\rdown z
     = \rounddown{y \circ z}
  \leq \roundnear{y \circ z}
     = y \mop_\rnear z
  \leq \roundup{y \circ z}
     = y \mop_\rup z
\);
as $y \mop_\rdown z \neq 0$, $y \mop_\rnear z \neq 0$ and $y \mop_\rup z \neq 0$,
the numerical order is reflected into the symbolic order to give
$y \mop_\rdown z \sleq y \mop_\rnear z \sleq y \mop_\rup z$.

\item[$-\fmin < y \circ z < 0:$]
in this case we have
$y \mop_\rdown z = -\fmin \leq y \mop_\rnear z \leq y \mop_\rup z =-0$
by Definition~\ref{def:rounding-functions};
since either $\roundnear{y \circ z} = -\fmin$
or $\roundnear{y \circ z} = -0$, we have$\roundnear{y \circ z} \neq +0$,
thus $y \mop_\rdown z \sleq y \mop_\rnear z \sleq y \mop_\rup z$.

\item[$0 < y \circ z<\fmin:$]
in this case we have
$y \mop_\rdown z = +0 \leq y \mop_\rnear z \leq y \mop_\rup z = \fmin$
by Definition~\ref{def:rounding-functions};
again, since either $\roundnear{y \circ z} = +0$
or $\roundnear{y \circ z} = \fmin$ we know that
$\roundnear{y \circ z} \neq -0$, and thus
$y \mop_\rdown z\sleq y \mop_\rnear z \sleq y \mop_\rup z$.

\item[$y \circ z = 0:$]
in this case, for multiplication and division the result is the same
for all rounding modes, i.e., $+0$ or $-0$ depending on the sign of the
arguments \cite[Section~6.3]{IEEE-754-2008};
for addition or subtraction we have
$y \mop_\rdown z \neq -0$ while $y \mop_\rnear z = y \mop_\rup z = +0$;
hence, also in this case,
$y \mop_\rdown z \sleq y \mop_\rnear z \sleq y \mop_\rup z$ holds.
\end{description}

Note now that, by Definition~\ref{def:rounding-functions},
if $y \circ z > 0$ then $y \mop_\rzero z = y \mop_\rdown z$ whereas,
if $y \circ z > 0$, then $y \mop_\rzero z = y \mop_\rup z$.
Therefore we can conclude that:
\begin{itemize}
\item if $y \circ z > 0$, then
\(
               y \mop_\rdown z
             = y \mop_\rzero z
  \sleq y \mop_\rnear z
  \sleq y \mop_\rup z
\)
while,
\item if $y \circ z < 0$, then
\(
               y \mop_\rdown z
  \sleq y \mop_\rnear z
  \sleq y \mop_\rzero z
             = y \mop_\rup z;
\)
moreover,
\item if $y \circ z = 0$ and $\circ \notin \{ \mathord{+}, \mathord{-} \}$, then
\(
    y \mop_\rdown z
  = y \mop_\rnear z
  = y \mop_\rzero z
  = y \mop_\rup z
\)
while,
\item if $y \circ z = 0$ and $\circ \in \{ \mathord{+}, \mathord{-} \}$, then
\(
               y \mop_\rdown z
  \sleq y \mop_\rnear z
             = y \mop_\rzero z
             = y \mop_\rup z.
\)
\end{itemize}
In order to prove inequality~\eqref{eq:worst-case-rm-direct-between},
it is now sufficient to consider all possible sets $S \sseq R$
and use the relations above.

For claim~\eqref{eq:worst-case-rm-direct-equal},
observe that $r_\ell(S, y, \mathord{\mop}, z) \in S$ for any combination
of rounding modes in $S$ except for one case:
that is when $y \circ z > 0$, and $\rzero \in S$ but $\rdown \notin S$.
In this case, however, by Definition~\ref{def:rounding-functions},
$y \mop_\rdown z = y \mop_\rnear z$.
Similarly, $r_u(S, y, \mathord{\mop}, z)\in S$ except for the case when
$y \circ z \leq 0$, $\rzero \in S$ but $\rup \notin S$.
First, assume that $y \circ z < 0$: in this case, by Definition~\ref{def:rounding-functions},
$y \mop_\rup z = y \mop_\rnear z$.
For the remaining case, that is $y \circ z = 0$, we observe that
for multiplication and division the result is the same
for all rounding modes \cite[Section~6.3]{IEEE-754-2008}, while
for addition or subtraction we have
$y \mop_\rzero z =  y \mop_\rnear z = y \mop_\rup z = +0$.
\end{proof}

Thanks to Proposition~\ref{prop:worst-case-rounding-modes} we need not
be concerned with sets of rounding modes, as any such set $S \sseq R$
can always be mapped to a pair of ``worst-case rounding modes'' which,
in addition, are never round-to-zero.
Therefore, projection functions can act as if the only possible rounding mode
in effect was the one returned by the selection functions,
greatly simplifying their logic.
For example, consider the constraint $x = y \mop_S z$, meaning
``$x$ is obtained as the result of $y \mop_r z$
for some $r \in S$.''
Of course, $x = y \mop_S z$ implies $x \sgeq y \mop_S z$
and $x \sleq y \mop_S z$, which,
by Proposition~\ref{prop:worst-case-rounding-modes},
imply $x \sgeq y \mop_{r_\ell} z$
and $x \sleq y \mop_{r_u} z$,
where
$r_\ell \defeq r_\ell(S, y, \mathord{\mop}, z)$
and
$r_u \defeq r_u(S, y, \mathord{\mop}, z)$.
The results obtained by projection functions that only consider
$r_\ell$ and $r_u$ are consequently valid for any $r \in S$.

\subsection{Rounding Errors}

For the precise treatment of all rounding modes it is useful to introduce
a notation that expresses, for each floating-point number $x$, the maximum
error that has been committed by approximating with $x$ a real number
under the different rounding modes (as shown in the previous section,
we need not be concerned with round-to-zero).

\begin{definition} \summary{(Rounding Error Functions.)}
\label{def:rounding-error-functions}
The partial functions
$\pard{\ferrup}{\Fset}{\extRset}$,
$\pard{\ferrdown}{\Fset}{\extRset}$,
$\pard{\ftwiceerrnearneg}{\Fset}{\extRset}$ and
$\pard{\ftwiceerrnearpos}{\Fset}{\extRset}$
are defined as follows,
for each $x \in \Fset$ that is not a NaN:
\begin{align}
  \ferrdown(x)
      &=
        \begin{cases}
          \text{undefined},
            &\text{if $x = +\infty$;} \\
          \fsucc(x) - x,
            &\text{otherwise;}
        \end{cases} \\
    \ferrup(x)
      &=
        \begin{cases}
          \text{undefined},
            &\text{if $x = -\infty$;} \\
          \fpred(x) - x,
            &\text{otherwise;}
        \end{cases} \\
  \ftwiceerrnearneg(x)
  \label{eq:def-ftwiceerrnearneg}
      &=
        \begin{cases}
          +\infty
            &\text{if $x = -\infty$;} \\
          x - \fsucc(x),
            &\text{if $x = -\fmax$;} \\
          \fpred(x) - x,
            &\text{otherwise;}
        \end{cases} \\
  \ftwiceerrnearpos(x)
      &=
        \begin{cases}
          -\infty,
            &\text{if $x = +\infty$;} \\
          x - \fpred(x),
            &\text{if $x = \fmax$;} \\
          \fsucc(x) - x,
            &\text{otherwise.}
        \end{cases}
\end{align}
\end{definition}

An interesting observation is that the values of the functions
introduced in Definition~\ref{def:rounding-error-functions} are
always representable in $\Fset$ and thus their computation
does not require extra-precision, something that, as we shall see,
is exploited in the implementation.  This is the reason why,
for round-to-nearest, $\ftwiceerrnearneg$ and $\ftwiceerrnearpos$
have been defined as \emph{twice} the approximation error bounds:
the absolute value of the bounds themselves, being $\fmin/2$,
is not representable in $\Fset$ for each $x \in \Fset$ such
that $|x| \leq \fnormin$.

When the round-to-nearest rounding mode is in effect,
Proposition~\ref{prop:min-max-x-ftwiceerrnearneg-ftwiceerrnearpos}
relates the bounds of a floating-point interval $[x_\ell,x_u]$
with those of the corresponding interval of $\extRset$ it represents.

\begin{proposition}
\label{prop:min-max-x-ftwiceerrnearneg-ftwiceerrnearpos}
Let $x_\ell, x_u \in \Fset \inters \Rset$.
Then
\begin{align}
\label{eq:min-x-plus-ftwiceerrnearneg}
  \min_{x_\ell \leq x \leq x_u} \bigl(x + \ftwiceerrnearneg(x)/2\bigr)
    = x_\ell + \ftwiceerrnearneg(x_\ell)/2, \\
\label{eq:max-x-plus-ftwiceerrnearpos}
  \max_{x_\ell \leq x \leq x_u} \bigl(x + \ftwiceerrnearpos(x)/2\bigr)
    = x_u + \ftwiceerrnearpos(x_u)/2.
\end{align}
\end{proposition}
\begin{proof}[sketch]
To prove \eqref{eq:min-x-plus-ftwiceerrnearneg},
we separately consider the two cases defined by \eqref{eq:def-ftwiceerrnearneg}.

If $x_\ell = -\fmax$, we prove that
$x_\ell + \ftwiceerrnearneg(x_\ell) / 2 = -2^{\emax}(2 - 2^{-p})$,
while for all $x_\ell < x \leq x_u$ we have
$x + \ftwiceerrnearneg(x) / 2 = (x + \fpred(x)) / 2$.
By monotonicity of `$\fpred$', the minimum value of $(x + \fpred(x)) / 2$
occurs when $x = \fsucc(-\fmax)$, and
\[
  (\fsucc(-\fmax) - \fmax) / 2
  = -2^{\emax}(2 - 2^{-p})
  = x_\ell + \ftwiceerrnearneg(x_\ell) / 2,
\]
which proves \eqref{eq:min-x-plus-ftwiceerrnearneg} in this case.

If, instead, $x_\ell > -\fmax$, applying monotonicity of `$\fpred$' suffices.

The full proof is in Appendix~\ref{sec:proofs-rounding-modes-and-rounding-errors},
together with the one of \eqref{eq:max-x-plus-ftwiceerrnearpos}, which is symmetric.
\end{proof}

\subsection{Real Approximations of Floating-Point Constraints}
\label{se:RealApprox}

In this section we show how inequalities of the form
$x \sgeq y \mop_r z$ and $x \sleq y \mop_r z$,
with $r \in \{ \rdown, \rup, \rnear \}$ can be reflected
on the reals.
Indeed, it is possible to
algebraically manipulate constraints on the reals so as to
numerically bound the values of floating-point quantities.
The results of this and of the next section will be useful
in designing inverse projections.

\begin{proposition}
\label{prop:real-approx-of-fp-constraints}
The following implications hold, for each $x, y, z \in \Fset$
such that all the involved expressions do not evaluate to NaN,
for each floating-point operation
\(
  \mathord{\mop}
    \in
      \{ \mathord{\madd}, \mathord{\msub}, \mathord{\mmul}, \mathord{\mdiv} \}
\)
and the corresponding extended real operation
$\circ \in \{ \mathord{+}, \mathord{-}, \mathord{\cdot}, \mathord{/} \}$,
where the entailed inequalities are to be interpreted over $\extRset$:
\begin{align}
\label{eq:real-approx-of-fp-constraints:1}
  x \sleq y \mop_\rdown z
    &\implies
      x \leq y \circ z;
\intertext{%
moreover, if $x \neq -\infty$,}
      \label{eq:real-approx-of-fp-constraints:2}
  x \sleq y \mop_\rup z
    &\implies
      x + \ferrup(x) < y \circ z; \\
           \label{eq:real-approx-of-fp-constraints:3}
  x \sleq y \mop_\rnear z
    &\implies
      \begin{cases}
        x + \ftwiceerrnearneg(x)/2 \leq y \circ z,
          &\text{if $\feven(x)$ or $x = +\infty$;} \\
        x + \ftwiceerrnearneg(x)/2 < y \circ z
          &\text{if $\fodd(x)$;}
      \end{cases} \\
\intertext{%
conversely,
}
\label{eq:real-approx-of-fp-constraints:4}
  x \sgeq y \mop_\rdown z
    &\implies
      x + \ferrdown(x) > y \circ z; \\
\intertext{%
moreover, if $x \neq +\infty$,}
      \label{eq:real-approx-of-fp-constraints:5}
  x \sgeq y \mop_\rup z
    &\implies
      x \geq y \circ z; \\
      \label{eq:real-approx-of-fp-constraints:6}
  x \sgeq y \mop_\rnear z
    &\implies
      \begin{cases}
        x + \ftwiceerrnearpos(x)/2 \geq y \circ z,
          &\text{if $\feven(x)$ or $x = -\infty$;} \\
        x + \ftwiceerrnearpos(x)/2 > y \circ z,
          &\text{if $\fodd(x)$.}
      \end{cases}
\end{align}
\end{proposition}
The proof of Proposition~\ref{prop:real-approx-of-fp-constraints}
is carried out by applying the inequalities of
Proposition~\ref{prop:round-properties} to each rounded operation,
resulting in a quite long case analysis.
It can be found in Appendix~\ref{sec:proofs-rounding-modes-and-rounding-errors}.

\subsection{Floating-Point Approximations of Constraints on the Reals}

In this section, we show how possibly complex constraints involving
floating-point operations can be approximated directly
using floating-point computations, without necessarily using
infinite-precision arithmetic.

Without being too formal, let us consider the domain $E_\Fset$ of abstract
syntax trees with leafs labelled by constants in $\Fset$ and internal
nodes labeled with a symbol in
$\{ \mathord{+}, \mathord{-}, \mathord{\cdot}, \mathord{/} \}$
denoting an operation on the reals.
While developing propagation algorithms, it is often
necessary to deal with inequalities between real numbers
and expressions described by such syntax trees.
In order to successfully approximate them using the available
floating-point arithmetic, we need two functions:
$\fund{\evaldown{\cdot}}{E_\Fset}{\Fset}$ and
$\fund{\evalup{\cdot}}{E_\Fset}{\Fset}$.
These functions provide an abstraction
of evaluation algorithms that: (a) respect the indicated approximation
direction; and (b) are as precise as practical.
Point (a) can always be achieved by substituting the real operations
with the corresponding floating-point operations rounded in the right
direction.  For point (b), maximum precision can trivially be achieved
whenever the expression involves only one operation; generally speaking,
the possibility of efficiently computing a maximally precise
(i.e., correctly rounded) result
depends on the form of the expression (see, e.g., \cite{KornerupLLM09}).

\begin{definition} \summary{(Evaluation functions.)}
\label{evaluation functions}
The two partial functions $\pard{\evaldown{\cdot}}{E_\Fset}{\Fset}$
and $\pard{\evalup{\cdot}}{E_\Fset}{\Fset}$ are such that,
for each $e \in \Fset$ that evaluates on $\extRset$ to a nonzero value,
\begin{align}
  \evaldown{e} &\sleq \rounddown{e}, \\
  \evalup{e}   &\sgeq \roundup{e}.
\end{align}
\end{definition}

\begin{proposition}
\label{prop:fp-approx-of-real-constraints}
Let $x \in \Fset$ be a non-NaN floating point number and
$e \in E_\Fset$ an expression that evaluates on $\extRset$ to a nonzero value.
The following implications hold:
\begin{align}
\label{prop:fp-approx-of-real-constraints:1}
  x \geq e
    &\implies
      x \sgeq \evaldown{e}; \\
\label{prop:fp-approx-of-real-constraints:2}
  \text{if $\evaldown{e} \neq +\infty$, }
  x > e
    &\implies
      x \sgeq \fsucc\bigl(\evaldown{e}\bigr); \\
\label{prop:fp-approx-of-real-constraints:3}
  x \leq e
    &\implies
      x \sleq \evalup{e}; \\
\label{prop:fp-approx-of-real-constraints:4}
  \text{if $\evaldown{e} \neq -\infty$, }
  x < e
    &\implies
      x \sleq \fpred\bigl(\evalup{e}\bigr). \\
\intertext{%
In addition, if $\fpred\bigl(\evalup{e}\bigr) < e$
(or, equivalently, $\evalup{e} = \roundup{e}$)
we also have}
\label{case:x-geq-e-implies-x-sgeq-evalup-e}
  x \geq e
    &\implies
      x \sgeq \evalup{e};
\intertext{%
likewise, if $\fsucc\bigl(\evaldown{e}\bigr) > e$
(or, equivalently, $\evaldown{e} = \rounddown{e}$)
we have}
\label{case:x-leq-e-implies-x-sleq-evaldown-e}
  x \leq e
    &\implies
      x \sleq \evaldown{e}.
\end{align}
\end{proposition}
The implications of Proposition~\ref{prop:fp-approx-of-real-constraints}
can be derived from Definition~\ref{evaluation functions}
and Proposition~\ref{prop:round-properties}.
Their proof is postponed to Appendix~\ref{sec:proofs-rounding-modes-and-rounding-errors}.

\section{Propagation for Simple Arithmetic Constraints}
\label{sec:propagation-arithmetic-constraints}

In this section we present our propagation procedure for the solution of
floating-point constraints obtained from the analysis of programs
engaging in IEEE~754 computations.

The general propagation algorithm,
which we already introduced in
Section~\ref{sec:constraint-propagation},
consists in an iterative procedure that applies the
direct and inverse filtering algorithms
associated with each constraint,
narrowing down the intervals associated with each variable.
The process stops when fixed point is reached,
i.e., when a further application of any filtering
algorithm does not change the domain of any variable.

\subsection{Propagation Algorithms: Definitions}
\label{sec:propagation-algorithms-definitions}

Constraint propagation is a process that
prunes the domains of program variables
by deleting values that do not satisfy any
of the constraints involving those variables.
In this section, we will state these ideas more formally.

Let
\(
  \mathord{\mop}
   \in
     \{ \mathord{\madd}, \mathord{\msub}, \mathord{\mmul}, \mathord{\mdiv} \}
\)
and $S \sseq R$.
Consider a constraint
$x = y \mop_S z$ with
$x \in X = [x_\ell, x_u]$,
$y \in Y = [y_\ell, y_u]$ and
$z \in Z = [z_\ell, z_u]$.

\emph{Direct propagation} aims at inferring a narrower interval
for variable $x$, by considering the domains of $y$ and $z$.
It amounts to computing a possibly
refined interval for $x$,
$X' = [x'_\ell, x'_u] \sseq X$, such that
\begin{equation}
\label{eq:direct-propagation-correctness}
  \forall r \in S, x \in X, y \in Y, z \in Z
    \itc x = y \mop_r z \implies x \in X'.
\end{equation}
Property~\eqref{eq:direct-propagation-correctness}
is known as the \emph{direct propagation correctness property}.

Of course it is always possible to take $X' = X$,
but the objective of constraint propagation is to compute a ``small'',
possibly the smallest $X'$ enjoying~\eqref{eq:direct-propagation-correctness},
compatibly with the available information.
The smallest $X'$
that satisfies~\eqref{eq:direct-propagation-correctness}
is called optimal and is such that
\begin{equation}
\label{eq:direct-propagation-optimality}
  \forall X''\subset X'
    \itc \exists r \in S, y \in Y, z \in Z
      \st y \mop_r z \not\in X''.
\end{equation}
Property~\eqref{eq:direct-propagation-optimality} is called
the \emph{direct propagation optimality property}.

\emph{Inverse propagation}, on the other hand,
uses the domain of the result $x$ to deduct
new domains for the operands, $y$ or $z$.
For the same constraint,
$x = y \mop_S z$,
it means computing a possibly refined interval
for $y$,
$Y' = [y'_\ell, y'_u] \sseq Y$, such that
\begin{equation}
\label{eq:inverse-propagation-correctness}
 \forall r \in S, x \in X, y \in Y, z \in Z
    \itc x = y \mop_r z \implies y \in Y'.
\end{equation}
Property~\eqref{eq:inverse-propagation-correctness}
is known as the \emph{inverse propagation correctness property}.
Again, taking $Y' = Y$ is always possible and
sometimes unavoidable. The best we can hope for is to be able
to determine the smallest such set, i.e., satisfying
\begin{equation}
\label{eq:inverse-projection-optimality}
  \forall Y'' \sslt Y
    \itc \exists r \in S, y \in Y' \setdiff Y'', z \in Z
      \st y \mop_r z \not\in X.
\end{equation}
Property~\eqref{eq:inverse-projection-optimality}
is called the \emph{inverse propagation optimality property}.
Satisfying this last property can be very difficult.

\subsection{The Boolean Domain for NaN}
\label{sec:boolean-domain-nan}

From now on, we will consider
floating-point intervals with boundaries in $\Fset$.
They allow for capturing the extended numbers in $\Fset$ only:
NaNs (quiet NaNs and signaling NaNs) should be tracked separately.
To this purpose, a Boolean domain
$\cN \defeq \{ \top, \bot \}$, where $\top$ stands for ``may be NaN''
and $\bot$ means ``cannot be NaN'',
can be used and coupled with the arithmetic filtering algorithms.

Let be $x = y \mathbin{\mop} z$ an arithmetic constraint over floating-point numbers,
and $(X, \nan_x)$, $(Y, \nan_y)$ and
$(Z, \nan_z)$ be the variable domains of $x$, $y$ and $z$ respectively.
In practice, the propagation process for such a constraint
reaches a fixed point when the combination of refining domains
$(X', \nan_x')$, $(Y', \nan_y')$ and $(Z', \nan_z')$
remains the same obtained in the previous iteration.
For each iteration of the algorithm we analyze the NaN domain of all the
constraint variables in order to define the next propagator action.

The IEEE~754 Standard \cite[Section 7.2]{IEEE-754-2008} lists
all combinations of operand values that yield a NaN result.
For the arithmetic operations considered in this paper,
NaN is returned if any of the operands is NaN.
Moreover, addition and subtraction return NaN when infinities are subtracted
(e.g., $+\infty \madd -\infty$ or $+\infty \msub +\infty$),
and also $\pm \infty \mmul 0 = \nan$, $0 \mmul \pm \infty = \nan$,
$0 \mdiv 0 = \nan$, and $\pm \infty \mdiv \pm \infty = \nan$.

Thus, direct projections are such that if $\nan_y = \top$ or $\nan_z = \top$,
then also $\nan_x' = \top$; indirect projections yield $\nan_y' = \nan_z' = \top$ if
$\nan_x = \top$.
Moreover, e.g., if $\mop = \madd$, then the direct projection yields
$\nan_x' = \top$ also if $\pm \infty \in Y$ and $\mp \infty \in Z$,
and the indirect one allows for $\pm \infty$ in $Y'$ and $\mp \infty$ in $Z'$
only if $\nan_x = \top$, and so on for the other operators.

\subsection{Filtering Algorithms for Simple Arithmetic Constraints}
\label{sec:filtering-algorithms}

Filtering algorithms for arithmetic constraints
are the main focus of this paper.
In the next sections, we will propose algorithms realizing
optimal \emph{direct} projections and correct \emph{inverse} projections
for the addition ($\madd$) and division ($\mdiv$) operations.
The reader interested in implementing constraint propagation
for all four operations can find the algorithms and results
for the missing operations in Appendix~\ref{se:Subtraction-Multiplication}.
We report the correctness proofs of the projections for addition
in the main text, leaving those for the remaining operations to
Appendix~\ref{sec:proofs-filtering}.

The filtering algorithms we are about to present are capable of
dealing with any set of rounding modes and are designed to distinguish
between all different (special) cases in order to be as precise as
possible, especially when the variable domains contain symbolic
values.  Much simpler projections can be designed whenever precision
is not of particular concern. Indeed, the algorithms presented in this
paper can be considered as the basis for finding a good trade-off
between efficiency and the required precision.

\subsubsection{Addition}
Here we deal with constraints of the form
$x = y \madd_S z$ with $S \sseq R$.
Let
$X = [x_\ell, x_u]$,
$Y = [y_\ell, y_u]$ and
$Z = [z_\ell, z_u]$.

Thanks to Proposition~\ref{prop:worst-case-rounding-modes},
any set of rounding modes $S \sseq R$
can be mapped to a pair of ``worst-case rounding modes'' which,
in addition, are never round-to-zero.
Therefore, the projection algorithms use the selectors presented in
Definition~\ref{def:rounding-mode-selectors} to choose the appropriate
worst-case rounding mode, and then operate as if it was the only one in effect,
yielding results implicitly valid for the entire set $S$.

\paragraph{Direct Propagation.}

For direct propagation, i.e., the process that infers a new interval
for $x$ starting from the interval for $y$ and $z$, we propose
Algorithm~\ref{algo1} and functions $\diraddl$ and $\diraddu$,
as defined in Figure~\ref{fig:direct-projection-addition}.
Functions $\diraddl$ and $\diraddu$ yield new bounds for interval $X$.
In particular, function $\diraddl$ gives the new lower bound, while
function $\diraddu$ provides the new upper bound of the interval.
Functions $\diraddl$ and $\diraddu$ handle all rounding modes and,
in order to be as precise as possible, they distinguish between several cases,
depending on the values of the bounds of intervals $Y$ and $Z$.
These cases are infinities ($-\infty$ and $+\infty$),
zeroes ($-0$ and $+0$), negative values ($\Rset_-$) and positive values ($\Rset_+$).

\begin{algorithm}
\caption{Direct projection for addition constraints.}
\label{algo1}
\begin{algorithmic}[1]
\REQUIRE $x = y \madd_S z$,
$x \in X = [x_\ell, x_u]$,
$y \in Y = [y_\ell, y_u]$ and
$z \in Z = [z_\ell, z_u]$.
\ENSURE
$X' \sseq  X$ and
\(
  \forall r \in S, x \in X, y \in Y, z \in Z
    \itc  x = y \madd_r z \implies x \in X'
\)
and
\(
\forall X''\subset X', \exists r \in S, y \in Y, z \in Z \itc y  \madd_r z\not\in X''\).
\STATE
$r_\ell \assign r_\ell(S, y_\ell, \madd, z_\ell)$; $r_u \assign r_u(S, y_u, \madd, z_u)$;
\STATE
$x'_\ell \assign \diraddl(y_\ell, z_\ell, r_\ell)$;
$x'_u \assign \diraddu(y_u, z_u, r_u)$;
\STATE
$X' \assign X \inters [x'_\ell, x'_u]$;
\end{algorithmic}
\end{algorithm}

\begin{figure}[h]
\begin{tabular}{L|CCCCCC}
\diraddl(y_\ell, z_\ell, r_\ell)
        & -\infty & \Rset_-           & -0      & +0      & \Rset_+           & +\infty \\
\hline
-\infty & -\infty & -\infty           & -\infty & -\infty & -\infty           & +\infty \\
\Rset_- & -\infty & y_\ell \madd_{r_\ell} z_\ell & y_\ell     & y_\ell     & y_\ell \madd_{r_\ell} z_\ell & +\infty \\
-0      & -\infty & z_\ell                & -0     &    a_1  & z_\ell                & +\infty \\
+0      & -\infty & z_\ell                & a_1    & +0      & z_\ell                & +\infty \\
\Rset_+ & -\infty & y_\ell \madd_{r_\ell} z_\ell & y_\ell     & y_\ell     & y_\ell \madd_{r_\ell} z_\ell & +\infty \\
+\infty & +\infty & +\infty           & +\infty & +\infty & +\infty           & +\infty \\
\end{tabular}
\[
a_1
  =
    \begin{cases}
      -0, &\text{if $r_\ell = \rdown$,} \\
      +0, &\text{otherwise;}
    \end{cases}
\]

\bigskip
\begin{tabular}{L|CCCCCC}
\diraddu(y_u, z_u, r_u)
       & -\infty & \Rset_- & -0      & +0      & \Rset_+ & +\infty \\
\hline
-\infty & -\infty & -\infty & -\infty & -\infty & -\infty & -\infty \\
\Rset_- & -\infty     &  y_u \madd_{r_u} z_u   & y_u     & y_u     &  y_u \madd_{r_u} z_u    & +\infty  \\
-0      & -\infty    & z_u    & -0      & a_2    & z_u    & +\infty  \\
+0 & -\infty    & z_u    & a_2     & +0      & z_u    & +\infty  \\
\Rset_+ & -\infty    &  y_u \madd_{r_u} z_u  & y_u     & y_u     &  y_u \madd_{r_u} z_u      & +\infty  \\
+\infty &  -\infty    & +\infty & +\infty & +\infty &  +\infty   &  +\infty \\
\end{tabular}
\[
a_2
  =
    \begin{cases}
      -0, &\text{if $r_u = \rdown$,} \\
      +0, &\text{otherwise.}
    \end{cases}
\]
\caption{Direct projection of addition: the function $\diraddl$ (resp., $\diraddu$);
         values for $y_\ell$ (resp., $y_u$) on rows,
         values for $z_\ell$ (resp., $z_u$) on columns.}
\label{fig:direct-projection-addition}
\end{figure}

It can be proved that \textup{Algorithm~\ref{algo1}}
computes a \emph{correct} and \emph{optimal direct projection},
as stated by its postconditions.

\begin{theorem}
\label{teo:direct-projection-addition}
\textup{Algorithm~\ref{algo1}} satisfies its contract.
\end{theorem}
\begin{proof}
Given the constraint $x = y \madd_S z$ with
$x \in X = [x_\ell, x_u]$,
$y \in Y = [y_\ell, y_u]$ and
$z \in Z = [z_\ell, z_u]$,
Algorithm~\ref{algo1} sets
$X' = [x'_\ell, x'_u] \inters X$; hence, we have $X' \sseq X$.
Moreover, by Proposition~\ref{prop:worst-case-rounding-modes},
for each $y \in Y$, $z \in Z$ and $r \in S$,
we have $y \madd_{r_\ell} z \sleq y \madd_r z \sleq y \madd_{r_u} z$,
and because $a \sleq b$ implies $a \leq b$ for any $a, b \in \Fset$
according to Definition~\ref{def:fp-symbolic-order},
we know that $y \madd_{r_\ell} z \leq y \madd_r z \leq y \madd_{r_u} z$.
Thus, by monotonicity of $\madd$, we have
\(
  y_\ell \madd_{r_\ell} z_\ell
  \leq y \madd_{r_\ell} z
  \leq y \madd_r z
  \leq y \madd_{r_u} z
  \leq y_u \madd_{r_u} z_u
\).
Therefore, we can focus on finding a lower bound for $y_\ell \madd_{r_\ell} z_\ell$
and an upper bound for $y_u \madd_{r_u} z_u$.

Such bounds are given by the functions $\diraddl$ and $\diraddu$
of Figure~\ref{fig:direct-projection-addition}.
Almost all of the cases reported in the tables can be trivially derived
from the definition of the addition operation in the IEEE~754 Standard
\cite{IEEE-754-2008};
just two cases need further explanation.
Concerning the entry of $\diraddl$ in which $y_\ell = -\infty$ and
$z_\ell = +\infty$, note that $z_\ell = +\infty$ implies $z_u = +\infty$.
Then for any $y > y_\ell = -\infty$, $y \madd +\infty = +\infty$.
On the other hand, by the IEEE~754 Standard \cite{IEEE-754-2008},
$-\infty \madd +\infty$ is an invalid operation.
For the symmetric case, i.e., the entry of $\diraddu$
in which $y_u = -\infty$ and
$z_u = +\infty$, we can reason dually.

We are now left to prove that
$\forall X'' \subset X' \itc \exists r \in S, y \in Y, z \in Z \itc y \madd_r z \not\in X''$.
Let us focus on the lower bound $x'_\ell$,
proving that there always exists a $r \in S$ such that $y_\ell \madd_r z_\ell = x'_\ell$.

First, consider the cases in which
$y_\ell \not\in (\Rset_- \cup \Rset_+)$ or $z_\ell \not\in (\Rset_- \cup \Rset_+)$.
In these cases, a case analysis proves that
$\diraddl(y_\ell, z_\ell,r_\ell)$ is equal to $y_\ell \madd_{r_\ell} z_\ell$.
Indeed, if either of the operands (say $y_\ell$) is $-\infty$ and the other one
(say $z_\ell$) is \emph{not} $+\infty$, then according to the IEEE~754 Standard
we have $y_\ell \madd_r z_\ell = -\infty$ for any $r \in R$.
Symmetrically, $y_\ell \madd_r z_\ell = +\infty$
if one operand is $+\infty$ and the other one is not $-\infty$.
If, w.l.o.g., $y_\ell = +\infty$ and $z_\ell = -\infty$,
the set $X'$ is non-empty only if $z_u \neq -\infty$,
and $y_\ell \madd_r z_u = +\infty$ for any $r \in R$.

For the cases in which $y_\ell \in (\Rset_- \cup \Rset_+)$
and $z_\ell \in (\Rset_- \cup \Rset_+)$
we have $x'_\ell = y_\ell \madd_{r_\ell} z_\ell$,
by definition of $\diraddl$ of Figure~\ref{fig:direct-projection-addition}.
Remember that, by Proposition~\ref{prop:worst-case-rounding-modes},
there exists $r \in S$ such that
$y_\ell \madd_\mathrm{r_\ell} z_\ell = y_\ell \madd_\mathrm{r} z_\ell$.
Since $y_\ell \in Y$ and $z_\ell \in Z$, we can conclude that
for any $X'' \subseteq X'$,
$x'_\ell \not\in X''$ implies $y_\ell \madd_\mathrm{r} z_\ell \not\in X''$.

An analogous argument allows us to conclude that there exists an $r \in S$
for which the following holds:
for any $X'' \subseteq X'$,
$x'_u \not\in X''$ implies $y_u \madd_\mathrm{r} z_u \not\in X''$.
\end{proof}

The following example will better illustrate how the tables in
Figure~\ref{fig:direct-projection-addition}
should be used to compute functions $\diraddl$ and $\diraddu$.
All examples in this section refer to the IEEE~754 binary single precision format.
\begin{example}
  Assume $Y=[+0,5]$, $Z=[-0, 8]$, and that the selected rounding mode
  is $r_\ell = r_u = \rdown$.
  In order to compute the lower bound $x_\ell'$ of $X'$, the new interval
  for $x$, function $\diraddl(+0, -0, \rdown)$ is called.
  These arguments fall in case $a_1$, which yields $-0$ with rounding mode $\rdown$.
  Indeed, when the rounding mode is $\rdown$, the sum of $-0$ and $+0$
  is $-0$, which is clearly the lowest result that can be obtained with
  the current choice of $Y$ and $Z$.
  For the upper bound $x_u'$, the algorithm calls $\diraddu(5, 8, \rdown)$.
  This falls in the case in which both operands are positive numbers
  ($y_u, z_u \in \Rset_+$), and therefore
  $x_u = y_u \madd_{r_u} z_u = 13$.
  In conclusion, the new interval for $x$ is
  $X' = [-0, 13]$.

  If any other rounding mode was selected (say, $r_\ell = r_u = \rnear$),
  the new interval computed by the projection would have been
  $X'' = [+0, 13]$.
\end{example}

\paragraph{Inverse Propagation.}

For inverse propagation,
i.e., the process that infers a new interval for $y$ (or for $z$)
starting from the interval from $x$ and $z$ ($x$ and $y$, resp.)
we define Algorithm~\ref{algo2} and functions $\invaddl$ in
Figure~\ref{fig:indirect-projection-addition-yl} and $\invaddu$ in
Figure~\ref{fig:indirect-projection-addition-yu},
where $\equiv$ indicates the syntactic substitution of expressions.
Since the inverse operation of addition is subtraction,
note that the values of $x$ and $z$ that minimize $y$ are $x_\ell$ and $z_u$;
analogously, the values of $x$ and $z$ that maximize $y$ are $x_u$ and $z_\ell$.

When the round-to-nearest rounding mode is in effect,
addition presents some nice properties.
Indeed, several expressions for lower and upper bounds can be
easily computed without approximations, using floating-point operations.
In more detail, it can be shown (see the proof of Theorem~\ref{teo:indirect-projection-addition})
that when $x$ is subnormal $\ftwiceerrnearpos(x)$ and $\ftwiceerrnearneg(x)$ are negligible.
This allows us to define tight bounds in this case.
On the contrary, when the terms $\ftwiceerrnearneg(x_\ell)$ and $\ftwiceerrnearpos(x_u)$
are non negligible, we need to approximate the values of expressions
$e_\ell$ and
$e_u$.
This can always be done with reasonable efficiency \cite{KornerupLLM09},
but we leave this as an implementation choice, thus accounting for the case
when the computation is exact
($\evalup{e_\ell} = \roundup{e_\ell}$ and $\evaldown{e_u} = \rounddown{e_u}$)
as well as when it is not
($\evalup{e_\ell} > \roundup{e_\ell}$ and $\evaldown{e_u} < \rounddown{e_u}$).

\begin{algorithm}
\caption{Inverse projection for addition constraints.}
\label{algo2}
\begin{algorithmic}[1]
\REQUIRE $x = y \madd_S z$,
$x\in X = [x_\ell, x_u]$,
$y\in Y = [y_\ell, y_u]$ and
$z\in Z = [z_\ell, z_u]$.
\ENSURE
$Y' \sseq  Y$ and
\(
  \forall r \in S, x \in X, y \in Y, z \in Z
    \itc  x = y \madd_r z \implies y \in Y'
\).

\STATE
  $\bar{r}_\ell\assign\bar{r}_\ell^\ell(S, x_\ell,  \madd, z_u)$;   $\bar{r}_u\assign  \bar{r}_u^\ell(S, x_u,  \madd, z_\ell) ;$
\STATE
$y'_\ell \assign \invaddl(x_\ell, z_u, \bar{r}_\ell)$;
$y'_u \assign \invaddu(x_u, z_\ell, \bar{r}_u)$;
\IF {$y'_\ell \in \Fset$ and $y'_u\in \Fset$}
\STATE
    $Y' \assign Y \inters [y'_\ell, y'_u]$;
\ELSE
\STATE
    $Y' \assign \emptyset$;
 \ENDIF
\end{algorithmic}
\end{algorithm}

\begin{figure}[h!]
\begin{tabular}{L|CCCCCC}
\invaddl(x_\ell, z_u,\bar{r}_\ell)
       & -\infty & \Rset_- & -0      & +0      & \Rset_+ & +\infty \\
\hline
-\infty & -\infty& -\infty & -\infty & -\infty & -\infty & -\infty \\
\Rset_- & \uns   & a_3      & x_\ell     & x_\ell     & a_3      & -\fmax  \\
-0      & \uns   & -z_u    & -0      & -0      & -z_u    & -\fmax  \\
+0    & \uns    & a_4    & a_4    & a_5   & a_4    & -\fmax  \\
\Rset_+ & \uns     & a_3      & x_\ell     & x_\ell     & a_3      & -\fmax  \\
+\infty & \uns    & +\infty & +\infty & +\infty & a_6   & -\fmax  \\
\end{tabular}
\begin{align*}
  e_\ell
    &\equiv x_\ell + \ftwiceerrnearneg(x_\ell)/2 - z_u; \\
  a_3
    &=
      \begin{cases}
      -0,
        &\text{if $\bar{r}_\ell = \rnear$, $\ftwiceerrnearneg(x_\ell) = -\fmin$ and $x_\ell = z_u$;} \\
      x_\ell \msub_\rup z_u,
        &\text{if $\bar{r}_\ell = \rnear$, $\ftwiceerrnearneg(x_\ell) = -\fmin$ and $x_\ell \neq z_u$;} \\
     \evalup{e_\ell},
        &\text{if $\bar{r}_\ell = \rnear$, $\feven(x_\ell)$, $\ftwiceerrnearneg(x_\ell) \neq -\fmin$
               and $\evalup{e_\ell} = \roundup{e_\ell}$;} \\
               \evaldown{e_\ell},
        &\text{if $\bar{r}_\ell = \rnear$, $\feven(x_\ell)$, $\ftwiceerrnearneg(x_\ell) \neq -\fmin$ and $\evalup{e_\ell} > \roundup{e_\ell}$;} \\
      \fsucc\bigl(\evaldown{e_\ell}\bigr),
        &\text{if $\bar{r}_\ell = \rnear$, otherwise;} \\
      -0,
        &\text{if $\bar{r}_\ell = \rdown$ and $x_\ell = z_u$;} \\
      x_\ell \msub_\rup z_u,
        &\text{if $\bar{r}_\ell = \rdown$ and $x_\ell \neq z_u$;} \\
      \fsucc\bigl(\fpred(x_\ell) \msub_\rdown z_u\bigr),
        &\text{if $\bar{r}_\ell = \rup$;}
      \end{cases} \\
  (a_4, a_5)
    &=
      \begin{cases}
        \bigl(\fsucc(-z_u), +0\bigr), &\text{if $\bar{r}_\ell = \rdown$;} \\
        (-z_u, -0),                   &\text{otherwise;}
      \end{cases} \qquad
  a_6
    =
      \begin{cases}
        +\infty,                       &\text{if $\bar{r}_\ell = \rdown$;} \\
        \fsucc(\fmax \msub_\rdown z_u), &\text{if $\bar{r}_\ell = \rup$;} \\
        \mathrlap{\fmax \madd_\rup \bigl(\ftwiceerrnearpos(\fmax)/2 \msub_\rup z_u\bigr),} \\
                                       &\text{if $\bar{r}_\ell = \rnear$.}
      \end{cases}
\end{align*}
\caption{Inverse projection of addition: function $\invaddl$.}
\label{fig:indirect-projection-addition-yl}
\end{figure}

\begin{figure}[h!]
\begin{tabular}{L|CCCCCC}
\invaddu(x_u, z_\ell, \bar{r}_u)
       & -\infty & \Rset_- & -0      & +0      & \Rset_+ & +\infty \\
\hline
-\infty & \fmax  & a_7   & -\infty & -\infty & -\infty &  \uns     \\
\Rset_- & \fmax  & a_8      & x_u     & x_u     & a_8      &  \uns     \\
-0       & \fmax  & a_9       & a_{10}   & a_9 & a_9    &  \uns    \\
+0      & \fmax  & -z_\ell    & +0      & +0      & -z_\ell    &  \uns     \\
\Rset_+ & \fmax  & a_8      & x_u     & x_u     & a_8      &  \uns     \\
+\infty & +\infty & +\infty & +\infty & +\infty & +\infty & +\infty \\
\end{tabular}
\begin{align*}
  e_u
    &\equiv x_u + \ftwiceerrnearpos(x_u)/2 - z_\ell; \\
  a_7
    &=
      \begin{cases}
        -\infty, &\text{if $\bar{r}_u = \rup$;} \\
        \fpred(-\fmax \msub_\rup z_\ell), &\text{if $\bar{r}_u = \rdown$;} \\
        \mathrlap{-\fmax \madd_\rdown \bigl(\ftwiceerrnearneg(-\fmax) \msub_\rdown z_\ell\bigr);} \\
                                       &\text{if $\bar{r}_u = \rnear$;}
      \end{cases} \qquad
  (a_9, a_{10})
    =
      \begin{cases}
        (-z_\ell, +0),        &\text{if $\bar{r}_u = \rdown$;} \\
        \bigl(\fpred(-z_\ell), -0\bigr), &\text{otherwise;}
      \end{cases} \\
  a_8
    &=
      \begin{cases}
      +0,
        &\text{if $\bar{r}_u = \rnear$, $\ftwiceerrnearpos(x_u) = \fmin$ and $x_u = z_\ell$;}\\
      x_u \msub_\rdown z_\ell,
        &\text{if $\bar{r}_u = \rnear$, $\ftwiceerrnearpos(x_u) = \fmin$ and  $x_u \neq z_\ell$;} \\
      \evaldown{e_u},
        &\text{if $\bar{r}_u = \rnear$, $\feven(x_u)$, $\ftwiceerrnearpos(x_u) \neq \fmin$
               and $\evaldown{e_u} = \rounddown{e_u}$;} \\
                 \evalup{e_u},
        &\text{if $\bar{r}_u = \rnear$, $\feven(x_u)$, $\ftwiceerrnearpos(x_u) \neq \fmin$   and $\evalup{e_u} <\roundup{e_u}$;} \\
      \fpred\bigl(\evalup{e_u}\bigr),
        &\text{$\bar{r}_u = \rnear$, otherwise;} \\
      \fpred\bigl(\fsucc(x_u) \msub_\rup z_\ell\bigr),
        &\text{if $\bar{r}_u = \rdown$;} \\
      +0,
        &\text{if $\bar{r}_u = \rup$ and $x_u = z_\ell$;} \\
      x_u \msub_\rdown z_\ell,
        &\text{if $\bar{r}_u = \rup$ and $x_u \neq z_\ell$.}
      \end{cases} \\
\end{align*}
\caption{Inverse projection of addition: function $\invaddu$.}
\label{fig:indirect-projection-addition-yu}
\end{figure}

The next result assures us that our algorithm computes a
\emph{correct inverse projection}, as claimed by its postcondition.

\begin{theorem}
\label{teo:indirect-projection-addition}
\textup{Algorithm~\ref{algo2}} satisfies its contract.
\end{theorem}
\begin{proof}
Given the constraint $x = y \madd_S z$ with
$x \in X = [x_\ell, x_u]$,
$y \in Y = [y_\ell, y_u]$ and
$z \in Z = [z_\ell, z_u]$,
Algorithm~\ref{algo2} computes a new and refined domain $Y'$ for variable $y$.

First, observe that the newly computed interval $[y'_\ell, y'_u]$
is either intersected with the old domain $Y$, so that
$Y' = [y'_\ell, y'_u] \inters Y$, or set to $Y' = \emptyset$.
Hence, $Y' \sseq Y$ holds.

Proposition~\ref{prop:worst-case-rounding-modes} and the monotonicity of $\madd$
allow us to find a lower bound for $y$ by exploiting the constraint
$y \madd_{\bar{r}_\ell} z_u = x_\ell$, and an upper bound for $y$
by exploiting the constraint $y \madd_{\bar{r}_u} z_\ell = x_u$.
We will now prove that the case analyses of functions $\invaddl$,
described in Figure~\ref{fig:indirect-projection-addition-yl},
and $\invaddu$, described in Figure~\ref{fig:indirect-projection-addition-yu},
express such bounds correctly.

Concerning the operand combinations in which $\invaddl$ takes the value
described by the case analysis $a_4$,
remember that, by the IEEE~754 Standard \cite{IEEE-754-2008},
whenever the sum of two operands with opposite sign is zero,
the result of that sum is $+0$ in all rounding-direction attributes
except roundTowardNegative: in that case the result is $-0$.
Then, since
$z_u \madd_{\rdown}(-z_u) = -0$, when $\bar{r}_\ell = \rdown$,
$y_\ell$ can safely be set to $\fsucc(-z_u)$.

As for the case in which $\invaddl$ takes one of the values determined by $a_5$,
the IEEE~754 Standard \cite{IEEE-754-2008} asserts that
$+0 \madd_{\rdown} +0 = +0$, while $-0 \madd_{\rdown} +0 = -0$:
the correct lower bound for $y$ is $y'_\ell = +0$, in this case.
As we already pointed out, for any other rounding-direction attribute
$+0 \madd -0 = +0$ holds, which allows us to include $-0$ in the new domain.

Concerning cases of $\invaddl$ that  give  the  result described by the case analysis $a_6$,
we clearly must have $y = +\infty$ if $\bar{r}_\ell = \rdown$;
if $\bar{r}_\ell = \rup$, it should be $y + z_u > \fmax$ and thus
$y > \fmax - z_u$ and,
by~\eqref{prop:fp-approx-of-real-constraints:4}
of Proposition~\ref{prop:fp-approx-of-real-constraints},
$y \sgeq \fsucc(\fmax \msub_\rdown z_u)$.
If $\bar{r}_\ell = \rnear$, there are two cases:
\begin{description}
\item[$z_u < \ftwiceerrnearpos(\fmax)/2.$]
  In this case,
  $y$ must be greater than $\fmax$,
  since $\fmax+ z_u < \fmax+ \ftwiceerrnearpos(\fmax)/2$ implies that
  $\fmax \madd_{\rnear} z_u = \fmax < +\infty$. Note that
  in this case $\ftwiceerrnearpos(\fmax)/2 \msub_\rup z_u \geq \fmin$, hence
  $\fmax \madd_\rup \bigl(\ftwiceerrnearpos(\fmax)/2 \msub_\rup z_u\bigr) = +\infty$.
\item[$z_u \geq \ftwiceerrnearpos(\fmax)/2.$]
  Since $\fodd(\fmax)$, for $x_\ell = +\infty$ we need
  $y$ to be greater than or equal to $\fmax + \ftwiceerrnearpos(\fmax)/2 - z_u$.
  Note that $y \geq \fmax + \ftwiceerrnearpos(\fmax)/2 - z_u$ together with
  \begin{equation}\label{eq:tab:addition-special-cases}
    \bigroundup{\fmax + \ftwiceerrnearpos(\fmax)/2 - z_u}
    = \fmax \madd_\rup \bigl(\ftwiceerrnearpos(\fmax)/2 \msub_\rup z_u\bigr)
  \end{equation}
  allows us to apply~\eqref{case:x-geq-e-implies-x-sgeq-evalup-e} of
  Proposition~\ref{prop:fp-approx-of-real-constraints},
  concluding
  $y \sgeq \fmax \madd_\rup \bigl(\ftwiceerrnearpos(\fmax)/2 \msub_\rup z_u\bigr)$.

  Equality~\eqref{eq:tab:addition-special-cases} holds because
  either the application of `$\mathord{\msub_\rup}$' is exact
  or the application of `$\mathord{\madd_\rup}$' is exact.
  In fact,
  since $z_u = m \cdot 2^e \geq \ftwiceerrnearpos(\fmax)/2 = 2^{\emax -p}$,
  for some $1 \leq m < 2$, there are two cases: either $e = \emax$ or
  $\emax - p \leq e < \emax$.

  Suppose first that $e = \emax$: we have
  \begin{align*}
    \ftwiceerrnearpos(\fmax)/2 - z_u
      &= 2^{\emax -p} - m \cdot 2^\emax \\
      &= -2^\emax (m - 2^{-p}),
  \end{align*}
  and thus
\[
  \ftwiceerrnearpos(\fmax)/2 \msub_\rup z_u
    =
      \begin{cases}
        -2^\emax (m - 2^{1-p}),
          &\text{if $m > 1$;} \\
        -2^{\emax - 1}(2 - 2^{1-p}),
          &\text{if $m = 1$.}
      \end{cases}
\]
Since if $e = \emax$ the application of `$\mathord{\msub_\rup}$' is not exact,
we prove that the application of `$\mathord{\madd_\rup}$' is exact.
Hence, if $m > 1$, we prove that
\begin{align*}
  \fmax + (\ftwiceerrnearpos(\fmax)/2 \msub_\rup z_u)
    &= 2^\emax(2 - 2^{1-p}) - 2^\emax(m - 2^{1-p}) \\
    &= 2^\emax(2 - 2^{1-p} - m + 2^{1-p}) \\
    &= 2^\emax(2 - m) \\
    &= 2^{\emax - k} \bigl(2^k (2 - m)\bigr)
\end{align*}
where $k \defeq -\bigl\lfloor\log_2(2-m)\bigr\rfloor$.
It is worth noting that $2^k (2 - m)$ can be represented by a normalized mantissa;
moreover, since $1 \leq k \leq p-1$,
$\emin \leq \emax - k \leq \emax$,
hence,
$\fmax + (\ftwiceerrnearpos(\fmax)/2 \msub_\rup z_u) \in \Fset$.
If, instead, $m = 1$,
\begin{align*}
  \fmax + (\ftwiceerrnearpos(\fmax)/2 \msub_\rup z_u)
    &= 2^\emax(2 - 2^{1-p}) - 2^{\emax-1}(2 - 2^{1-p}) \\
    &= (2^\emax - 2^{\emax-1})(2 - 2^{1-p}) \\
    &= 2^{\emax-1}(2 - 2^{1-p})
\end{align*}
and, also in this case,
$\fmax + (\ftwiceerrnearpos(\fmax)/2 \msub_\rup z_u) \in \Fset$.

Suppose now that $\emax - p \leq e < \emax$ and let
$h \defeq e - \emax + p$ so that $0 \leq h \leq p-1$.
In this case we show that the application of `$\mathord{\msub_\rup}$' is exact.
Indeed, we have
\begin{align*}
  \ftwiceerrnearpos(\fmax)/2 - z_u
    &= 2^{\emax -p} - m \cdot 2^e \\
    &= -2^{\emax - p} (m \cdot 2^h - 1) \\
    &= -2^{\emax - p + h} (m - 2^{-h}).
\end{align*}
If $e = \emax - p$ and $m = 1$, then $h = 0$, $m - 2^{-h} = 0$
and thus $\ftwiceerrnearpos(\fmax)/2 - z_u = 0$.
Otherwise, let $k \defeq -\bigl\lfloor\log_2(m - 2^{-h})\bigr\rfloor$.
We have
\[
  \ftwiceerrnearpos(\fmax)/2 - z_u
    = -2^{\emax - p + h -k} \bigl(2^k (m - 2^{-h})\bigr),
\]
which is an element of $\Fset$.
\end{description}

Dual arguments w.r.t. the ones used to justify cases of $\invaddl$
that give the result described by $a_4$, $a_6$ and $a_5$
can be used to justify the cases of $\invaddu$
described by $a_9$, $a_{10}$ and $a_7$.

We now tackle the entries of $\invaddl$ described by $a_3$,
and those of $\invaddu$ described by $a_8$.
Exploiting $x \sleq y \madd z$ and $x \sgeq y \madd z$,
by Proposition~\ref{prop:real-approx-of-fp-constraints}, we have
\begin{align*}
  y + z
    &\begin{cases}
      \mathord{} \geq x,
        &\text{if $\bar{r}_\ell = \rdown$;} \\
      \mathord{} > x + \ferrup(x) = \fpred(x),
        &\text{if $\bar{r}_\ell = \rup$;} \\
      \mathord{} \geq x + \ftwiceerrnearneg(x)/2,
        &\text{if $\bar{r}_\ell = \rnear$ and $\feven(x)$;} \\
      \mathord{} > x + \ftwiceerrnearneg(x)/2,
        &\text{if $\bar{r}_\ell = \rnear$ and $\fodd(x)$.}
    \end{cases} \\
  y + z
    &\begin{cases}
      \mathord{} < x + \ferrdown(x) = \fsucc(x),
        &\text{if $\bar{r}_u = \rdown$;} \\
      \mathord{} \leq x,
        &\text{if $\bar{r}_u = \rup$;} \\
      \mathord{} \leq x + \ftwiceerrnearpos(x)/2,
        &\text{if $\bar{r}_u = \rnear$ and $\feven(x)$;} \\
      \mathord{} < x + \ftwiceerrnearpos(x)/2,
        &\text{if $\bar{r}_u = \rnear$ and $\fodd(x)$.}
    \end{cases}
\end{align*}
The same case analysis gives us
\begin{align*}
  y
   &\begin{cases}
      \mathord{} \geq x - z,
        &\text{if $\bar{r}_\ell = \rdown$;} \\
      \mathord{} > \fpred(x) - z,
        &\text{if $\bar{r}_\ell = \rup$} \\
      \mathord{} \geq x + \ftwiceerrnearneg(x)/2 - z,
        &\text{if $\bar{r}_\ell = \rnear$ and $\feven(x)$;} \\
      \mathord{} > x + \ftwiceerrnearneg(x)/2 - z,
        &\text{if $\bar{r}_\ell = \rnear$ and $\fodd(x)$;}
    \end{cases} \\
  y
   &\begin{cases}
      \mathord{} < \fsucc(x) - z,
        &\text{if $\bar{r}_u = \rdown$;} \\
       \mathord{} \leq x - z,
        &\text{if $\bar{r}_u = \rup$;} \\
      \mathord{} \leq x + \ftwiceerrnearpos(x)/2 - z,
        &\text{if $\bar{r}_u = \rnear$ and $\feven(x)$;} \\
      \mathord{} < x + \ftwiceerrnearpos(x)/2 - z,
        &\text{if $\bar{r}_u = \rnear$ and $\fodd(x)$.}
    \end{cases}
\end{align*}

We can now exploit the fact that $x \in [x_\ell, x_u]$ and $z \in [z_\ell, z_u]$
with $x_\ell, x_u, z_\ell, z_u \in \Fset$ to obtain, using
Proposition~\ref{prop:min-max-x-ftwiceerrnearneg-ftwiceerrnearpos} and
the monotonicity of `$\fpred$' and `$\fsucc$':
\begin{align}
\label{eq:add-refine-real-lower_bounds}
 y
   &\begin{cases}
      \mathord{} \geq x_\ell - z_u,
        &\text{if $\bar{r}_\ell = \rdown$;} \\
      \mathord{} > \fpred(x_\ell) - z_u,
        &\text{if $\bar{r}_\ell = \rup$;} \\
      \mathord{} \geq x_\ell + \ftwiceerrnearneg(x_\ell)/2 - z_u,
        &\text{if $\bar{r}_\ell = \rnear$ and $\feven(x_\ell)$;} \\
      \mathord{} > x_\ell + \ftwiceerrnearneg(x_\ell)/2 - z_u,
        &\text{if $\bar{r}_\ell = \rnear$ and $\fodd(x_\ell)$.}
    \end{cases} \\
\label{eq:add-refine-real-upper_bounds}
 y
   &\begin{cases}
      \mathord{} < \fsucc(x_u) - z_\ell,
        &\text{if $\bar{r}_u = \rdown$;} \\
      \mathord{} \leq x_u-z_\ell,
        &\text{if $\bar{r}_u = \rup$;} \\
      \mathord{} \leq x_u + \ftwiceerrnearpos(x_u)/2 - z_\ell,
        &\text{if $\bar{r}_u = \rnear$ and $\feven(x_u)$;} \\
      \mathord{} < x_u + \ftwiceerrnearpos(x_u)/2 - z_\ell,
        &\text{if $\bar{r}_u = \rnear$ and $\fodd(x_u)$.}
    \end{cases}
\end{align}
We can now exploit Proposition~\ref{prop:fp-approx-of-real-constraints}
and obtain
\begin{align}
\label{eq:add-inv-num-rup-rdown-lower}
  y'_\ell &\defeq
  \begin{cases}
      -0,
        &\text{if $\bar{r}_\ell = \rdown$ and $x_\ell = z_u$;} \\
      x_\ell \msub_\rup z_u,
        &\text{if $\bar{r}_\ell = \rdown$ and $x_\ell \neq z_u$;} \\
      \fsucc\bigl(\fpred(x_\ell) \msub_\rdown z_u\bigr),
        &\text{if $\bar{r}_\ell = \rup$;}
  \end{cases} \\
\label{eq:add-inv-num-rup-rdown-upper}
     y'_u &\defeq
  \begin{cases}
      \fpred\bigl(\fsucc(x_u) \msub_\rup z_\ell\bigr),
        &\text{if $\bar{r}_u = \rdown$;} \\
      +0,
        &\text{if $\bar{r}_u = \rup$ and $x_u = z_\ell$;} \\
      x_u \msub_\rdown z_\ell,
        &\text{if $\bar{r}_u = \rup$ and $x_u \neq z_\ell$.}
  \end{cases}
\end{align}

In fact, if $x_\ell = z_u$, then, according to IEEE~754
\cite[Section~6.3]{IEEE-754-2008}, for each non-NaN, nonzero and
finite $w \in \Fset$, $-0$ is the least value for $y$ that satisfies
$w = y \madd_\rdown w$.
If $x_\ell \neq z_u$,
then case~\eqref{case:x-geq-e-implies-x-sgeq-evalup-e}
of Proposition~\ref{prop:fp-approx-of-real-constraints}
applies and we have $y \sgeq x_\ell \msub_\rup z_u$.
Suppose now that $\fpred(x_\ell) = z_u$,
then $\fpred(x_\ell) \msub_\rdown z_u \equiv 0$ and
$\fsucc\bigl(\fpred(x_\ell) \msub_\rdown z_u\bigr) = \fmin$,
coherently with the fact that, for each non-NaN, nonzero
and finite $w \in \Fset$, $\fmin$ is the least value for $y$ that satisfies
$w = y \madd_\rup \fpred(w)$.
If $\fpred(x_\ell) \neq z_u$,
then case~\eqref{prop:fp-approx-of-real-constraints:2}
of Proposition~\ref{prop:fp-approx-of-real-constraints}
applies and we have
$y \sgeq \fsucc\bigl(\fpred(x_\ell) \msub_\rdown z_u\bigr)$.
A symmetric argument justifies~\eqref{eq:add-inv-num-rup-rdown-upper}.

For the remaining cases, we first show that when
$\ftwiceerrnearpos(x) = \fmin$,
\begin{equation}\label{eq:inverse-addition-1}
\bigrounddown{x_u + \ftwiceerrnearpos(x_u)/2 - z_\ell } = \rounddown{x_u - z_\ell }.
\end{equation}
The previous equality has the following main consequences:
we can perform the computation
in $\Fset$, that is, we do not need to compute
$\ftwiceerrnearpos(x)/2$ and, since $\rounddown{x_u - z_\ell } = \evaldown{x_u - z_\ell}$,
we can apply~\eqref{case:x-leq-e-implies-x-sleq-evaldown-e}
of Proposition~\ref{prop:fp-approx-of-real-constraints},
obtaining a tight bound for $y_u'$.

Let us prove~\eqref{eq:inverse-addition-1}.
Suppose $\ftwiceerrnearpos(x_u) = \fmin$, and assume $x_u \neq z_\ell$.
There are two cases:
\begin{description}
\item[$\rounddown{x_u - z_\ell} = x_u - z_\ell:$]
then we have $y \leq \rounddown{x_u - z_\ell} = x_u - z_\ell$
since the addition of
$\ftwiceerrnearpos(x_u)/2 = \fmin/2$ is insufficient to
reach $\fsucc(x_u - z_\ell)$, whose distance from $x_u - z_\ell$ is at least $\fmin$.

\item[$\rounddown{x_u - z_\ell} < x_u - z_\ell < \roundup{x_u - z_\ell}:$]
since by Definition~\ref{def:binary-floating-point-number}
every finite floating-point number is an integral multiple of $\fmin$,
so are $x_u - z_\ell$ and $\roundup{x_u - z_\ell}$.
Therefore,
again, $y \leq \rounddown{x_u - z_\ell}$, since the addition
of $\ftwiceerrnearpos(x_u)/2 = \fmin/2$ $x_u - z_\ell$
is insufficient to reach $\roundup{x_u - z_\ell}$,
whose distance from $x_u - z_\ell$ is at least $\fmin$.
\end{description}
In the case where $x_u = z_\ell$ we have
\(
  \bigrounddown{x_u + \ftwiceerrnearpos(x_u)/2 - z_\ell}
    = \rounddown{0 + \fmin/2 }
    = +0
\),
hence~\eqref{eq:inverse-addition-1} holds.
As we have already pointed out, this allows us to
apply~\eqref{case:x-leq-e-implies-x-sleq-evaldown-e}
of Proposition~\ref{prop:fp-approx-of-real-constraints}
to the case $\ftwiceerrnearpos(x_u) = \fmin$, obtaining the bound
$y \sleq \rounddown{x_u - z_\ell}$.

Similar arguments can be applied to
$\ftwiceerrnearneg(x_\ell)$
whenever $\ftwiceerrnearneg(x_\ell) = -\fmin$
to prove that
$\bigroundup{x_\ell + \ftwiceerrnearneg(x_\ell)/2 - z_u} = \roundup{x_\ell - z_u}$.
Then, by~\eqref{case:x-geq-e-implies-x-sgeq-evalup-e}
of Proposition~\ref{prop:fp-approx-of-real-constraints},
we obtain $y \sgeq \roundup{x_\ell - z_u}$.

When the terms $\ftwiceerrnearneg(x_\ell)$ and $\ftwiceerrnearpos(x_u)$
are non-negligible, we need to approximate the values of the expressions
$e_\ell \defeq x_\ell + \ftwiceerrnearneg(x_\ell)/2 - z_u$ and
$e_u \defeq x_u + \ftwiceerrnearpos(x_u)/2 - z_\ell$.
Hence, we have the cases
$\evalup{e_\ell} = \roundup{e_\ell}$ and $\evaldown{e_u} = \rounddown{e_u}$
as well as
$\evalup{e_\ell} > \roundup{e_\ell}$ and $\evaldown{e_u} < \rounddown{e_u}$.
Thus, when
$\evaldown{e_u} < \rounddown{e_u}$
by~\eqref{eq:add-refine-real-upper_bounds}
and~\eqref{prop:fp-approx-of-real-constraints:3}
of Proposition~\ref{prop:fp-approx-of-real-constraints}
we obtain $y \sleq \evalup{e_u}$, while, when
$\evaldown{e_\ell} > \rounddown{e_\ell}$
by~\eqref{eq:add-refine-real-upper_bounds}
and~\eqref{prop:fp-approx-of-real-constraints:1}
of Proposition~\ref{prop:fp-approx-of-real-constraints}
we obtain $y \sgeq \evaldown{e_\ell}$.
Finally, when
$\fodd(x_u)$, by~\eqref{eq:add-refine-real-upper_bounds}
and~\eqref{prop:fp-approx-of-real-constraints:4}
of Proposition~\ref{prop:fp-approx-of-real-constraints},
we obtain $y \sleq \fpred\bigl(\evalup{e_u}\bigr)$.
Dually, when $\fodd(x_\ell)$ by~\eqref{eq:add-refine-real-lower_bounds}
and~\eqref{prop:fp-approx-of-real-constraints:2}
of Proposition~\ref{prop:fp-approx-of-real-constraints},
we obtain $y \sgeq  \fsucc\bigl(\evaldown{e_\ell}\bigr)$.

Thus, for the case  $\bar{r}_\ell = \rnear$ we have
\begin{align}
\label{eq:add-inv-num-rnear-lower}
  y'_\ell &\defeq
    \begin{cases}
      -0,
        &\text{if $\ftwiceerrnearneg(x_\ell) = \fmin$ and $x_\ell = z_u$;} \\
      x_\ell \msub_\rup z_u,
        &\text{if $\ftwiceerrnearneg(x_\ell) = \fmin$ and $x_\ell \neq z_u$;} \\
     \evalup{e_\ell},
        &\text{if $\feven(x_\ell)$, $\ftwiceerrnearneg(x_\ell) \neq \fmin$
               and $\evalup{e_\ell} = \roundup{e_\ell}$;} \\
               \evaldown{e_\ell},
        &\text{if $\feven(x_\ell)$, $\ftwiceerrnearneg(x_\ell) \neq \fmin$ and $\evalup{e_\ell} >\roundup{e_\ell}$;} \\
      \fsucc\bigl(\evaldown{e_\ell}\bigr),
        &\text{otherwise.}
    \end{cases}
    \end{align}

whereas, for the case where $\bar{r}_u = \rnear$, we have
\begin{align}
\label{eq:add-inv-num-rnear-upper}
  y'_u &\defeq
    \begin{cases}
      +0,
        &\text{if $\ftwiceerrnearpos(x_u) = \fmin$ and $x_u = z_\ell$;}\\
      x_u \msub_\rdown z_\ell,
        &\text{if $\ftwiceerrnearpos(x_u) = \fmin$ and  $x_u \neq z_\ell$;} \\
      \evaldown{e_u},
        &\text{if $\feven(x_u)$, $\ftwiceerrnearpos(x_u) \neq \fmin$
               and $\evaldown{e_u} = \rounddown{e_u}$;} \\
                 \evalup{e_u},
        &\text{if $\feven(x_u)$, $\ftwiceerrnearneg(x_u) \neq \fmin$ and $\evalup{e_u} <\roundup{e_u}$;} \\
      \fpred\bigl(\evalup{e_u}\bigr),
        &\text{otherwise.}
    \end{cases}
\end{align}
\end{proof}

\begin{example}
Let $X = [+0, +\infty]$ and $Z = [-\infty, +\infty]$.
Regardless of the rounding mode, the calls to functions
$\invaddl(+0, +\infty, \bar{r}_\ell)$ and
$\invaddu(+\infty, -\infty, \bar{r}_u)$
yield $Y' = [-\fmax, +\infty]$.
Note that $-\fmax$ is the lowest value that variable $y$
could take, since there is no value for $z \in Z$
that summed with $-\infty$ gives a value in $X$.
Indeed, if we take $z = +\fmax$,
then we have $-\fmax \madd_r +\fmax = +0 \in X$
for any $r \in R$.
On the other hand, $+\infty$ is clearly the highest value
$y$ could take, since $+\infty \madd_r z = +\infty \in X$
for any value of $z \in Z \setdiff \{-\infty\}$.
In this case, our projections yield a more refined result
than the competing tool FPSE \cite{BotellaGM06},
which computes the wider interval $Y' = [-\infty, +\infty]$.
\end{example}

\begin{example}
Consider also
$X = [1.0, 2.0]$ and $Z = [-\float{1.0}{30}, \float{1.0}{30}]$ and $S=\{\rnear\}.$
With our inverse projection we obtain
$Y = [-\float{1.1 \cdots 1}{29}, \float{1.0}{30}]$
which is correct but not optimal.
For example, pick $y = \float{1.0}{30}$:
for $z = -\float{1.0}{30}$ we have $y \madd_S z = 0$
and $y \madd_S z^+ = 64$.
By monotonicity of $\madd_S$, for no $z \in [-\float{1.0}{30}, \float{1.0}{30}]$
we can have $y \madd_S z \in [1.0, 2.0]$.
\end{example}

One of the reasons the inverse projection for addition is not optimal is
because floating point numbers present some peculiar properties that
are not related in any way to those of real numbers.  For
interval-based consistency approaches, \cite{MarreM10} identified a
property of the representation of floating-point numbers and proposed
to exploit it in filtering algorithms for addition and subtraction
constraints.  In \cite{BagnaraCGG13ICST,BagnaraCGG16IJOC} some of these authors
revised and corrected the Michel and Marre filtering algorithm on
intervals for addition/subtraction constraints under the round to
nearest rounding mode.  A generalization of such algorithm to the all
rounding modes should be used to enhance the precision of the
classical inverse projection of addition.  Indeed, classical and
maximum ULP filtering \cite{BagnaraCGG16IJOC} for addition are
orthogonal: both should be applied in order to obtain optimal
results. Therefore, inverse projections for addition, as the one
proposed above, have to be intersected with a filter based on the
Michel and Marre property in order to obtain more precise results.

\begin{example}
Assume, again, $X = [1.0, 2.0]$
and $Z = [-\float{1.0}{30}, \float{1.0}{30}]$
and $S = \{\rnear\}$.
By applying maximum ULP filtering~\cite{MarreM10,BagnaraCGG16IJOC},
we obtain the much tighter intervals
$Y, Z = [-\float{1.1\cdots1}{24}, \float{1.0}{25}]$.
These are actually optimal as
\(
  -\float{1.1\cdots1}{24} \madd_S \float{1.0}{25}
  =
  \float{1.0}{25} \madd_S -\float{1.1\cdots1}{24}
  = 2.0
\).
This example shows that filtering by maximum ULP can be stronger
than our interval-consistency based filtering.
However, the opposite phenomenon is also possible. Consider again
$X = [1.0, 2.0]$ and
$Z = [1.0, 5.0]$.
Filtering by maximum ULP projection gives
$Z = [-\float{1.1\cdots1}{24}, \float{1.0}{25}]$;
in contrast, our inverse projection exploits the available
information on $Z$ to obtain $Y = [-4, 1.0 \cdots 01]$.
As we already stated, our filtering and maximum ULP filtering
should both be applied in order to obtain precise results.
\end{example}

Exploiting the commutative property of addition,
the refinement $Z'$ of $Z$ can be defined analogously.

\subsubsection{Division}

In this section we deal with constraints of the form
$x = y \mdiv_S z$ with $S \sseq R$.

\paragraph{Direct Propagation.}

For direct propagation, interval $Z$ is partitioned
into the sign-homogeneous intervals
$Z_- \defeq Z \inters [-\infty, -0]$ and $Z_+ \defeq Z \inters [+0, +\infty]$.
This is needed because the sign of operand $z$ determines the
monotonicity with respect to $y$,
and therefore the interval bounds to be used for propagation
depend on it.
Hence, once $Z$ has been partitioned into sign-homogeneous intervals,
we use the interval $Y$ and $W = Z_-$, to obtain
the new interval $[x^-_\ell, x^-_u]$, and $Y$ and  $W=Z_+$, to obtain $[x^+_\ell, x^+_u]$.
The appropriate bounds for interval propagation are chosen
by function $\tau$ of Figure~\ref{fig:the-tau-function}.
Note that the sign of $z$ is, by construction, constant over interval $W$.
The selected values are then taken as arguments by functions
$\dirdivl$ and $\dirdivu$ of Figure~\ref{fig:direct-projection-division},
which return the correct bounds for the aforementioned
new intervals for $X$.
The intervals $X \inters [x^-_\ell, x^-_u]$ and
$X \inters [x^+_\ell, x^+_u]$ are eventually joined using convex union,
denoted by $\biguplus$, to obtain the refining interval $X'$.

\begin{figure}[ht!]
\[
  \tau(y_\ell, y_u, w_\ell, w_u)
    \defeq
            \begin{cases}
        (y_u, y_\ell, w_\ell, w_u),
          &\text{if $\sgn(w_u)=\sgn(y_u) = -1$;} \\
        (y_u, y_\ell, w_u, w_\ell),
          &\text{if $-\sgn(w_u)=\sgn(y_\ell) = 1$;} \\
        (y_u, y_\ell, w_u, w_u),
          &\text{if $-\sgn(w_u)=-\sgn(y_\ell) = \sgn(y_u)= 1$;}\\
            (y_\ell, y_u, w_\ell, w_u),
          &\text{if $-\sgn(w_\ell)=\sgn(y_u) = -1$;} \\
        (y_\ell, y_u, w_u, w_\ell),
          &\text{if $\sgn(w_\ell)=\sgn(y_\ell) = 1$;} \\
        (y_\ell, y_u, w_\ell, w_\ell),
          &\text{if $\sgn(w_\ell)=-\sgn(y_\ell) = \sgn(y_u)= 1$.}
      \end{cases} \\
         \]
\caption{Direct  projection of division: the function $\tau$; assumes $\sgn(w_\ell) = \sgn(w_u)$}
         \label{fig:the-tau-function}
\end{figure}

\begin{algorithm}
\caption{Direct projection for division constraints.}
\label{algo5}
\begin{algorithmic}[1]
\REQUIRE $x = y \mdiv_S z$,
$x \in X = [x_\ell, x_u]$,
$y \in Y = [y_\ell, y_u]$ and
$z \in Z = [z_\ell, z_u]$.
\ENSURE
$X' \sseq  X$ and
\(
  \forall r \in S, x \in X, y \in Y, z \in Z
    \itc  x = y \mdiv_r z \implies x \in X'
\) and
\(
\forall X''\subset X, \exists r \in S, y \in Y, z \in Z \itc y  \mdiv_r z\not\in X''\).
\STATE
$Z_-\assign  Z \inters [-\infty, -0];$
\IF {$Z_- = [z^-_\ell, z^-_u]  \neq \emptyset$ }
\STATE $W \assign Z_- $;
    \STATE
    $(y_L, y_U,w_L, w_U)\assign \tau(y_\ell, y_u, w_\ell, w_u)$
   \STATE $r_\ell \assign r_\ell(S, y_L, \mdiv, w_L)$; $r_u \assign r_u(S, y_U, \mdiv, w_U)$;
   \STATE
 $x^-_\ell \assign \dirdivl(y_L, w_L, r_\ell)$;
  $x^-_u \assign \dirdivu(y_U, w_U, r_u)$;
\ELSE
      \STATE
      $[x^-_\ell, x^-_u] \assign \emptyset$;
      \ENDIF
 \STATE
   $X'_-= X \inters[x^-_\ell, x^-_u];  $
      \STATE
$Z_+\assign  Z \inters [+0, +\infty];$
\IF {$Z_+ = [z^+_\ell, z^+_u]  \neq \emptyset$ }
    \STATE $W \assign Z_+ $;
 \STATE
    $(y_L, y_U,w_L, w_U)\assign \tau(y_\ell, y_u, w_\ell, w_u)$
       \STATE $r_\ell \assign r_\ell(S, y_L, \mdiv, w_L)$; $r_u \assign r_u(S, y_U, \mdiv, w_U)$;
   \STATE
$x^+_\ell \assign \dirdivl(y_L, w_L, r_\ell)$;
  $x^+_u \assign \dirdivu(y_U, w_U, r_u)$;
  \ELSE
      \STATE
      $[x^+_\ell, x^+_u] \assign \emptyset$;
     \ENDIF
    \STATE
   $X'_+= X \inters[x^+_\ell, x^+_u];  $
             \STATE
$X' \assign X'_-  \biguplus X'_+ $;
  \end{algorithmic}
\end{algorithm}

\begin{figure}[ht!]
\begin{tabular}{L|CCCCCC}
\dirdivl(y_L, w_L,r_\ell)
        & -\infty & \Rset_-  & -0        & +0      & \Rset_+  & +\infty      \\
\hline
-\infty & +\infty & +\infty  & +\infty   & -\infty & -\infty  & -0        \\
\Rset_- & +0      & y_L \mdiv_{r_\ell} w_L      & +\infty   & -\infty & y_L \mdiv_{r_\ell} w_L      & -0        \\
-0      & +0      & +0       & +\infty   & -0      & -0       & -0        \\
+0      & -0      & -0       & -0        & +\infty & +0       & +0        \\
\Rset_+ & -0      & y_L \mdiv_{r_\ell} w_L      & -\infty   & +\infty & y_L \mdiv_{r_\ell} w_L      & +0        \\
+\infty & -0      & -\infty  & -\infty   & +\infty & +\infty  & +\infty   \\
\end{tabular}

\bigskip

\begin{tabular}{L|CCCCCC}
\dirdivu(y_U, w_U,r_u)
        & -\infty & \Rset_-  & -0        & +0      & \Rset_+  & +\infty      \\
\hline
-\infty & +0      & +\infty  & +\infty   & -\infty & -\infty  & -\infty   \\
\Rset_- & +0      & y_U \mdiv_{r_u} w_U     & +\infty   & -\infty & y_U \mdiv_{r_u} w_U     & -0        \\
-0      & +0      & +0       & +0        & -\infty & -0       & -0        \\
+0      & -0      & -0       & -\infty   & +0      & +0       & +0        \\
\Rset_+ & -0      & y_U \mdiv_{r_u} w_U    & -\infty   & +\infty & y_U \mdiv_{r_u} w_U     & +0        \\
+\infty & -\infty & -\infty  & -\infty   & +\infty & +\infty  & +0        \\
\end{tabular}
\caption{Case analyses for direct propagation of division.}
\label{fig:direct-projection-division}
\end{figure}

It can be proved that \textup{Algorithm~\ref{algo5}}
computes a \emph{correct} and \emph{optimal direct projection},
as ensured by its postconditions.

\begin{theorem}
\label{teo:direct-projection-division}
\textup{Algorithm~\ref{algo5}} satisfies its contract.
\end{theorem}

\begin{example}
Consider $Y = [-0, 42]$, $Z = [-3, 6]$ and any value of $S$.
First, $Z$ is split into $Z_- = [-3, -0]$ and $Z_+ = [+0, 6]$.
For the negative interval, the third case of $\tau(-0, 42, -3, -0)$ applies,
yielding $(y_L, y_U,w_L, w_U) = (42, -0, -0, -0)$.
Then, the projection functions are invoked, and we have
$\dirdivl(42, -0, r_\ell) = -\infty$ and
$\dirdivu(-0, -0, r_u) = +0$, i.e., $[x^-_\ell, x^-_u] = [-\infty, +0]$.
For the positive part, we have $\tau(-0, 42, +0, 6) = (-0, 42, +0, +0)$ (sixth case).
From the projections we obtain
$\dirdivl(-0, +0, r_\ell) = -0$ and
$\dirdivu(42, +0, r_u) = +\infty$, and $[x^+_\ell, x^+_u] = [-0, +\infty]$.
Finally, $X' = [x^-_\ell, x^-_u] \biguplus [x^+_\ell, x^+_u] = [-\infty, +\infty]$.
\end{example}

\paragraph{Inverse Propagation (First Projection).}

The inverse projections of division must be handled separately for each operand.
The projection on $y$ is the \emph{first} inverse projection.
This case requires, as explained for Algorithm~\ref{algo5},
to split $Z$ into the sign-homogeneous intervals
$Z_- \defeq Z \inters [-\infty, -0]$ and
$Z_+ \defeq Z \inters [+0, +\infty]$.
Then, in order to select the extrema that determine
the appropriate lower and upper bound for $y$,
function $\sigma$ of Figure~\ref{fig:the-sigma-function}
is applied.

\begin{figure}[h]
\[
  \sigma(z_\ell, z_u, x_\ell, x_u)
    \defeq
      \begin{cases}
        (z_\ell, z_u, x_\ell, x_u),
          &\text{if $\sgn(z_\ell) = \sgn(x_\ell) = 1$;} \\
        (z_u, z_\ell, x_\ell, x_u),
          &\text{if $\sgn(z_\ell) = -\sgn(x_u) = 1$;} \\
        (z_u, z_u, x_\ell, x_u),
          &\text{if $\sgn(z_\ell) = -\sgn(x_\ell) = \sgn(x_u) = 1$;} \\
        (z_u, z_\ell, x_u, x_\ell),
          &\text{if $\sgn(z_u) = \sgn(x_u) = -1$;} \\
        (z_\ell, z_u, x_u, x_\ell),
          &\text{if $-\sgn(z_u) = \sgn(x_\ell) = 1$;} \\
        (z_\ell, z_\ell, x_u, x_\ell),
          &\text{if $-\sgn(z_u) = -\sgn(x_\ell) = \sgn(x_u) = 1$.}
      \end{cases}
\]
\caption{First inverse projection of division: the function $\sigma$;
         assumes $\sgn(z_\ell) = \sgn(z_u)$}
         \label{fig:the-sigma-function}
\end{figure}

\begin{algorithm}[h]
\caption{First inverse projection for division constraints.}
\label{algo6}
\begin{algorithmic}[1]
\REQUIRE $x = y \mdiv_S z$,
$x \in X = [x_\ell, x_u]$,
$y \in Y = [y_\ell, y_u]$ and
$z \in Z = [z_\ell, z_u]$.
\ENSURE
$Y' \sseq  Y$ and
\(
  \forall r \in S, x \in X, y \in Y, z \in Z
    \itc  x = y \mdiv_r z \implies y \in Y'
\).
\STATE
$Z_-\assign  Z \inters [-\infty, -0];$
\IF {$Z_- = [z^-_\ell, z^-_u]  \neq \emptyset$ }
\STATE $W \assign Z_- $;
 \STATE
    $(w_L, w_U,x_L, x_U)\assign \sigma(w_\ell, w_u, x_\ell, x_u)$
         \STATE
        $\bar{r}_\ell\assign\bar{r}_\ell^\ell(S, x_L,  \mdiv, w_L)$;   $\bar{r}_u\assign  \bar{r}_u^\ell(S, x_U,  \mdiv, w_U) ;$
\STATE
$y^-_\ell \assign \invfirstdivl(x_L, w_L, \bar{r}_\ell)$;
  $y^-_u \assign \invfirstdivu(x_U, w_U, \bar{r}_u)$;
 \IF {$y^-_\ell \in \Fset$ and $y^-_u\in \Fset$}
   \STATE
   $Y'_-= Y \inters[y^-_\ell, y^-_u];  $
\ELSE
\STATE
     $  Y'_-=\emptyset;  $
        \ENDIF
    \ELSE
\STATE
     $ Y'_-=\emptyset; $
     \ENDIF
  \STATE
$Z_+\assign  Z \inters [+0,+\infty];$
\IF {$Z_+ = [z^+_\ell, z^+_u]  \neq \emptyset$ }
\STATE $W \assign Z_+ $;
\STATE
 $(w_L, w_U,x_L, x_U)\assign \sigma(w_\ell, w_u, x_\ell, x_u)$;
 \STATE
         $\bar{r}_\ell\assign\bar{r}_\ell^\ell(S, x_L,  \mdiv, w_L)$;   $\bar{r}_u\assign  \bar{r}_u^\ell(S, x_U,  \mdiv, w_U) ;$
             \STATE
$y^+_\ell \assign \invfirstdivl(x_L, w_L, \bar{r}_\ell)$;
  $y^+_u \assign \invfirstdivu(x_U, w_U, \bar{r}_u)$;
 \IF {$y^+_\ell \in \Fset$ and $y^+_u\in \Fset$}
   \STATE
   $Y'_+= Y \inters[y^+_\ell, y^+_u];  $
\ELSE
\STATE
     $  Y'_+=\emptyset;  $
        \ENDIF
    \ELSE
\STATE
     $ Y'_+=\emptyset; $
     \ENDIF
\STATE
$Y' \assign Y'_-  \biguplus Y'_+ $;
 \end{algorithmic}
\end{algorithm}

\begin{figure}[p!]
\begin{tabular}{L|CCCCCC}
\invfirstdivl(x_L, w_L,\bar{r}_\ell)
        & -\infty & \Rset_- & -0      & +0      & \Rset_+ & +\infty \\
\hline
-\infty & \uns    & a_4     & \fmin   & -\infty & -\infty & -\infty \\
\Rset_- & \uns    & a^-_3   & \fmin   & \fmin & a^+_3   & -\fmax  \\
-0      & +0      & +0      & +0      & \fmin & a_7     & -\fmax  \\
+0      & -\fmax  & a_6     & \fmin  & +0      & +0      & +0      \\
\Rset_+ & -\fmax  & a^-_3   & \fmin & \fmin   & a^+_3   & \uns    \\
+\infty & -\infty & -\infty & -\infty & \fmin   & a_5     & \uns    \\
\end{tabular}
\begin{align*}
e^+_\ell & \equiv (x_L+ \ftwiceerrnearneg(x_L)/2) \cdot w_L;\\
a^+_3
    &=
      \begin{cases}
      \evalup{e^+_\ell},
        &\text{if $\bar{r}_\ell = \rnear$, $\feven(x_L)$
               and $\evalup{e^+_\ell} = \roundup{e^+_\ell}$;} \\
      \evaldown{e^+_\ell},
        &\text{if $\bar{r}_\ell = \rnear$, $\feven(x_L)$ and $\evalup{e^+_\ell} > \roundup{e^+_\ell}$;} \\
      \fsucc\bigl(\evaldown{e^+_\ell}\bigr),
        &\text{if $\bar{r}_\ell = \rnear$, otherwise;} \\
      x_L \mmul_\rup w_L,
        &\text{if $\bar{r}_\ell = \rdown$;} \\
      \fsucc\bigl(\fpred(x_L) \mmul_\rdown w_L\bigr),
        &\text{if $\bar{r}_\ell = \rup$;}
      \end{cases} \\
 e^-_\ell & \equiv (x_L+ \ftwiceerrnearpos(x_L)/2) \cdot w_L; \\
 a^-_3
    &=
      \begin{cases}
      \evalup{e^-_\ell},
        &\text{if $\bar{r}_\ell = \rnear$, $\feven(x_L)$
               and $\evalup{e^-_\ell} = \roundup{e^-_\ell}$;} \\
      \evaldown{e^-_\ell},
        &\text{if $\bar{r}_\ell= \rnear$, $\feven(x_L)$  and $\evalup{e^-_\ell} >\roundup{e^-_\ell}$;} \\
      \fsucc\bigl(\evaldown{e^-_\ell}\bigr),
        &\text{if $\bar{r}_\ell = \rnear$, otherwise;} \\
      x_L \mmul_\rup w_L,
        &\text{if $\bar{r}_\ell = \rup$;} \\
      \fsucc\bigl(\fsucc(x_L) \mmul_\rdown w_L\bigr),
        &\text{if $\bar{r}_\ell = \rdown$;}
      \end{cases} \\
e^1_\ell & \equiv (-\fmax+ \ftwiceerrnearneg(-\fmax)/2) \cdot w_L; \\
a_4
    &=
    \begin{cases}
    +\infty,
        &\text{if $\bar{r}_\ell = \rup$;} \\
    \fsucc(-\fmax \mmul_\rdown w_L),
        &\text{if $\bar{r}_\ell = \rdown$;} \\
    \evalup{ e^1_\ell},
        &\text{if $\bar{r}_\ell = \rnear$ and   $\roundup{e^1_\ell}  = \evalup{e^1_\ell}$;} \\
    \evaldown{e^1_\ell},
        &\text{if $\bar{r}_\ell = \rnear$, otherwise;} \\
    \end{cases} \\
e^2_\ell &  \equiv (\fmax + \ftwiceerrnearpos(\fmax)/2) \cdot w_L; \\
a_5
    &=
    \begin{cases}
    +\infty,
        &\text{if $\bar{r}_\ell = \rdown$;} \\
    \fsucc(\fmax \mmul_\rdown w_L),
        &\text{if $\bar{r}_\ell = \rup$;} \\
    \evalup{ e^2_\ell},
        &\text{if $\bar{r}_\ell = \rnear$ and $\roundup{e^2_\ell} = \evalup{e^2_\ell}$;} \\
    \evaldown{e^2_\ell},
        &\text{if $\bar{r}_\ell = \rnear$, otherwise;} \\
    \end{cases} \\
(a_6, a_7)
    &=
    \begin{cases}
    (-0,\;\fsucc(-\fmin \mmul_\rdown w_L)),
        &\text{if $\bar{r}_\ell = \rup$;} \\
    (\fsucc(\fmin \mmul_\rdown w_L)),\;-0),
        &\text{if $\bar{r}_\ell = \rdown$;} \\
    ((\fmin \mmul_\rup w_L)/2, \;(-\fmin \mmul_\rup w_L)/2)),
        &\text{if $\bar{r}_\ell = \rnear$.} \\
    \end{cases} \\
\end{align*}
\caption{First inverse projection of division: function $\invfirstdivl$.}
\label{fig:first-indirect-projection-division-yl}
\end{figure}

\begin{figure}[p!]
\begin{tabular}{L|CCCCCC}
\invfirstdivu(x_U, w_U,\bar{r}_u)
        & -\infty & \Rset_- & -0      & +0      & \Rset_+ & +\infty \\
\hline
-\infty & +\infty & +\infty & +\infty & -\fmin  & a_9     & \uns    \\
\Rset_- & \fmax   & a^-_8   & -\fmin  & -\fmin  & a^+_8   & \uns    \\
-0      & \fmax   & a_{12}  & -\fmin  & -0      & -0      & -0      \\
+0      & -0      & -0      & -0      & -\fmin  & a_{11}  & \fmax   \\
\Rset_+ & \uns    & a^-_8   & -\fmin  & -\fmin  & a^+_8   & \fmax   \\
+\infty & \uns    & a_{10}  & -\fmin  & +\infty & +\infty & +\infty \\
\end{tabular}\begin{align*}
e^+_u &\equiv (x_U + \ftwiceerrnearpos(x_U)/2) \cdot w_U; \\
a^+_8 &=
      \begin{cases}
      \evaldown{e^+_u},
        &\text{if $\bar{r}_u = \rnear$, $\feven(x_U)$
               and $\evaldown{e^+_u} = \rounddown{e^+_u}$;} \\
      \evalup{e^+_u},
        &\text{if $\bar{r}_u = \rnear$, $\feven(x_U)$ and $\evaldown{e^+_u} < \rounddown{e^+_u}$;} \\
      \fpred\bigl(\evalup{e^+_u}\bigr),
        &\text{if $\bar{r}_u = \rnear$, otherwise;} \\
      \fpred\bigl(\fsucc(x_U) \mmul_\rup w_U\bigr),
        &\text{if $\bar{r}_u = \rdown$;} \\
      x_U \mmul_\rdown w_U,
        &\text{if $\bar{r}_u = \rup$;}
      \end{cases} \\
e^-_u &\equiv (x_U + \ftwiceerrnearneg(x_U)/2) \cdot w_U; \\
a^-_8 &=
      \begin{cases}
      \evaldown{e^-_u},
        &\text{if $\bar{r}_u = \rnear$, $\feven(x_U)$
               and $\evaldown{e^-_u} = \rounddown{e^-_u}$;} \\
      \evalup{e^-_u},
        &\text{if $\bar{r}_u = \rnear$, $\feven(x_U)$ and $\evaldown{e^-_u} < \rounddown{e^-_u}$;} \\
      \fpred\bigl(\evalup{e^-_u}\bigr),
        &\text{if $\bar{r}_u = \rnear$, otherwise;} \\
      \fpred\bigl(\fpred(x_U) \mmul_\rup w_U\bigr),
        &\text{if $\bar{r}_u = \rup$;} \\
      x_U \mmul_\rdown w_U,
        &\text{if $\bar{r}_u = \rdown$;}
      \end{cases} \\
e^1_u &\equiv (-\fmax + \ftwiceerrnearneg(-\fmax)/2) \cdot w_U; \\
a_9
    &=
      \begin{cases}
      -\infty, &\text{if $\bar{r}_u = \rup$;} \\
      \fpred(-\fmax \mmul_\rup w_U), &\text{if $\bar{r}_u = \rdown$;} \\
      \evaldown{e^1_u}
        &\text{if $\bar{r}_u = \rnear$ and $\rounddown{e^1_u} = \evaldown{e^1_u}$}; \\
      \evalup{e^1_u} &\text{if $\bar{r}_u = \rnear$, otherwise;} \\
      \end{cases} \\
 e^2_u &\equiv (\fmax + \ftwiceerrnearpos(\fmax)/2) \cdot w_U; \\
 a_{10}
    &=
      \begin{cases}
      -\infty, &\text{if $\bar{r}_u = \rdown$;} \\
      \fpred(\fmax \mmul_\rup w_U), &\text{if $\bar{r}_u = \rup$;} \\
      \evaldown{e^2_u}
        &\text{if $\bar{r}_u = \rnear$ and $\rounddown{e^2_u} = \evaldown{e^2_u}$}; \\
      \evalup{e^2_u} &\text{if $\bar{r}_u = \rnear$, otherwise;} \\
      \end{cases} \\
 (a_{11},a_{12})
    &=
      \begin{cases}
        (+0,\;\fpred(-\fmin \mmul_\rup w_U)),
        &\text{if $\bar{r}_u = \rup$;} \\
        (\fpred(\fmin \mmul_\rup w_U),\;+0),  &\text{if $\bar{r}_u = \rdown$;} \\
        ((\fmin \mmul_\rdown w_U)/2,\; (-\fmin \mmul_\rdown w_U)/2),
          &\text{if $\bar{r}_u = \rnear$.} \\
      \end{cases} \\
\end{align*}
\caption{First inverse projection of division: function $\invfirstdivu$.}
\label{fig:first-indirect-projection-division-yu}
\end{figure}

\begin{example}
Suppose $X = [-42, +0]$, $Z = [-\float{1.0}{100}, -0]$ and $S = \{\rnear\}$.
In this case, $Z_- = Z$, and $Z_+ = \emptyset$. We obtain
$\sigma(-\float{1.0}{100}, -0, -42, +0) = (-\float{1.0}{100}, -\float{1.0}{100}, +0, -42)$
from the sixth case of $\sigma$. Then,
\(
\invfirstdivl(+0, -\float{1.0}{100}, \rnear)
= (\fmin \mmul_\rup (-\float{1.0}{100}))/2
= -\float{1.0}{-50}
\),
because the lowest value of $y_\ell$ is obtained when a division by $-\float{1.0}{100}$ underflows.
Moreover, $\invfirstdivu(-42, -\float{1.0}{100}, \rnear) = \float{1.0101}{105}$.
Therefore, the projected interval is $Y' = [-\float{1.0}{-50}, \float{1.0101}{105}]$.
\end{example}

The following result assures us that Algorithm~\ref{algo6}
computes a \emph{correct first inverse projection},
as ensured by its postcondition.

\begin{theorem}
\label{teo:first-indirect-projection-division}
\textup{Algorithm~\ref{algo6}} satisfies its contract.
\end{theorem}

Once again, in order to obtain more precise results in some cases,
the first inverse projection for division has to be intersected
with a filter based on an extension of the Michel and Marre property
originally proposed in \cite{MarreM10} and
extended to multiplication and division in \cite{BagnaraCGG16IJOC}.
Indeed, when interval $X$ does not contain zeroes
and interval $Z$ contains zeros and infinities,
the proposed filtering by maximum ULP algorithm
is able to derive more precise bounds than the ones obtained with
the inverse projection we are proposing. Thus, for division (and for multiplication as well), the
 indirect projection and filtering by maximum ULP are mutually exclusive: one applies
when the other cannot derive anything useful \cite{BagnaraCGG16IJOC}.

\begin{example}
\label{ex:mu-fdiv1}
Consider the IEEE~754 single-precision constraint
$x = y \mdiv_S z$ with initial intervals $X=[-\float{1.0}{-110}, -\float{1.0}{-121}]$
and
$Y = Z = [-\infty, +\infty]$.
When $S = \{\rnear\}$, filtering by maximum ULP results in the possible refinement
$Y' = [-\float{1.1 \cdots 1}{17}, \float{1.1 \cdots 1}{17}]$,
while Algorithm~\ref{algo6} would return the less precise
$Y'=[-\fmax, \fmax]$, with any rounding mode.
\end{example}

\paragraph{Inverse Propagation (Second Projection).}

The second inverse projection for division
computes a new interval for operand $z$.
For this projection,
we need to partition interval $X$ into sign-homogeneous intervals
$X_- \defeq X \inters [-\infty, -0]$ and
$X_+ \defeq X \inters [+0, +\infty]$ since, in this case,
it is the sign of $X$ that matters for deriving correct
bounds for $Z$.
Once $X$ has been partitioned,
we use intervals $X_-$ and $Y$ to obtain
the interval $[z^-_\ell, z^-_u]$;
intervals $X_+$ and $Y$ to obtain $[z^+_\ell, z^+_u]$.
The new bounds for $z$ are computed by functions
$\invsecdivl$ of Figure~\ref{fig:second-indirect-projection-division-zl} and
$\invsecdivu$ of Figure~\ref{fig:second-indirect-projection-division-zu},
after the appropriate interval extrema of $Y$
and $V = X_-$ (or $V = X_+$)
have been selected by function $\tau$.
The intervals $Z \inters [z^-_\ell, z^-_u]$ and
$Z \inters [z^+_\ell, z^+_u]$ will be then joined
with convex union to obtain $Z'$.

\begin{algorithm}
\caption{Second inverse projection for division constraints.}
\label{algo7}
\begin{algorithmic}[1]
\REQUIRE $x = y \mmul_S z$,
$x \in X = [x_\ell, x_u]$,
$y \in Y = [y_\ell, y_u]$ and
$z \in Z = [z_\ell, z_u]$.
\ENSURE
$Z' \sseq  Z$ and
\(
  \forall r \in S, x \in X, y \in Y, z \in Z
    \itc  x = y \mmul_r z \implies z \in Z'
\).
\STATE
$X_-\assign  X \inters [-\infty, -0];$
\IF {$X_-   \neq \emptyset$ }
\STATE $V\assign X_- $;
  \STATE
    $(y_L, y_U,v_L, v_U)\assign \tau(y_\ell, y_u, v_\ell, v_u)$
   \STATE
$\bar{r}_\ell\assign\bar{r}_r^r(S, v_L,  \mdiv, y_L)$;   $\bar{r}_u\assign  \bar{r}_u^r(S, v_U,  \mdiv, y_U) ;$
\STATE
$z^-_\ell \assign \invsecdivl(y_L, v_L, \bar{r}_\ell)$;
  $z^-_u \assign \invsecdivu(y_U, v_U, \bar{r}_u)$;
 \IF {$z^-_\ell \in \Fset$ and $z^-_u\in \Fset$}
   \STATE
   $Z'_-= Z \inters[z^-_\ell, z^-_u];  $                     
\ELSE  
\STATE  
     $  Z'_-=\emptyset;  $
        \ENDIF                  
    \ELSE  
\STATE  
     $ Z'_-=\emptyset; $ 
     \ENDIF   
     \STATE
$X_+\assign  X \inters [+0, +\infty];$
\IF {$X_+  \neq \emptyset$ }
    \STATE $V \assign X_+ $;
 \STATE
    $(y_L, y_U,v_L, v_U)\assign \tau(y_\ell, y_u, v_\ell, v_u)$     
   \STATE
$\bar{r}_\ell\assign\bar{r}_r^r(S, v_L,  \mdiv, y_L)$;   $\bar{r}_u\assign  \bar{r}_u^r(S, v_U,  \mdiv, y_U) ;$
\STATE
   $z^+_\ell \assign \invsecdivl(y_L, v_L, \bar{r}_\ell)$;
  $z^+_u \assign \invsecdivu(y_U, v_U, \bar{r}_u)$;    
 \IF {$z^+_\ell \in \Fset$ and $z^+_u\in \Fset$}
   \STATE
   $Z'_+= Z \inters[z^+_\ell, z^+_u];  $                     
\ELSE  
\STATE  
     $  Z'_+=\emptyset;  $
        \ENDIF                  
    \ELSE  
\STATE  
     $ Z'_+=\emptyset; $ 
     \ENDIF   
             \STATE
$Z' \assign Z'_-  \biguplus Z'_+ $;
  \end{algorithmic}
\end{algorithm}

\begin{figure}[p!]
\begin{tabular}{L|CCCCCC}
\invsecdivl(y_L, v_L, \bar{r}_\ell)
        & -\infty & \Rset_- & -0      & +0      & \Rset_+ & +\infty \\
\hline
-\infty & +0      & \uns    & \uns    & -\infty & -\fmax  & -\fmax  \\
\Rset_- & +0      &  a^-_3  & a_4     & -\infty &  a^-_3  & a_6     \\
-0      & +0      & \fmin   & \fmin   & -\infty & +0      & +0      \\
+0      & +0      & +0      & -\infty & \fmin   & \fmin   & +0      \\
\Rset_+ & a_7     & a^+_3   & -\infty & a_5     & a^+_3   & +0      \\
+\infty & -\fmax  & -\fmax  & -\infty & \uns    & \uns    & +0      \\
\end{tabular}
\begin{align*}
e^+_\ell & \equiv y_L / (v_L + \ftwiceerrnearpos(v_L)/2);\\
a^+_3
    &=
      \begin{cases}
      \evalup{e^+_\ell},
        &\text{if $\bar{r}_\ell = \rnear$, $\feven(v_L)$
               and $\evalup{e^+_\ell} = \roundup{e^+_\ell}$;} \\
      \evaldown{e^+_\ell},
        &\text{if $\bar{r}_\ell = \rnear$, $\feven(v_L)$ and $\evalup{e^+_\ell} > \roundup{e^+_\ell}$;} \\
      \fsucc(\evaldown{e^+_\ell})
        &\text{if $\bar{r}_\ell = \rnear$, otherwise;} \\
      y_L \mdiv_\rup v_L,
        &\text{if $\bar{r}_\ell = \rup$;} \\
      \fsucc\bigl(y_L \mdiv_\rdown \fsucc(v_L)\bigr),
        &\text{if $\bar{r}_\ell = \rdown$;}
      \end{cases} \\
e^-_\ell & \equiv y_L / (v_L + \ftwiceerrnearneg(v_L)/2); \\
a^-_3
    &=
      \begin{cases}
        \evalup{e^-_\ell},
          &\text{if $\bar{r}_\ell = \rnear$, $\feven(v_L)$
                 and $\evalup{e^-_\ell} = \roundup{e^-_\ell}$;} \\
        \evaldown{e^-_\ell},
          &\text{if $\bar{r}_\ell = \rnear$, $\feven(v_L)$ and $\evalup{e^-_\ell} > \roundup{e^-_\ell}$;} \\
        \fsucc(\evaldown{e^-_\ell})
          &\text{if $\bar{r}_\ell = \rnear$, otherwise;} \\
                 y_L \mdiv_\rup v_L,
          &\text{if $\bar{r}_\ell = \rdown$;} \\
        \fsucc\bigl(y_L \mdiv_\rdown \fpred(v_L)\bigr),
          &\text{if $\bar{r}_\ell = \rup$;}
      \end{cases} \\
(a_4, a_5)
    &=
      \begin{cases}
      (+\infty, \fsucc(y_L \mdiv_\rdown \fmin)),
        &\text{if $\bar{r}_\ell = \rdown$}; \\
      (\fsucc(y_L \mdiv_\rdown -\fmin), +\infty),
        &\text{if $\bar{r}_\ell = \rup$}; \\
      ((y_L \mdiv_\rup -\fmin) \cdot 2, (y_L \mdiv_\rup \fmin) \cdot 2),
        &\text{otherwise}; \\
      \end{cases} \\
e^1_\ell & \equiv y_L / (\fmax + \ftwiceerrnearpos(\fmax)/2); \\
a_6
    &=
      \begin{cases}
      -0,
        &\text{if $\bar{r}_\ell = \rdown$;} \\
      \fsucc(y_L \mdiv_\rdown \fmax),
        &\text{if $\bar{r}_\ell = \rup$;} \\
      \evalup{ e^1_\ell},
        &\text{if $\bar{r}_\ell = \rnear$ and   $\roundup{e^1_\ell}  = \evalup{e^1_\ell}$;} \\
      \evaldown{e^1_\ell},
        &\text{if $\bar{r}_\ell = \rnear$, otherwise;} \\
      \end{cases} \\
e^2_\ell & \equiv y_L / (-\fmax + \ftwiceerrnearneg(-\fmax)/2); \\
a_7
    &=
      \begin{cases}
      -0,
        &\text{if $\bar{r}_\ell = \rup$;} \\
      \fsucc(y_L \mdiv_\rdown -\fmax),
        &\text{if $\bar{r}_\ell = \rdown$;} \\
      \evalup{ e^2_\ell},
        &\text{if $\bar{r}_\ell = \rnear$ and   $\roundup{e^2_\ell}  = \evalup{e^2_\ell}$;} \\
      \evaldown{e^2_\ell},
        &\text{if $\bar{r}_\ell = \rnear$, otherwise.} \\
      \end{cases} \\
\end{align*}
\caption{Second inverse projection of division: function $\invsecdivl$.}
\label{fig:second-indirect-projection-division-zl}
\end{figure}

\begin{figure}[p!]
\begin{tabular}{L|CCCCCC}
\invsecdivu(y_U, v_U, \bar{r}_u)
        & -\infty & \Rset_- & -0      & +0      & \Rset_+ & +\infty  \\
\hline
-\infty & \fmax   & \fmax   & +\infty & \uns    & \uns    & -0       \\
\Rset_- & a_{11}  & a^-_8   & +\infty & a_9     &  a^-_8  & -0       \\
-0      & -0      & -0      & +\infty & -\fmin  & -\fmin  & -0       \\
+0      & -0      & -\fmin  & -\fmin  & +\infty & -0      & -0       \\
\Rset_+ & -0      & a^+_8   & a_{10}  & +\infty &  a^+_8  & a_{12}   \\
+\infty & -0      & \uns    & \uns    & +\infty & \fmax   & \fmax    \\
\end{tabular}
\begin{align*}
e^+_u & \equiv y_U / (v_U + \ftwiceerrnearneg(v_U)/2); \\
a^+_8
    &=
      \begin{cases}
        \evaldown{e^+_u},
          &\text{if $\bar{r}_u = \rnear$, $\feven(v_U)$
                 and $\evaldown{e^+_u} = \rounddown{e^+_u}$;} \\
        \evalup{e^+_u},
          &\text{if $\bar{r}_u = \rnear$, $\feven(v_U)$ and $\evaldown{e^+_u} < \rounddown{e^+_u}$;} \\
        \fpred(\evalup{e^+_u})
          &\text{if $\bar{r}_u = \rnear$, otherwise;}\\
        y_U \mdiv_\rdown v_U,
          &\text{if $\bar{r}_u = \rdown$;} \\
        \fpred\bigl(y_U \mdiv_\rup \fpred(v_U)\bigr),
          &\text{if $\bar{r}_u = \rup$;}
      \end{cases} \\
e^-_u & \equiv y_U / (v_U + \ftwiceerrnearpos(v_U)/2);\\
a^-_8
    &=
      \begin{cases}
        \evaldown{e^-_u},
          &\text{if $\bar{r}_u = \rnear$, $\feven(v_U)$
                 and $\evaldown{e^-_u} = \rounddown{e^-_u}$;} \\
        \evalup{e^-_u},
          &\text{if $\bar{r}_u = \rnear$, $\feven(v_U)$ and $\evaldown{e^-_u} < \rounddown{e^-_u}$;} \\
        \fpred(\evalup{e^-_u})
          &\text{if $\bar{r}_u = \rnear$, otherwise;} \\
        y_U \mdiv_\rdown v_U,
          &\text{if $\bar{r}_u = \rup$;} \\
        \fpred\bigl(y_U \mdiv_\rup \fsucc(v_U)\bigr),
          &\text{if $\bar{r}_u = \rdown$;}
      \end{cases} \\
(a_9, a_{10})
    &=
    \begin{cases}
      (-\infty, \fpred(y_U \mdiv_\rup -\fmin)),
        &\text{if $\bar{r}_u = \rup$}; \\
      (\fpred(y_U \mdiv_\rup \fmin),-\infty),
        &\text{if $\bar{r}_u = \rdown$}; \\
      ((y_U \mdiv_\rdown \fmin) \cdot 2,(y_U \mdiv_\rdown -\fmin) \cdot 2),
        &\text{otherwise}; \\
   \end{cases} \\
e^1_u & \equiv y_U / (-\fmax + \ftwiceerrnearneg(-\fmax)/2);\\
a_{11}
    &=
    \begin{cases}
      +0,
        &\text{if $\bar{r}_u = \rup$;} \\
      \fpred(y_U \mdiv_\rup -\fmax),
        &\text{if $\bar{r}_u = \rdown$;} \\
      \evaldown{ e^1_u},
        &\text{if $\bar{r}_u= \rnear$ and $\rounddown{e^1_u} = \evaldown{e^1_u}$;} \\
      \evalup{e^1_u},
        &\text{if $\bar{r}_u= \rnear$, otherwise;} \\
    \end{cases} \\
e^2_u & \equiv y_U / (\fmax + \ftwiceerrnearpos(\fmax)/2); \\
a_{12}
    &=
    \begin{cases}
      +0,
        &\text{if $\bar{r}_u = \rdown$;} \\
      \fpred(y_U \mdiv_\rup \fmax),
        &\text{if $\bar{r}_u = \rup$;} \\
      \evaldown{e^2_u},
        &\text{if $\bar{r}_u = \rnear$ and $\rounddown{e^2_u} = \evaldown{e^2_u}$;} \\
      \evalup{e^2_u},
        &\text{if $\bar{r}_u = \rnear$, otherwise.} \\
    \end{cases} \\
\end{align*}
\caption{Second inverse projection of division: function $\invsecdivu$.}
\label{fig:second-indirect-projection-division-zu}
\end{figure}

Our algorithm computes a \emph{correct second inverse projection}.
\begin{theorem}
\label{teo:second-indirect-projection-div}
\textup{Algorithm~\ref{algo7}} satisfies its contract.
\end{theorem}

\begin{example}
Consider $X = [6, +\infty]$, $Y = [+0, 42]$ and $S = \{\rnear\}$.
In this case, we only have $X_+ = X$, and $X_- = \emptyset$.
With this input, $\tau(+0, 42, 6, +\infty) = (+0, 42, +\infty, 6)$ (case 5).
Therefore, we obtain $\invsecdivl(+0, +\infty, \rnear) = +0$,
because any number in $Y$ except $+0$ yields $+\infty$ when divided by $+0$.
If we compute intermediate values exactly, $\invsecdivu(42, 6) = 7$
and the refined interval is $Z' = [+0, 7]$.
If not, then $z'_u = \float{1.110 \cdots 01}{2} = \fsucc(7)$.
\end{example}

In order to obtain more precise results, the result of our second inverse projection
can also be intersected with the interval obtained by the maximum ULP filter
proposed in \cite{BagnaraCGG16IJOC}.
Indeed, when interval $X$ does not contain zeros and interval $Y$ contains zeros and infinities,
the proposed filtering by maximum ULP algorithm
is able to derive tighter bounds than those obtained with
the inverse projection presented in this work.

\begin{example}
\label{ex:mu-fdiv2}
Consider the IEEE~754 single-precision division constraint
$x = y \mdiv_S z$ with initial intervals
$x \in [\float{1.0 \cdots 010}{110}, \float{1.0}{121}]$ and
$Y = Z = [-\infty, +\infty]$.
When $S = \{\rnear\}$, filtering by maximum ULP results in the possible refinement
$Z' = [-\float{1.0}{18}, \float{1.0}{18}]$,
while Algorithm~\ref{algo7} would compute $Z' = [-\fmax,\fmax]$,
regardless of the rounding mode.
\end{example}


\afterpage{\clearpage}

\section{Experimental Evaluation}
\label{sec:experimental-evaluation}

The main aim of this section is to motivate the need of provably
correct filtering algorithms, by highlighting the issues caused by
the unsoundness of most available implementations of similar methods.

\subsection{Software Verification}

As we reported in Section~\ref{sec:applications-to-program-analysis},
we implemented our work in the commercial tool ECLAIR.  While the
initial results on a wide range of self-developed tests looked very
promising, we wanted to compare them with the competing tools
presented in the literature, in order to better assess the strength of our
approach with respect to the state of the art.  Unfortunately, most of
these tools were either unavailable, or not sufficiently equipped to
analyze real-world C/\Cplusplus{} programs.  We could, however, do a
comparison with the results obtained in \cite{WuLZ17}.  It presents a
tool called seVR-fpe, for floating-point exception detection based on
symbolic execution and value-range analysis. The same task can be
carried out by the constraint-based symbolic model checker we
included in ECLAIR.  The authors of seVR-fpe tested their tool
both on a self-developed benchmark suite and on real-world programs.
Upon contacting them, they were unfortunately unable to provide us
with more detailed data regarding their analysis of real world
programs. This prevents us from doing an in-depth comparison of the
tools, since we only know the total number of bugs found, but not
their exact nature and location.  Data with this level of detail was
instead available for (most of) their self-developed benchmarks.
The results obtained by running ECLAIR on
them are reported in Table~\ref{table:eclair-vs-wulz17}.
\begin{table}
  \centering
  \begin{tabular}{| l | r | r | r |}
    \hline
    Exception type  & ECLAIR & seVR-fpe & Difference \\
    \hline
    total           & 135    & 66       & 69         \\
    overflow        & 55     & 26       & 29         \\
    underflow       & 30     & 13       & 17         \\
    invalid         & 47     & 8        & 39         \\
    divbyzero       & 3      & 3        & 0          \\
    false positives & 0      & 15       & -15        \\
    \hline
  \end{tabular}
  \caption{Number of exceptions found by ECLAIR and seVR-fpe
    on the self-developed benchmarks of \cite{WuLZ17}.}
  \label{table:eclair-vs-wulz17}
\end{table}
ECLAIR could find a number of possible bugs significantly higher than
seVR-fpe.  As expected, due to the provable correctness of the
algorithms employed in ECLAIR, no false positives were detected among
the inputs it generated. This confirms the solid results obtainable by means of
the algorithms presented in this paper.

\subsection{SMT Solvers}
\label{sec:smt-solvers}

We compare ECLAIR with several SMT solvers that support floating-point arithmetic
by executing them on a benchmark suite devised to test their floating-point theory for soundness.
Since our aim is to evaluate filtering algorithms for floating-point addition,
subtraction, multiplication and division, we only included tests that do not rely
on other theories, such as arrays, bit-vectors and uninterpreted functions,
as well as those containing quantifiers.
Such features are typical of SMT, and are tackled by techniques which are out of the
scope of this paper.
The suite is made of a total of 151,432 tests, of which:
\begin{itemize}
\item 10,380 are randomly generated tests by Florian Schanda.
(\texttt{random} directory from the benchmark suite%
\footnote{\url{https://github.com/florianschanda/smtlib_schanda},
  last accessed on October~28th, 2021.}
used in \cite{BrainSS19}.)

\item 11,544 are randomly generated tests by Christoph M. Wintersteiger.
(\texttt{QF\_FP/wintersteiger} directory from the SMT-LIB benchmark repository.%
\footnote{\url{https://smtlib.cs.uiowa.edu/benchmarks.shtml}, last accessed on October~28th, 2021.})

\item 126,909 tests were generated from the IBM Test Suite for IEEE 754R Compliance%
\footnote{\url{https://www.research.ibm.com/haifa/projects/verification/fpgen/test_suite_download.shtml}, last accessed on October~28th, 2021.}
created with FPgen \cite{AharoniAFKN03}.
\end{itemize}

The experiments were carried out on a high-end laptop with an x86\_64 CPU (6 cores @2.20GHz)
and 16 GB of RAM, running Ubuntu 20.04.
Each solver's version is reported in the table.
MathSAT has been executed with the option \texttt{-theory.fp.mode=2},
which enables the ACDL-based solver for the floating-point theory.

Results are reported in Table~\ref{table:smt-comparison}.
The main observation we can make is that ECLAIR is the only tool based on interval reasoning
to be completely sound.
On the contrary, Colibri and MathSAT are unsound on numerous tests,
even though they did not explicitly report any error.
This hinders their use for program verification.
On the other hand, bit-blasting based tools CVC4 and Z3 do not present such issues,
because their bit-vector encodings for floating-point arithmetic
are derived from solid and formally verified circuit designs such as \cite{MullerP00}.
This demonstrates that the provably-correct filters presented in this paper
are needed to achieve reliable implementations of interval-based constraint solving methods.

The execution times seem to be mainly determined by implementation details
such as the programming language used.
In fact, Colibri and ECLAIR, the slowest tools, were written respectively
in ECLiPSe Prolog\footnote{\url{https://eclipseclp.org/}, last accessed on October~28th, 2021.}
and SWI Prolog,\footnote{\url{https://www.swi-prolog.org/}, last accessed on October~28th, 2021.}
while other tools were written in \Cplusplus{}.

\begin{table}
\centering
\begin{tabular}{| l | l | r | r | r | r |}
\hline
Solver	&	Version	&	Solved	&	Errors	&	Unsound	&	Time (h:m:s)	\\
\hline
Colibri	&	2176	&	148,766	&	0	&	67	&	01:02:42.03	\\
CVC4	&	1.8	&	148,833	&	0	&	0	&	00:09:32.11	\\
ECLAIR	&	---	&	148,833	&	0	&	0	&	01:59:28.80	\\
MathSAT ACDL	&	5.6.5	&	147,958	&	0	&	875	&	00:07:05.89	\\
z3	&	4.8.10	&	148,833	&	0	&	0	&	00:08:38.33	\\
\hline
\end{tabular}
\caption{Results of the evaluation of SMT-solvers.}
\label{table:smt-comparison}
\end{table}

\section{Discussion and Conclusion}
\label{sec:discussion-and-conclusion}

With the increasing use of floating-point computations in mission- and safety-critical settings,
the issue of reliably verifying their correctness has risen to a point
in which testing or other informal techniques are not acceptable any more.
Indeed, this phenomenon has been fostered by the wide adoption
of the IEEE~754 floating-point format, which has significantly simplified the
use of floating-point numbers, by providing a precise, sound,
and reasonably cross-platform specification of floating-point representations,
operations and their semantics.
The approach we propose in this paper exploits these solid foundations
to enable a wide range of floating-point program verification techniques.
It is based on the solution of constraint satisfaction problems by means
of interval-based constraint propagation, which is enabled by the filtering algorithms
we presented. These algorithms cover the whole range of possible floating-point
values, including symbolic values, with respect to interval-based reasoning.
Moreover, they not only support all IEEE~754 available rounding-modes,
but they also allow to take care of uncertainty on the rounding-mode in use.
Some important implementation aspects are also taken into account,
by allowing both the use of machine floating-point arithmetic for all computations
(for increased performance), and of extended-precision arithmetic
(for better precision with the round-to-nearest rounding mode).
In both cases, correctness is guaranteed, so that no valid solutions
can erroneously be removed from the constraint system.
This is supported by the extensive correctness proofs of all algorithms and tables,
which allow us to claim that neither false positives, nor false negatives may be produced.
The experimental evaluation of Section~\ref{sec:experimental-evaluation}
shows that soundness is, indeed, a widespread issue in several floating-point
verification tools. Our work provides solid foundations to soundly develop
such kind of tools.

Several aspects of the constraint-based verification of floating-point
programs remain, however, open problems, both from a theoretical and a
practical perspective.  As we showed throughout the paper, the
filtering algorithms we presented are not optimal, i.e., they may not
yield the tightest possible intervals containing all solutions to the
constraint system. They must be interleaved with the filtering
algorithms of \cite{BagnaraCGG16IJOC}, and they may require multiple
passes before reaching the maximum degree of variable-domain pruning
they are capable of.  Therefore, the next possible advance in
this direction would be conceiving optimal filtering algorithms, that
reduce variable domains to intervals as tight as possible with a
single application. This has been achieved in \cite{Gallois-WongBC20},
but only for addition.

However, filtering algorithms only represent a significant, but to
some extent limited, part of the constraint solving process.  Indeed,
even an optimally pruned interval may contain values that are not
solutions to the constraint system, due to the possible non-linearity
thereof.  If the framework in use supports multi-intervals, this issue
is dealt with by means of labeling techniques: when a
constraint-solving process reaches quiescence, i.e., the application
of filtering algorithms fails to prune variable domains any further,
such intervals are split into two or more sub-intervals, and the process
continues on each partition separately. In this context, the main
issues are \emph{where} to split intervals, and in \emph{how many}
parts. These issues are currently addressed with heuristic labeling
strategies.  Indeed, significant improvements to the
constraint-propagation process could be achieved by investigating
better labeling strategies. To this end, possible advancements would
include the identification of objective criteria for the evaluation of
labeling strategies on floating point-numbers, and the conception of
labeling strategies tailored to the properties of constraint systems
most commonly generated by numeric programs.

In conclusion, we believe the work presented in this paper can be an extensive
reference for the readers interested in realizing applications for formal
reasoning on floating-point computations, as well as a solid foundation
for further improvements in the state of the art.

\begin{acknowledgements}
The authors express their gratitude to Roberto Amadini,
Isacco Cattabiani and Laura Savino
for their careful reading of early versions of this paper.
\end{acknowledgements}


\begin{thebibliography}{10}
\providecommand{\url}[1]{{#1}}
\providecommand{\urlprefix}{URL }
\expandafter\ifx\csname urlstyle\endcsname\relax
  \providecommand{\doi}[1]{DOI~\discretionary{}{}{}#1}\else
  \providecommand{\doi}{DOI~\discretionary{}{}{}\begingroup
  \urlstyle{rm}\Url}\fi

\bibitem{AharoniAFKN03}
Aharoni, M., Asaf, S., Fournier, L., Koyfman, A., Nagel, R.: {FPgen} --- a test
  generation framework for datapath floating-point verification.
\newblock In: Eighth {IEEE} International High-Level Design Validation and Test
  Workshop, pp. 17--22. {IEEE} Computer Society, San Francisco, CA, USA (2003).
\newblock \doi{10.1109/HLDVT.2003.1252469}

\bibitem{AhoLSU06}
Aho, A.V., Lam, M.S., Sethi, R., Ullman, J.D.: Compilers: Principles,
  Techniques, and Tools, 2nd edn.
\newblock Addison-Wesley Longman Publishing Co., Inc., Boston, MA, USA (2006)

\bibitem{Armv8}
{Arm Limited}: Arm\textregistered{} Architecture Reference Manual, {Armv8}, for
  {A-profile} architecture edn. (2021).
\newblock
  \urlprefix\url{https://developer.arm.com/architectures/cpu-architecture/a-profile/docs}.
\newblock Last accessed on October 27th, 2021

\bibitem{BagnaraCGG13ICST}
Bagnara, R., Carlier, M., Gori, R., Gotlieb, A.: Symbolic path-oriented test
  data generation for floating-point programs.
\newblock In: Proceedings of the 6th IEEE International Conference on Software
  Testing, Verification and Validation. IEEE Press, Luxembourg City, Luxembourg
  (2013).
\newblock \doi{10.1109/ICST.2013.17}

\bibitem{BagnaraCGG16IJOC}
Bagnara, R., Carlier, M., Gori, R., Gotlieb, A.: Exploiting binary
  floating-point representations for constraint propagation.
\newblock {INFORMS} Journal on Computing \textbf{28}(1), 31--46 (2016).
\newblock \doi{10.1287/ijoc.2015.0663}

\bibitem{BagnaraCGB21}
Bagnara, R., Chiari, M., Gori, R., Bagnara, A.: A practical approach to
  verification of floating-point {C/C++} programs with math.h/cmath functions.
\newblock ACM Transactions on Software Engineering and Methodology
  \textbf{30}(1), 9:1--9:53 (2021).
\newblock \doi{10.1145/3410875}

\bibitem{BarrVLS13}
Barr, E.T., Vo, T., Le, V., Su, Z.: Automatic detection of floating-point
  exceptions.
\newblock In: The 40th Annual {ACM} {SIGPLAN-SIGACT} Symposium on Principles of
  Programming Languages, {POPL} '13, Rome, Italy - January 23 - 25, 2013, pp.
  549--560 (2013).
\newblock \doi{10.1145/2429069.2429133}

\bibitem{BarrettCDHJKRT11}
Barrett, C.W., Conway, C.L., Deters, M., Hadarean, L., Jovanovic, D., King, T.,
  Reynolds, A., Tinelli, C.: {CVC4}.
\newblock In: Computer Aided Verification, Proceedings of the 23rd
  International Conference (CAV 2011), \emph{Lecture Notes in Computer
  Science}, vol. 6806, pp. 171--177. Springer, Snowbird, UT, USA (2011).
\newblock \doi{10.1007/978-3-642-22110-1\_14}

\bibitem{BarrettT18}
Barrett, C.W., Tinelli, C.: Satisfiability modulo theories.
\newblock In: E.M. Clarke, T.A. Henzinger, H.~Veith, R.~Bloem (eds.) Handbook
  of Model Checking, pp. 305--343. Springer (2018).
\newblock \doi{10.1007/978-3-319-10575-8\_11}

\bibitem{BotellaGM06}
Botella, B., Gotlieb, A., Michel, C.: Symbolic execution of floating-point
  computations.
\newblock Software Testing, Verification and Reliability \textbf{16}(2),
  97--121 (2006).
\newblock \doi{10.1002/stvr.333}

\bibitem{BrainDGHK14}
Brain, M., {D'Silva}, V., Griggio, A., Haller, L., Kroening, D.: Deciding
  floating-point logic with abstract conflict driven clause learning.
\newblock Formal Methods in System Design \textbf{45}(2), 213--245 (2014).
\newblock \doi{10.1007/s10703-013-0203-7}

\bibitem{BrainSS19}
Brain, M., Schanda, F., Sun, Y.: Building better bit-blasting for
  floating-point problems.
\newblock In: Tools and Algorithms for the Construction and Analysis of
  Systems, Proceedings of the 25th International Conference (TACAS 2019), Part
  {I}, \emph{Lecture Notes in Computer Science}, vol. 11427, pp. 79--98.
  Springer (2019).
\newblock \doi{10.1007/978-3-030-17462-0\_5}

\bibitem{BrainTRW15}
Brain, M., Tinelli, C., R{\"u}mmer, P., T-Wahl: An automatable formal semantics
  for {IEEE-754} floating-point arithmetic.
\newblock In: 22nd {IEEE} Symposium on Computer Arithmetic (ARITH 2015), pp.
  160--167. {IEEE}, Lyon, France (2015).
\newblock \doi{10.1109/ARITH.2015.26}

\bibitem{CimattiGSS13}
Cimatti, A., Griggio, A., Schaafsma, B.J., Sebastiani, R.: The {MathSAT5} {SMT}
  solver.
\newblock In: Tools and Algorithms for the Construction and Analysis of
  Systems, Proceedings of the 19th International Conference (TACAS 2013),
  \emph{Lecture Notes in Computer Science}, vol. 7795, pp. 93--107. Springer
  (2013).
\newblock \doi{10.1007/978-3-642-36742-7\_7}

\bibitem{ClarkeR85}
Clarke, L.A., Richardson, D.J.: Applications of symbolic evaluation.
\newblock Journal of Systems and Software \textbf{5}(1), 15--35 (1985).
\newblock \doi{10.1016/0164-1212(85)90004-4}

\bibitem{CousotC77}
Cousot, P., Cousot, R.: Abstract interpretation: A unified lattice model for
  static analysis of programs by construction or approximation of fixpoints.
\newblock In: Proceedings of the Fourth Annual ACM Symposium on Principles of
  Programming Languages, pp. 238--252. ACM Press, Los Angeles, CA, USA (1977).
\newblock \doi{10.1145/512950.512973}

\bibitem{CuytKVV02}
Cuyt, A., Kuterna, P., Verdonk, B., D.~Verschaeren, D.: Underflow revisited.
\newblock CALCOLO \textbf{39}(3), 169--179 (2002).
\newblock \doi{10.1007/s100920200003}

\bibitem{MouraB08}
{de Moura}, L.M., Bj{\o}rner, N.: {Z3:} an efficient {SMT} solver.
\newblock In: Tools and Algorithms for the Construction and Analysis of
  Systems, Proceedings of the 14th International Conference (TACAS 2008),
  \emph{Lecture Notes in Computer Science}, vol. 4963, pp. 337--340. Springer
  (2008).
\newblock \doi{10.1007/978-3-540-78800-3\_24}

\bibitem{DelmasGPSTV09}
Delmas, D., Goubault, E., Putot, S., Souyris, J., Tekkal, K., V\'{e}drine, F.:
  Towards an industrial use of {FLUCTUAT} on safety-critical avionics software.
\newblock In: Formal Methods for Industrial Critical Systems, Proceedings of
  the 14th International Workshop (FMICS 2009), \emph{Lecture Notes in Computer
  Science}, vol. 5825, pp. 53--69. Springer, Eindhoven, The Netherlands (2009).
\newblock \doi{10.1007/978-3-642-04570-7\_6}

\bibitem{Fiedler10}
Fiedler, G.: Floating point determinism.
\newblock Gaffer On Games Blog, February 24, 2019,
  \url{https://gafferongames.com/post/floating_point_determinism/} (2010).
\newblock Last accessed on October~28th, 2021

\bibitem{Gallois-WongBC20}
{Gallois-Wong}, D., Boldo, S., Cuoq, P.: Optimal inverse projection of
  floating-point addition.
\newblock Numerical Algorithms \textbf{83}(3), 957--986 (2020).
\newblock \doi{10.1007/s11075-019-00711-z}

\bibitem{GotliebBR98}
Gotlieb, A., Botella, B., Rueher, M.: Automatic test data generation using
  constraint solving techniques.
\newblock In: Proceedings of the 1998 ACM SIGSOFT International Symposium on
  Software Testing and Analysis, ISSTA '98, pp. 53--62. ACM, New York, NY, USA
  (1998).
\newblock \doi{10.1145/271771.271790}

\bibitem{GotliebBR00}
Gotlieb, A., Botella, B., Rueher, M.: A {CLP} framework for computing
  structural test data.
\newblock In: Computational Logic --- CL 2000: First International Conference
  London, UK, July 24--28, 2000 Proceedings, pp. 399--413. Springer (2000).
\newblock \doi{10.1007/3-540-44957-4_27}

\bibitem{IEEE-754-2008}
The Institute of Electrical and Electronics Engineers, Inc.: {IEEE} Standard
  for Floating-Point Arithmetic, {IEEE Std 754-2008} (revision of {IEEE Std
  754-1985}) edn. (2008).
\newblock Available at
  \url{http://ieeexplore.ieee.org/servlet/opac?punumber=4610933}

\bibitem{Intelx8664}
{Intel Corporation}: Intel\textregistered{} 64 and IA-32 Architectures Software
  Developer Manuals, 325462-075us edn. (2021).
\newblock
  \urlprefix\url{https://www.intel.com/content/www/us/en/developer/articles/technical/intel-sdm.html}.
\newblock Last accessed on October 27th, 2021

\bibitem{King76}
King, J.C.: Symbolic execution and program testing.
\newblock Communications of the ACM \textbf{19}(7), 385--394 (1976).
\newblock \doi{10.1145/360248.360252}

\bibitem{KornerupLLM09}
Kornerup, P., Lefevre, V., Louvet, N., Muller, J.M.: On the computation of
  correctly-rounded sums.
\newblock In: Proceedings of the 19th IEEE Symposium on Computer Arithmetic
  (ARITH 2009), pp. 155--160. Portland, OR, USA (2009).
\newblock \doi{10.1109/ARITH.2009.16}

\bibitem{MarreBC17}
Marre, B., Bobot, F., Chihani, Z.: Real behavior of floating point numbers.
\newblock In: Proceedings of the 15th International Workshop on Satisfiability
  Modulo Theories, \emph{CEUR Workshop Proceedings}, vol. 1889, pp. 50--62
  (2017)

\bibitem{MarreM10}
Marre, B., Michel, C.: Improving the floating point addition and subtraction
  constraints.
\newblock In: D.~Cohen (ed.) Proceedings of the 16th International Conference
  on Principles and Practice of Constraint Programming (CP 2010), \emph{Lecture
  Notes in Computer Science}, vol. 6308, pp. 360--367. Springer, St. Andrews,
  Scotland, UK (2010).
\newblock \doi{10.1007/978-3-642-15396-9\_30}

\bibitem{Michel02}
Michel, C.: Exact projection functions for floating point number constraints.
\newblock In: Proceedings of the 7th International Symposium on Artificial
  Intelligence and Mathematics. Fort Lauderdale, FL, USA (2002)

\bibitem{MichelRL01}
Michel, C., Rueher, M., Lebbah, Y.: Solving constraints over floating-point
  numbers.
\newblock In: Principles and Practice of Constraint Programming - {CP} 2001,
  7th International Conference, {CP} 2001, Paphos, Cyprus, November 26 -
  December 1, 2001, Proceedings, pp. 524--538 (2001).
\newblock \doi{10.1007/3-540-45578-7\_36}

\bibitem{Mine04}
Min\'{e}, A.: Relational abstract domains for the detection of floating-point
  run-time errors.
\newblock In: Programming Languages and Systems, Proceedings of the 13th
  European Symposium on Programming (ESOP 2004), \emph{Lecture Notes in
  Computer Science}, vol. 2986, pp. 3--17. Springer, Barcelona, Spain (2004).
\newblock \doi{10.1007/978-3-540-24725-8\_2}

\bibitem{Monniaux08}
Monniaux, D.: The pitfalls of verifying floating-point computations.
\newblock ACM Transactions on Programming Languages and Systems \textbf{30}(3),
  12:1--12:41 (2008).
\newblock \doi{10.1145/1353445.1353446}

\bibitem{MullerP00}
M{\"u}ller, S.M., Paul, W.J.: Computer Architecture - Complexity and
  Correctness.
\newblock Springer (2000).
\newblock \doi{10.1007/978-3-662-04267-0}

\bibitem{Rump13}
Rump, S.M.: Accurate solution of dense linear systems, {Part II}: Algorithms
  using directed rounding.
\newblock Journal of Computational and Applied Mathematics \textbf{242},
  185--212 (2013).
\newblock \doi{10.1016/j.cam.2012.09.024}

\bibitem{RumpO07}
Rump, S.M., Ogita, T.: Super-fast validated solution of linear systems.
\newblock Journal of Computational and Applied Mathematics \textbf{199}(2),
  199--206 (2007).
\newblock \doi{10.1016/j.cam.2005.07.038}.
\newblock Special Issue on Scientific Computing, Computer Arithmetic, and
  Validated Numerics (SCAN 2004)

\bibitem{Watte08}
Watte, J.: Floating point determinism.
\newblock GameDev.net Forum, June 30, 2008,
  \url{https://www.gamedev.net/forums/topic/499435-floating-point-determinism/}
  (2008).
\newblock Last accessed on October~28th, 2021

\bibitem{WuLZ17}
Wu, X., Li, L., Zhang, J.: Symbolic execution with value-range analysis for
  floating-point exception detection.
\newblock In: 24th Asia-Pacific Software Engineering Conference, {APSEC} 2017,
  Nanjing, China, December 4-8, 2017, pp. 1--10 (2017).
\newblock \doi{10.1109/APSEC.2017.6}

\end{thebibliography}

\hyphenation{ Ba-gna-ra Bie-li-ko-va Bruy-noo-ghe Common-Loops DeMich-iel
  Dober-kat Di-par-ti-men-to Er-vier Fa-la-schi Fell-eisen Gam-ma Gem-Stone
  Glan-ville Gold-in Goos-sens Graph-Trace Grim-shaw Her-men-e-gil-do Hoeks-ma
  Hor-o-witz Kam-i-ko Kenn-e-dy Kess-ler Lisp-edit Lu-ba-chev-sky
  Ma-te-ma-ti-ca Nich-o-las Obern-dorf Ohsen-doth Par-log Para-sight Pega-Sys
  Pren-tice Pu-ru-sho-tha-man Ra-guid-eau Rich-ard Roe-ver Ros-en-krantz
  Ru-dolph SIG-OA SIG-PLAN SIG-SOFT SMALL-TALK Schee-vel Schlotz-hauer
  Schwartz-bach Sieg-fried Small-talk Spring-er Stroh-meier Thing-Lab Zhong-xiu
  Zac-ca-gni-ni Zaf-fa-nel-la Zo-lo }

\clearpage
\appendix

\ESM{Appendices}

\section{Filtering algorithms: Subtraction and Multiplication}
\label{se:Subtraction-Multiplication}

\subsection{Subtraction}

Here we deal with constraints of the form
$x = y \msub_S z$.

Assume
$X = [x_\ell, x_u]$,
$Y = [y_\ell, y_u]$ and
$Z = [z_\ell, z_u]$.

Again, thanks to Proposition ~\ref{prop:worst-case-rounding-modes} we need not
be concerned with sets of rounding modes, as any such set $S \sseq R$
can always be mapped to a pair of ``worst-case rounding modes'' which,
in addition are never round-to-zero.

\paragraph{Direct Propagation.}
For direct propagation, we use Algorithm~\ref{algo:direct-subtraction} and
functions $\dirsubl$ and $\dirsubu$, as defined in
Figure~\ref{fig:direct-projection-subtraction}.

\begin{algorithm}
\caption{Direct projection for subtraction constraints.}
\label{algo:direct-subtraction}
\begin{algorithmic}[1]
\REQUIRE $x = y \msub_S z$,
$x \in X = [x_\ell, x_u]$,
$y \in Y = [y_\ell, y_u]$ and
$z \in Z = [z_\ell, z_u]$.
\ENSURE
$X' \sseq  X$ and
\(
  \forall r \in S, x \in X, y \in Y, z \in Z
    \itc  x = y \msub_r z \implies x \in X'
\)
and
\(
\forall X''\subset X, \exists r \in S, y \in Y, z \in Z \itc y
\msub_r z\not\in X''\).
\STATE
$r_\ell \assign r_\ell(S, y_\ell, \msub, z_u)$; $r_u \assign r_u(S, y_u, \msub, z_\ell)$;
\STATE
$x'_\ell \assign \dirsubl(y_\ell, z_u, r_\ell)$;
$x'_u \assign \dirsubu(y_u, z_\ell, r_u)$;
\STATE
$X' \assign X \inters [x'_\ell, x'_u]$;
\end{algorithmic}
\end{algorithm}

\begin{figure}[ht!]
\begin{tabular}{L|CCCCCC}
\dirsubl(y_\ell, z_u, r_\ell)
        & -\infty & \Rset_-             & -0      & +0      & \Rset_+             & +\infty \\
\hline
-\infty & +\infty & -\infty             & -\infty & -\infty & -\infty             & -\infty \\
\Rset_- & +\infty & y_\ell \msub_{r_\ell} z_u & y_\ell     & y_\ell     & y_\ell \msub_{r_\ell} z_u & -\infty \\
-0      & +\infty & -z_u                & a_1     & -0      & -z_u                & -\infty \\
+0      & +\infty & -z_u                & +0      & a_1     & -z_u                & -\infty \\
\Rset_+ & +\infty & y_\ell \msub_{r_\ell} z_u & y_\ell     & y_\ell     & y_\ell \msub_{r_\ell} z_u & -\infty \\
+\infty & +\infty & +\infty             & +\infty & +\infty & +\infty             & +\infty \\
\end{tabular}
\[
a_1
 =
   \begin{cases}
     -0, &\text{if $r_\ell = \rdown$,} \\
     +0, &\text{otherwise;}
   \end{cases}
\]

\bigskip
\begin{tabular}{L|CCCCCC}
\dirsubu(y_u, z_\ell, r_u)
        & -\infty & \Rset_-             & -0      & +0      & \Rset_+             & +\infty \\
\hline
-\infty & -\infty & -\infty             & -\infty & -\infty & -\infty             & -\infty \\
\Rset_- & +\infty & y_u \msub_{r_u} z_\ell & y_u     & y_u     & y_u \msub_{r_u} z_\ell & -\infty \\
-0      & +\infty & -z_\ell                & a_2     & -0      & -z_\ell                & -\infty \\
+0      & +\infty & -z_\ell                & +0      & a_2     & -z_\ell                & -\infty \\
\Rset_+ & +\infty & y_u \msub_{r_u} z_\ell & y_u     & y_u     & y_u \msub_{r_u} z_\ell & -\infty \\
+\infty & +\infty & +\infty             & +\infty & +\infty & +\infty             & -\infty \\
\end{tabular}
\[
a_2
 =
   \begin{cases}
     -0, &\text{if $r_u = \rdown$,} \\
     +0, &\text{otherwise.}
   \end{cases}
\]
\caption{Direct projection of subtraction: function $\dirsubl$
         (resp., $\dirsubu$);
         values for $y_\ell$ (resp., $y_u$) on rows,
         values for $z_u$ (resp., $z_\ell$) on columns.}
\label{fig:direct-projection-subtraction}
\end{figure}

\begin{theorem}
\label{teo:direct-projection-subtraction}
\textup{Algorithm~\ref{algo:direct-subtraction}} satisfies its contract.
\end{theorem}

\paragraph{Inverse Propagation.}

For inverse propagation, we have to deal with two different cases
depending on which variable we are computing: the first inverse projection
on $y$ or the second inverse projection on $z$.

The first inverse projection of subtraction is somehow similar to
the direct projection of addition. In this case we define
Algorithm~\ref{algo:first-indirect-subtraction} and functions $\invfirstsubl$
and $\invfirstsubu$, as defined in
Figure~\ref{fig:first-indirect-projection-subtraction-lower}
and~\ref{fig:first-indirect-projection-subtraction-upper} respectively.

\begin{algorithm}
\caption{First inverse projection for subtraction constraints.}
\label{algo:first-indirect-subtraction}
\begin{algorithmic}[1]
\REQUIRE $x = y \msub_S z$,
$x\in X = [x_\ell, x_u]$,
$y\in Y = [y_\ell, y_u]$ and
$z\in Z = [z_\ell, z_u]$.
\ENSURE
$Y' \sseq  Y$ and
\(
  \forall r \in S, x \in X, y \in Y, z \in Z
    \itc  x = y \msub_r z \implies y \in Y'
\).

\STATE
         $\bar{r}_\ell\assign\bar{r}_\ell^\ell(S, x_\ell,  \msub, z_\ell)$;   $\bar{r}_u\assign  \bar{r}_u^\ell(S, x_u,  \msub, z_u) ;$
\STATE
$y'_\ell \assign \invfirstsubl(x_\ell, z_\ell, \bar{r}_\ell)$;
$y'_u \assign \invfirstsubu(x_u, z_u, \bar{r}_u)$;
\IF {$y'_\ell \in \Fset$ and $y'_u\in \Fset$}
\STATE
    $Y' \assign Y \inters [y'_\ell, y'_u]$;
\ELSE
\STATE
    $Y' \assign \emptyset$;
 \ENDIF
\end{algorithmic}
\end{algorithm}

\begin{figure}
\begin{tabular}{L|CCCCCC}
\invfirstsubl(x_\ell, z_\ell, \bar{r}_\ell)
       & -\infty & \Rset_- & -0      & +0      & \Rset_+ & +\infty \\
\hline
-\infty & -\infty & -\infty & -\infty & -\infty & -\infty & -\infty \\
\Rset_- & -\fmax  & a_3      & x_\ell     & x_\ell     & a_3      & \uns    \\
-0      & -\fmax  & z_\ell     & -0      & -0      & z_\ell    & \uns    \\
+0      & -\fmax  & a_4    & a_5      & a_4      & a_4    & \uns    \\
\Rset_+ & -\fmax  & a_3      & x_\ell     & x_\ell     & a_3      & \uns    \\
+\infty & -\fmax  & a_6   & +\infty & +\infty & +\infty & \uns    \\
\end{tabular}
\begin{align*}
e_\ell & \equiv x_\ell + \ftwiceerrnearneg(x_\ell)/2 + z_\ell; \\
a_3 & =
     \begin{cases}
     -0, &
           \text{if $\bar{r}_\ell = \rnear$, $\ftwiceerrnearneg(x_\ell) = -\fmin$
           and $x_\ell = -z_\ell$;} \\
     x_\ell \madd_\rup z_\ell, &
           \text{if $\bar{r}_\ell = \rnear$, $\ftwiceerrnearneg(x_\ell) = -\fmin$
           and $x_\ell \neq -z_\ell$;} \\
    \evalup{e_\ell}, &
           \text{if $\bar{r}_\ell = \rnear$, $\feven(x_\ell)$,
           $\ftwiceerrnearneg(x_\ell) \neq -\fmin$ and $\evalup{e_\ell} =
           \roundup{e_\ell}$;} \\
    \evaldown{e_\ell}, &
           \text{if $\bar{r}_\ell = \rnear$, $\feven(x_\ell)$,
           $\ftwiceerrnearneg(x_\ell) \neq -\fmin$ and
           $\evalup{e_\ell} >\roundup{e_\ell}$;} \\
     \fsucc\bigl(\evaldown{e_\ell}\bigr), &
           \text{if $\bar{r}_\ell = \rnear$, otherwise;} \\
     -0, &
           \text{if $\bar{r}_\ell = \rdown$ and $x_\ell = -z_\ell$;} \\
     x_\ell \madd_\rup z_\ell, &
           \text{if $\bar{r}_\ell = \rdown$ and $x_\ell \neq -z_\ell$;} \\
     \fsucc\bigl(\fpred(x_\ell) \madd_\rdown z_\ell\bigr), &
           \text{if $\bar{r}_\ell = \rup$;}
     \end{cases} \\
(a_4, a_5) & =
     \begin{cases}
     (\fsucc(z_\ell), +0), &  \bar{r}_\ell = \rdown; \\
     (z_\ell, -0), & \text{otherwise;}
     \end{cases} \\
a_6 & =
     \begin{cases}
     +\infty, & \bar{r}_\ell = \rdown; \\
     \fsucc(\fmax \madd_\rdown z_\ell), & \bar{r}_\ell = \rup; \\
     \fmax \madd_\rup \left(\ftwiceerrnearpos(\fmax)/2 \madd_\rup z_\ell\right), &
     \text{otherwise.}
     \end{cases}
\end{align*}
\caption{First inverse projection of subtraction: function $\invfirstsubl$.}
\label{fig:first-indirect-projection-subtraction-lower}
\end{figure}

\begin{figure}
\begin{tabular}{L|CCCCCC}
\invfirstsubu(x_u, z_u, \bar{r}_u)
        & -\infty & \Rset_- & -0      & +0      & \Rset_+ & +\infty \\
\hline
-\infty & \uns    & -\infty & -\infty & -\infty & a_9  & \fmax \\
\Rset_- & \uns    & a_7     & x_u     & x_u     & a_7     & \fmax \\
-0      & \uns    & a_8     & a_8     & a_8     & a_8     & \fmax \\
+0      & \uns    & z_u     & +0      & +0      & z_u     & \fmax \\
\Rset_+ & \uns    & a_7     & x_u     & x_u     & a_7     & \fmax \\
+\infty & +\infty & +\infty & +\infty & +\infty & +\infty & +\infty \\
\end{tabular}
\begin{align*}
e_u & \equiv x_u + \ftwiceerrnearpos(x_u)/2 + z_u; \\
a_7 & =
     \begin{cases}
     +0, &
           \text{if $\bar{r}_u = \rnear$, $\ftwiceerrnearpos(x_u) = \fmin$ and
           $x_u = -z_u$;} \\
     x_u \madd_\rdown z_u, &
           \text{if $\bar{r}_u = \rnear$, $\ftwiceerrnearpos(x_u) = \fmin$ and
           $x_u \neq -z_u$;} \\
     \evaldown{e_u}, &
           \text{if $\bar{r}_u = \rnear$, $\feven(x_u)$,
           $\ftwiceerrnearpos(x_u) \neq \fmin$
           and $\evaldown{e_u} = \rounddown{e_u}$;} \\
     \evalup{e_u}, &
           \text{if $\bar{r}_u = \rnear$, $\feven(x_u)$,
           $\ftwiceerrnearpos(x_u) \neq \fmin$
           and $\evaldown{e_u} < \rounddown{e_u}$;} \\
     \fpred\left(\evalup{e_u}\right), &
           \text{if $\bar{r}_u = \rnear$, otherwise;} \\
     \fpred\left(\fsucc(x_u) \madd_\rup z_u\right), &
           \text{if $\bar{r}_u = \rdown$;} \\
     +0, &
           \text{if $\bar{r}_u = \rup$ and $x_u = -z_u$;} \\
     x_u \madd_\rdown z_u, &
           \text{if $\bar{r}_u = \rup$ and $x_u \neq -z_u$;}
     \end{cases} \\
a_8 & =
     \begin{cases}
     z_u,  &
           \text{if $\bar{r}_u = \rdown$;} \\
     \fpred(z_u), &
           \text{otherwise;}
     \end{cases} \\
a_9 & =
     \begin{cases}
     -\infty,  &
           \text{if $\bar{r}_u = \rup$;} \\
     \fpred(z_u \madd_\rup -\fmax),  &
           \text{if $\bar{r}_u = \rdown$;} \\
     -\fmax \madd_\rdown \left(\ftwiceerrnearneg(-\fmax)/2 \madd_\rdown z_u\right), &
           \text{otherwise.}
     \end{cases}
\end{align*}
\caption{First inverse projection of subtraction: function $\invfirstsubu$.}
\label{fig:first-indirect-projection-subtraction-upper}
\end{figure}

\begin{theorem}
\label{teo:first-indirect-projection-subtraction}
\textup{Algorithm~\ref{algo:first-indirect-subtraction}} satisfies its contract.
\end{theorem}

The second inverse projection of subtraction is quite similar to
the case of direct projection of subtraction. Here we define
Algorithm~\ref{algo:second-indirect-subtraction} and functions $\invsecsubl$
and $\invsecsubu$, as defined in
Figures~\ref{fig:second-indirect-projection-subtraction-lower}
and~\ref{fig:second-indirect-projection-subtraction-upper} respectively.

\begin{algorithm}
\caption{Second inverse projection for subtraction constraints.}
\label{algo:second-indirect-subtraction}
\begin{algorithmic}[1]
\REQUIRE $x = y \msub_S z$,
$x \in X = [x_\ell, x_u]$,
$y \in Y = [y_\ell, y_u]$ and
$z \in Z = [z_\ell, z_u]$.
\ENSURE
$Z' \sseq Z$ and
\(
  \forall r \in S, x \in X, y \in Y, z \in Z
    \itc  x = y \msub_r z \implies z \in Z'
\).

\STATE
$\bar{r}_\ell \assign \bar{r}_\ell^r(S, x_u, \msub, y_\ell)$;
$\bar{r}_u \assign \bar{r}_u^r(S, x_\ell, \msub, y_u)$;

\STATE
$z'_\ell \assign \invsecsubl(y_\ell, x_u, \bar{r}_\ell)$;
$z'_u \assign \invsecsubu(y_u, x_\ell, \bar{r}_u)$;
\IF {$z'_\ell \in \Fset$ and $z'_u \in \Fset$}
\STATE
$Z' \assign Z \inters [z'_\ell, z'_u]$;
\ELSE
\STATE
$Z' \assign \emptyset$;
\ENDIF
\end{algorithmic}
\end{algorithm}

\begin{figure}
\begin{tabular}{L|CCCCCC}
\invsecsubl(y_\ell, x_u, \bar{r}_\ell)
        & -\infty & \Rset_- & -0     & +0     & \Rset_+ & +\infty \\
\hline
-\infty & -\fmax  & -\fmax  & -\fmax & -\fmax & -\fmax  & -\infty \\
\Rset_- & a_{13}  & a_{10}  & a_{11} & y_\ell    & a_{10}  & -\infty \\
-0      & +\infty & -x_u    & a_{12} & -0     & -x_u    & -\infty \\
+0      & +\infty & -x_u    & a_{11} & -0     & -x_u    & -\infty \\
\Rset_+ & +\infty & a_{10}  & a_{11} & y_\ell    & a_{10}  & -\infty \\
+\infty & \uns    & \uns    & \uns   & \uns   & \uns    & -\infty \\
\end{tabular}
\begin{align*}
e_\ell & \equiv y_\ell - \left(x_u + \ftwiceerrnearpos(x_u)/2\right); \\
a_{10} & =
       \begin{cases}
       -0, &
             \text{if $\bar{r}_\ell = \rnear$, $\ftwiceerrnearpos(x_u) = \fmin$
             and $x_u = y_\ell$;} \\
       y_\ell \msub_\rup x_u, &
             \text{if $\bar{r}_\ell = \rnear$, $\ftwiceerrnearpos(x_u) = \fmin$
             and $x_u \neq y_\ell$;} \\
       \evalup{e_\ell}, &
             \text{if $\bar{r}_\ell = \rnear$, $\feven(x_u)$,
             $\ftwiceerrnearpos(x_u) \neq \fmin$
             and $\evalup{e_\ell} = \roundup{e_\ell}$;} \\
       \evaldown{e_\ell}, &
             \text{if $\bar{r}_\ell = \rnear$, $\feven(x_u)$,
             $\ftwiceerrnearpos(x_u) \neq \fmin$ and
             $\evalup{e_\ell} >\roundup{e_\ell}$;} \\
       \fsucc\bigl(\evaldown{e_\ell}\bigr), &
             \text{if $\bar{r}_\ell = \rnear$, otherwise;} \\
       -0, &
             \text{if $\bar{r}_\ell = \rup$ and $x_u = y_\ell$;} \\
       y_\ell \msub_\rup x_u, &
           \text{if $\bar{r}_\ell = \rup$ and $x_u \neq y_\ell$;} \\
       \fsucc\left(y_\ell \msub_\rdown \fsucc(x_u)\right), &
           \text{if $\bar{r}_\ell = \rdown$;}
       \end{cases} \\
(a_{11}, a_{12}) & =
       \begin{cases}
       (y_\ell, -0), & \text{if $\bar{r}_\ell = \rdown$;} \\
       (\fsucc(y_\ell), +0),  & \text{otherwise;}
       \end{cases} \\
a_{13} & =
       \begin{cases}
       +\infty,  &
           \text{if $\bar{r}_\ell = \rup$;} \\
       \fsucc(y_\ell \madd_\rdown \fmax), &
           \text{if $\bar{r}_\ell = \rdown$;} \\
       \fmax \madd_\rup \left( \ftwiceerrnearpos(\fmax)/2 \madd_\rup y_\ell\right), &
           \text{otherwise.}
       \end{cases}
\end{align*}
\caption{Second inverse projection of subtraction: function $\invsecsubl$.}
\label{fig:second-indirect-projection-subtraction-lower}
\end{figure}

\begin{figure}
\begin{tabular}{L|CCCCCC}
\invsecsubu(y_u, x_\ell, \bar{r}_u)
       & -\infty & \Rset_- & -0      & +0      & \Rset_+ & +\infty  \\
\hline
-\infty & +\infty & \uns     & \uns    & \uns    & \uns    & \uns    \\
\Rset_- & +\infty & a_{14}    & y_u     & a_{15}   & a_{14}   & -\infty \\
-0      & +\infty & -x_\ell    & +0      & a_{15}   & -x_\ell    & -\infty \\
+0      & +\infty & -x_\ell    & +0      & a_{16}   & -x_\ell    & -\infty  \\
\Rset_+ & +\infty & a_{14}   & y_u     & a_{15}   & a_{14}   & a_{17}   \\
+\infty & +\infty & \fmax   & \fmax   & \fmax   & \fmax   & \fmax   \\
\end{tabular}
\begin{align*}
e_u & \equiv y_u - \left(x_\ell + \ftwiceerrnearneg(x_u)/2\right); \\
a_{14} & =
       \begin{cases}
       +0, &
             \text{if $\bar{r}_u = \rnear$, $\ftwiceerrnearneg(x_\ell) = -\fmin$
             and $x_\ell = y_u$;} \\
       y_u \msub_\rdown x_\ell, &
             \text{if $\bar{r}_u = \rnear$, $\ftwiceerrnearneg(x_\ell) = -\fmin$
             and  $x_\ell \neq y_u$;} \\
       \evaldown{e_u}, &
             \text{if $\bar{r}_u = \rnear$, $\feven(x_\ell)$,
             $\ftwiceerrnearneg(x_\ell) \neq -\fmin$
             and $\evaldown{e_u} = \rounddown{e_u}$;} \\
       \evalup{e_u}, &
             \text{if $\bar{r}_u = \rnear$, $\feven(x_u)$,
             $\ftwiceerrnearneg(x_u) \neq -\fmin$
             and $\evaldown{e_u} < \rounddown{e_u}$;} \\
       \fpred\left(\evalup{e_u}\right), &
             \text{if $\bar{r}_u = \rnear$, otherwise;} \\
       \fpred\left(y_u \msub_\rup \fpred(x_\ell)\right), &
             \text{if $\bar{r}_u = \rup$;} \\
       +0, &
             \text{if $\bar{r}_u = \rdown$ and $x_\ell = y_u$;} \\
       y_u \msub_\rdown x_\ell, &
             \text{if $\bar{r}_u = \rdown$ and $x_\ell \neq y_u$;}
       \end{cases} \\
(a_{15}, a_{16}) & =
       \begin{cases}
       (\fpred(y_u), -0), &
             \text{if $\bar{r}_u = \rdown$;} \\
       (y_u, +0), & \text{otherwise;}
       \end{cases} \\
a_{17} & =
       \begin{cases}
       -\infty, &
             \text{if $\bar{r}_u = \rdown$;} \\
       \fpred(y_u \msub_\rup \fmax), &
             \text{if $\bar{r}_u = \rup$;} \\
       -\fmax \madd_\rdown \left( \ftwiceerrnearneg(-\fmax)/2 \madd_\rdown y_u\right), &
             \text{otherwise.}
       \end{cases}
\end{align*}
\caption{Second inverse projection of subtraction: function $\invsecsubu$.}
\label{fig:second-indirect-projection-subtraction-upper}
\end{figure}

\begin{theorem}
\label{teo:second-indirect-projection-subtraction}
\textup{Algorithm~\ref{algo:second-indirect-subtraction}} is correct.
\end{theorem}

Since subtraction is very closely related to addition,
the proofs of Theorems~\ref{teo:first-indirect-projection-subtraction}
and \ref{teo:second-indirect-projection-subtraction}
can be obtained by reasoning in the same way as for the projections of addition.
Moreover, it is worth noting that in order to obtain more precise results,
inverse projections for subtraction need to be intersected with
maximum ULP filtering \cite{BagnaraCGG16IJOC}, as in the case of addition.

\subsection{Multiplication}

Here we deal with constraints of the form
$x = y \mmul _S z$.
As usual, assume
$X = [x_\ell, x_u]$,
$Y = [y_\ell, y_u]$ and
$Z = [z_\ell, z_u]$.

\paragraph{Direct Propagation.}

For direct propagation, a case analysis is performed in order to select
the interval extrema $y_L$ and $z_L$ (resp., $y_U$ and $z_U$)
to be used to compute the new lower (resp., upper) bound for $x$.

Firstly, whenever $\sgn(y_\ell) \neq \sgn(y_u)$ and $\sgn(z_\ell) \neq \sgn(z_u)$,
there is no unique choice for $y_L$ and $z_L$ (resp., $y_U$ and $z_U$);
therefore we need to compute the two candidate lower (and upper) bounds for $x$
and then choose the minumum (the maximum, resp).

The choice is instead unique in all cases where the signs of one among $y$ and $z$,
or both of them, are constant over the respective intervals.
Function $\sigma$ of Figure~\ref{fig:the-sigma-function}
determines the extrema of $y$ and $z$
useful to compute the new lower (resp., upper) bound
for $y$ when the sign of $z$ is constant.
When the sign of $y$ is constant, the appropriate choice for
the extrema of $y$ and $z$ can be determined by swapping the role of $y$ and $z$
in function $\sigma$.

Once the extrema $(y_L, y_U, z_L, z_U)$ have been selected,
functions $\dirmull$ and $\dirmulu$ of Figure~\ref{fig:direct-projection-multiplication}
are used to find new bounds for $x$.
It is worth noting that it is not necessary to compute new values
of $r_\ell$ and $r_u$ for the application of functions $\dirmull$ and $\dirmulu$
at line~\ref{algo3:dirmul-second-invocation} of Algorithm~\ref{algo3}.
This is true because, by Definition~\ref{def:rounding-mode-selectors},
the choice of $r_\ell$ (of $r_u$, resp.) is driven by the sign of $y_L \mmul z_L$
(of $y_U \mmul z_U$, resp.) only.
Since, in this case, the sign of $y_L \mmul z_L$ (of $y_U \mmul z_U$, resp.)
as defined at line~\ref{algo3:extrema-first-assign}
and the sign of $y_L \mmul z_L$ (of $y_U \mmul z_U$, resp.)
as defined at line~\ref{algo3:extrema-second-assign} are the same,
we do not need to compute $r_\ell$ and $r_u$ another time.

\begin{algorithm}
\caption{Direct projection for multiplication constraints.}
\label{algo3}
\begin{algorithmic}[1]
\REQUIRE $x = y \mmul_S z$,
$x \in X = [x_\ell, x_u]$,
$y \in Y = [y_\ell, y_u]$ and
$z \in Z = [z_\ell, z_u]$.
\ENSURE
$X' \sseq  X$ and
\(
  \forall r \in S, x \in X, y \in Y, z \in Z
    \itc  x = y \mmul_r z \implies x \in X'
\)
and
\(
\forall X''\subset X \itc \exists r \in S, y \in Y, z \in Z \st y  \mmul_r z\not\in X''\).
\IF {$\sgn(y_\ell) \neq \sgn(y_u)$ \AND $\sgn(z_\ell) \neq \sgn(z_u)$}
  \STATE
  $(y_L, y_U, z_L, z_U) \assign (y_\ell, y_\ell, z_u, z_\ell);$
  \label{algo3:extrema-first-assign}
  \STATE
  $r_\ell \assign r_\ell(S, y_L, \mmul, z_L)$;
  $r_u \assign r_u(S, y_U, \mmul, z_U)$;
  \STATE
  $v_\ell \assign \dirmull(y_L, z_L, r_\ell)$;
  $v_u \assign \dirmulu(y_U, z_U, r_u)$;
  \STATE
  $(y_L, y_U, z_L, z_U) \assign (y_u, y_u, z_\ell, z_u)$;
  \label{algo3:extrema-second-assign}
  \STATE
  $w_\ell \assign \dirmull(y_L, z_L, r_\ell)$;
  $w_u \assign \dirmulu(y_U, z_U, r_u)$;
  \label{algo3:dirmul-second-invocation}
  \STATE
  $x'_\ell \assign \min \{ v_\ell, w_\ell \}$;
  $x'_u \assign \max \{ v_u, w_u \}$;
\ELSE
  \IF {$\sgn(y_\ell) = \sgn(y_u)$}
  \STATE
  $(y_L, y_U, z_L, z_U) \assign \sigma(y_\ell, y_u, z_\ell, z_u);$
  \ELSE
  \STATE
  $(z_L, z_U, y_L, y_U) \assign \sigma(z_\ell, z_u, y_\ell, y_u);$
  \ENDIF
  \STATE
  $r_\ell \assign r_\ell(S, y_L, \mmul, z_L)$; $r_u \assign r_u(S, y_U, \mmul, z_U)$;
  \STATE
  $x'_\ell \assign \dirmull(y_L, z_L, r_\ell)$;
  $x'_u \assign \dirmulu(y_U, z_U, r_u)$;
\ENDIF
\STATE
$X' \assign X \inters [x'_\ell, x'_u]$;
\end{algorithmic}
\end{algorithm}

\begin{figure}[p!]
\begin{tabular}{L|CCCCCC}
\dirmull(y_L, z_L)
        & -\infty & \Rset_-                       & -0      & +0      & \Rset_+                       & +\infty \\
\hline
-\infty & +\infty & +\infty                       & +\infty & -0      & -\infty                       & -\infty \\
\Rset_- & +\infty & y_L \mathord{\mmul}_{r_\ell} z_L & +0      & -0      & y_L \mathord{\mmul}_{r_\ell} z_L & -\infty \\
-0      & +\infty & +0                            & +0      & -0      & -0                            & -0      \\
+0      & -0      & -0                            & -0      & +0      & +0                            & +\infty \\
\Rset_+ & -\infty & y_L \mathord{\mmul}_{r_\ell} z_L & -0      & +0      & y_L \mathord{\mmul}_{r_\ell} z_L & +\infty \\
+\infty & -\infty & -\infty                       & -0      & +\infty & +\infty                       & +\infty \\
\end{tabular}

\bigskip

\begin{tabular}{L|CCCCCC}
\dirmulu(y_U, z_U)
        & -\infty & \Rset_-                       & -0      & +0      & \Rset_+                       & +\infty \\
\hline
-\infty & +\infty & +\infty                       & +0      & -\infty & -\infty                       & -\infty \\
\Rset_- & +\infty & y_U \mathord{\mmul}_{r_u} z_U & +0      & -0      & y_U \mathord{\mmul}_{r_u} z_U & -\infty \\
-0      & +0      & +0                            & +0      & -0      & -0                            & -\infty \\
+0      & -\infty & -0                            & -0      & +0      & +0                            & +0      \\
\Rset_+ & -\infty & y_U \mathord{\mmul}_{r_u} z_U & -0      & +0      & y_U \mathord{\mmul}_{r_u} z_U & +\infty \\
+\infty & -\infty & -\infty                       & -\infty & +0      & +\infty                       & +\infty \\
\end{tabular}
\caption{Direct projection of multiplication: functions $\dirmull$ and $\dirmulu$.}
\label{fig:direct-projection-multiplication}
\end{figure}

\begin{theorem}
\label{teo:direct-projection-mult}
\textup{Algorithm~\ref{algo3}} satisfies its contract.
\end{theorem}

\paragraph{Inverse Propagation.}

For inverse propagation, Algorithm~\ref{algo4}
partitions interval $Z$ into the sign-homogeneous intervals
$Z_- \defeq Z \inters [-\infty, -0]$ and
$Z_+ \defeq Z \inters [+0, +\infty]$.
This is done because the sign of $Z$ must be taken into account
in order to derive correct bounds for $Y$.
Hence, once $Z$ has been partitioned into sign-homogeneous intervals,
we use intervals $X$ and $Z_-$ to obtain
interval $[y^-_\ell, y^-_u]$, and $X$ and $Z_+$ to obtain $[y^+_\ell, y^+_u]$.
To do so, the algorithm determines the appropriate extrema
of intervals $X$ and $W = Z_-$ or $W = Z_+$
to be used for constraint propagation.
To this aim, function $\tau$ of Figure~\ref{fig:the-tau-function} is employed;
note that the sign of $W$ is, by construction, constant over the interval.
The chosen extrema are then passed as parameters to functions
$\invmull$ of Figure~\ref{fig:indirect-projection-multiplication-yl}
and $\invmulu$ of Figure~\ref{fig:indirect-projection-multiplication-yu},
that compute the new, refined bounds for $y$,
by using the inverse operation of multiplication, i.e., division.
The so obtained intervals $Y \inters [y^-_\ell, y^-_u]$ and
$Y \inters [y^+_\ell, y^+_u]$ will be then joined with convex union,
denoted by $\biguplus$, to obtain $Y'$.

\begin{algorithm}
\caption{Inverse projection for multiplication constraints.}
\label{algo4}
\begin{algorithmic}[1]
\REQUIRE $x = y \mmul_S z$,
$x \in X = [x_\ell, x_u]$,
$y \in Y = [y_\ell, y_u]$ and
$z \in Z = [z_\ell, z_u]$.
\ENSURE
$Y' \sseq Y$ and
\(
  \forall r \in S, x \in X, y \in Y, z \in Z
    \itc  x = y \mmul_r z \implies y \in Y'
\).
\STATE
$Z_-\assign Z \inters [-\infty, -0];$
\IF {$Z_- \neq \emptyset$}
  \STATE $W \assign Z_- $;
  \STATE
  $(x_L, x_U,w_L, w_U)\assign \tau(x_\ell, x_u, w_\ell, w_u)$;
  \STATE
  $\bar{r}_\ell \assign \bar{r}_\ell^\ell(S, x_L, \mmul, w_L)$;
  $\bar{r}_u \assign \bar{r}_u^\ell(S, x_U, \mmul, w_U)$;
  \STATE
  $y^-_\ell \assign \invmull(x_L, w_L, \bar{r}_\ell)$;
  $y^-_u \assign \invmulu(x_U, w_U, \bar{r}_u)$;
  \IF {$y^-_\ell \in \Fset$ and $y^-_u\in \Fset$}
    \STATE
   $Y'_-= Y \inters[y^-_\ell, y^-_u];$
  \ELSE
    \STATE
    $Y'_- = \emptyset;$
  \ENDIF
\ELSE
  \STATE
  $Y'_- = \emptyset;$
\ENDIF
\STATE
$Z_+ \assign Z \inters [+0, +\infty];$
\IF {$Z_+ \neq \emptyset$ }
  \STATE $W \assign Z_+ $;
  \STATE
  $(x_L, x_U,w_L, w_U) \assign \tau(x_\ell, x_u, w_\ell, w_u)$;
  \STATE
  $\bar{r}_\ell \assign \bar{r}_\ell^\ell(S, x_L, \mmul, w_L)$;
  $\bar{r}_u \assign \bar{r}_u^\ell(S, x_U, \mmul, w_U)$;
  \STATE
  $y^+_\ell \assign \invmull(x_L, w_L, \bar{r}_\ell)$;
  $y^+_u \assign \invmulu(x_U, w_U, \bar{r}_u)$;
  \IF {$y^+_\ell \in \Fset$ and $y^+_u\in \Fset$}
    \STATE
    $Y'_+= Y \inters[y^+_\ell, y^+_u];$
  \ELSE
    \STATE
    $Y'_+ = \emptyset;$
  \ENDIF
\ELSE
  \STATE
  $Y'_+ = \emptyset;$
\ENDIF
\STATE
$Y' \assign Y'_- \biguplus Y'_+$;
\end{algorithmic}
\end{algorithm}

\begin{figure}[p!]
\begin{tabular}{L|CCCCCC}
\invmull(x_L, w_L)
       & -\infty & \Rset_- & -0      & +0      & \Rset_+ & +\infty      \\
\hline
-\infty & \fmin   & a_4     & \uns   & -\infty & -\infty  & -\infty   \\
\Rset_- & \fmin   &a^-_3       &\uns   & -\fmax  &a^+_3       & \fmin     \\
-0      & +0      & +0       & +0     & -\fmax  & a_5    & \fmin     \\
+0      & \fmin   & a_6      & -\fmax & +0      & +0       & +0       \\
\Rset_+ & \fmin   & a^-_3      & -\fmax &\uns    & a^+_3      & \fmin    \\
+\infty & -\infty & -\infty  & -\infty & \uns   & a_7  & \fmin     \\
\end{tabular}
\begin{align*}
e^+_\ell & \equiv (x_L+ \ftwiceerrnearneg(x_L)/2) / w_L;\\
a_3^+
   &=
     \begin{cases}
       \evalup{e^+_\ell},
       &\text{if   $\bar{r}_\ell = \rnear$,   $\feven(x_L)$
              and $\evalup{e^+_\ell} = \roundup{e^+_\ell}$;} \\
       \evaldown{e^+_\ell},
       &\text{if  $\bar{r}_\ell = \rnear$,  $\feven(x_L)$  and $\evalup{e^+_\ell} >\roundup{e^+_\ell}$;} \\
       \fsucc\bigl(\evaldown{e^+_\ell}\bigr),
       &\text{if  $\bar{r}_\ell = \rnear$, otherwise;} \\
       x_L \mdiv_\rup w_L,
       &\text{if   $\bar{r}_\ell = \rdown$;} \\
       \fsucc\bigl(\fpred(x_L) \mdiv_\rdown w_L\bigr),
       &\text{if  $\bar{r}_\ell = \rup$;}\\
     \end{cases} \\
e^-_\ell & \equiv (x_L+ \ftwiceerrnearpos(x_L)/2) / w_L; \\
a_3^- &=
   \begin{cases}
    \evalup{e^-_\ell},
       &\text{if  $\bar{r}_\ell = \rnear$, $\feven(x_L)$
              and $\evalup{e^+_\ell} = \roundup{e^+_\ell}$;} \\
              \evaldown{e^-_\ell},
       &\text{if  $\bar{r}_\ell = \rnear$, $\feven(x_L)$ and $\evalup{e^-_\ell} > \roundup{e^-_\ell}$;} \\
     \fsucc\bigl(\evaldown{e^-_\ell}\bigr),
       &\text{if  $\bar{r}_\ell = \rnear$, otherwise;} \\
     x_L \mdiv_\rup w_L,
       &\text{if  $\bar{r}_\ell = \rup$;} \\
     \fsucc\bigl(\fsucc(x_L) \mdiv_\rdown w_L\bigr),
       &\text{if  $\bar{r}_\ell = \rdown$;}
   \end{cases} \\
 e^1_\ell & \equiv (-\fmax+ \ftwiceerrnearneg(-\fmax)/2) / w_L; \\
 a_4
   &=
   \begin{cases}
     +\infty,
       &\text{if $\bar{r}_\ell = \rup$;} \\
     \fsucc(-\fmax \mdiv_\rdown w_L),
       &\text{if $\bar{r}_\ell = \rdown$;} \\
     \evalup{e^1_\ell},
       &\text{if $\bar{r}_\ell = \rnear$ and $\evalup{e^1_\ell} = \roundup{e^1_\ell}$;} \\
    \evaldown{e^1_\ell},
       &\text{if $\bar{r}_\ell = \rnear$, otherwise;} \\
   \end{cases} \\
(a_5, a_6)
   &=
   \begin{cases}
     (-0, \; \fsucc(\fmin \mdiv_\rdown w_L)),
       &\text{if $\bar{r}_\ell = \rdown$;} \\
     (\fsucc(-\fmax \mdiv_\rdown w_L),\;-0),
       &\text{if $\bar{r}_\ell = \rup$;} \\
    (-\fmin \mdiv_\rup (2 \cdot w_L), \; \fmin \mdiv_\rup (2 \cdot w_L)),
       &\text{if $\bar{r}_\ell = \rnear$;} \\
   \end{cases} \\
e^2_\ell &  \equiv (\fmax + \ftwiceerrnearpos(\fmax)/2) / w_L;\\
a_7
   &=
   \begin{cases}
     +\infty,
       &\text{if $\bar{r}_\ell = \rdown$;} \\
     \fsucc(\fmax \mdiv_\rdown w_L),
       &\text{if $\bar{r}_\ell = \rup$;} \\
     \evalup{e^2_\ell},
       &\text{if $\bar{r}_\ell = \rnear$ and $\evalup{e^2_\ell} = \roundup{e^2_\ell}$;} \\
     \evaldown{e^2_\ell},
       &\text{if $\bar{r}_\ell = \rnear$, otherwise.} \\
   \end{cases}
\end{align*}
\caption{Inverse projection of multiplication: function $\invmull$.}
\label{fig:indirect-projection-multiplication-yl}
\end{figure}

\begin{figure}[p!]

\bigskip
\begin{tabular}{L|CCCCCC}
\invmulu(x_U, w_U)
       & -\infty & \Rset_- & -0      & +0      & \Rset_+ & +\infty      \\
\hline
-\infty & +\infty & +\infty  & +\infty & \uns    & a_9     & -\fmin     \\
\Rset_- & -\fmin  & a^-_8      & \fmax   & \uns    & a^+_8        & -\fmin    \\
-0      & -\fmin  & a_{10}    & \fmax   & -0      & -0        & -0        \\
+0      & -0      & -0       & -0      & \fmax   & a_{11}      & -\fmin    \\
\Rset_+ & -\fmin  & a^-_8       & \uns   & \fmax   & a^+_8        & -\fmin    \\
+\infty & -\fmin  & a_{12}       & \uns    & +\infty & +\infty   & +\infty   \\
\end{tabular}
\begin{align*}
e^+_u &\equiv (x_U+ \ftwiceerrnearpos(x_U)/2) / w_U; \\
a^+_8 &=
     \begin{cases}
       \evaldown{e^+_u},
       &\text{if $\bar{r}_u = \rnear$, $\feven(x_U)$
              and $\evalup{e^+_u} = \roundup{e^+_u}$;} \\
       \evalup{e^+_u},
       &\text{if $\bar{r}_u = \rnear$, $\feven(x_U)$  and $\evalup{e^+_u} >\roundup{e^+_u}$;} \\
       \fpred\bigl(\evalup{e^+_u}\bigr),
       &\text{if $\bar{r}_u = \rnear$, otherwise;} \\
       \fpred\bigl(\fsucc(x_U) \mdiv_\rup w_U\bigr),
       &\text{if $\bar{r}_u = \rdown$;} \\
       x_U \mdiv_\rdown w_U,
       &\text{if $\bar{r}_u = \rup$;}
     \end{cases} \\
e^-_u &\equiv (x_U+ \ftwiceerrnearneg(x_U)/2) / w_U; \\
a^-_8 &=
     \begin{cases}
         \evaldown{e^-_u},
       &\text{if $\bar{r}_u = \rnear$, $\feven(x_U)$
              and $\evalup{e^-_u} = \roundup{e^-_u}$;} \\
         \evalup{e^-_u},
       &\text{if $\bar{r}_u = \rnear$, $\feven(x_U)$  and $\evalup{e^-_u} >\roundup{e^-_u}$;} \\
         \fpred\bigl(\evalup{e^-_u}\bigr),
       &\text{if $\bar{r}_u = \rnear$, otherwise;} \\
         \fpred\bigl(\fpred(x_U) \mdiv_\rup w_U\bigr),
       &\text{if $\bar{r}_u = \rup$;} \\
         x_U \mdiv_\rdown w_U,
       &\text{if $\bar{r}_u = \rdown$;}
     \end{cases} \\
e^1_u &\equiv (-\fmax + \ftwiceerrnearneg(-\fmax)/2)/ w_U;\\
a_9
   &=
     \begin{cases}
       -\infty, &\text{if $\bar{r}_u = \rup$;} \\
       \fpred(-\fmax \mdiv_\rup w_U), &\text{if $\bar{r}_u = \rdown$;} \\
       \evaldown{e^1_u}
           &\text{if $\bar{r}_u = \rnear$ and $\evaldown{e^1_u } = \rounddown{e^1_u}$}; \\
       \evalup{e^1_u}  &\text{if $\bar{r}_u = \rnear$, otherwise;} \\
     \end{cases} \\
(a_{10},a_{11})
   &=
     \begin{cases}
       (+0,\; \fpred(\fmin \mdiv_\rup w_U)),
       &\text{if $\bar{r}_u = \rdown$;} \\
       (\fpred(-\fmin \mdiv_\rup w_U),\;+0),  &\text{if $\bar{r}_u = \rup$;} \\
       (-\fmin \mdiv_\rdown (2 \cdot w_U),\; \fmin\mdiv_\rdown (2 \cdot w_U)),
       &\text{if $\bar{r}_u = \rnear$;} \\
     \end{cases} \\
e^2_u &\equiv (\fmax + \ftwiceerrnearpos(\fmax)/2)/ w_U;\\
a_{12}
   &=
     \begin{cases}
       -\infty, &\text{if $\bar{r}_u = \rdown$;} \\
       \fpred(\fmax \mdiv_\rup w_U), &\text{if $\bar{r}_u = \rup$;} \\
       \evaldown{e^2_u}
       &\text{if $\bar{r}_u = \rnear$ and $\evaldown{e^2_u }= \rounddown{e^2_u}$}; \\
       \evalup{e^2_u}  &\text{if $\bar{r}_u = \rnear$, otherwise.} \\
     \end{cases} \\
\end{align*}
\caption{Inverse projection of multiplication: function $\invmulu$.}
\label{fig:indirect-projection-multiplication-yu}
\end{figure}

\begin{theorem}
\label{teo:indirect-projection-mult}
\textup{Algorithm~\ref{algo4}} satisfies its contract.
\end{theorem}

Of course, the refinement $Z'$ of $Z$ can be defined analogously.



\clearpage

\section{Proofs of Results}\label{proofs}
 \subsection{Proofs of Results in Section~\ref{sec:preliminaries}}

\begin{delayedproof}[of Proposition~\ref{prop:round-properties}]
In order to prove \eqref{round-properties:first},
we first prove that $\rounddown{x} \leq x$.
To this aim, consider the following cases on $x \in \Rset \setdiff \{0\}$:
\begin{description}

\item[$-\fmax \leq x < 0   \;\lor\; \fmin \leq x:$]
by~\eqref{eq:round-towards-minus-infinity} we have
$\rounddown{x} = \max \{\, z \in \Fset \mid z \leq x \,\}$,
hence $\rounddown{x}\leq x$;

\item[$0 < x < \fmin:$]
by~\eqref{eq:round-towards-minus-infinity}
we have $\rounddown{x} = -0\leq x$;

\item[ $x < -\fmax:$]
by~\eqref{eq:round-towards-minus-infinity}
we have $\rounddown{x} = -\infty \leq x$.
\end{description}

We now prove that $x \leq \roundup{x}$. Consider the following cases on
$x \in \Rset \setdiff \{ 0 \}$:
\begin{description}

\item[$x > \fmax:$]
by~\eqref{eq:round-towards-plus-infinity} we have
$\roundup{x}= +\infty$ and thus $x \leq \roundup{x}$ holds;

\item[$x \leq -\fmin \;\lor\; 0 < x \leq \fmax:$]
by~\eqref{eq:round-towards-plus-infinity} we have
$\roundup{x} = \min \{\, z \in \Fset \mid z \geq x \,\}$,
hence $ x \leq \roundup{x}$ holds;

\item[$-\fmin< x < 0:$]
by~\eqref{eq:round-towards-plus-infinity} we have $\roundup{x} = -0$
hence $x \leq \roundup{x}$ holds.
\end{description}

In order to prove \eqref{round-properties:second},
consider the following cases on $x \in \Rset \setdiff \{ 0 \}$:
\begin{description}

\item[$x > 0:$]
by~\eqref{eq:round-towards-zero} we have
$\roundzero{x} = \rounddown{x} \leq \roundup{x}$;

\item[$x <0:$]
by~\eqref{eq:round-towards-zero} we have
$\rounddown{x} \leq \roundup{x} = \roundzero{x}$.
\end{description}

 In order to prove \eqref{round-properties:third},
consider the following cases on $x \in \Rset \setdiff \{ 0 \}$:
\begin{description}
\item[$ -\fmax\leq x\leq  \fmax:$]
we have the following cases
\begin{description}
\item[$\bigl|\rounddown{x} - x\bigr|
                   < \bigl|\roundup{x}   - x\bigr| \vee \bigl(\bigl|\rounddown{x} - x\bigr|
                   = \bigl|\roundup{x}   - x\bigr|\bigr)
            \wedge \feven\bigl(\rounddown{x}\bigr)$:]
by~\eqref{eq:round-to-nearest}, we have
$\roundnear{x} = \rounddown{x} \leq \roundup{x}$;

\item[$\bigl|\rounddown{x} - x\bigr|
                   > \bigl|\roundup{x}   - x\bigr|\vee \bigl(\bigl|\rounddown{x} - x\bigr|
                   = \bigl|\roundup{x}   - x\bigr|
                 \wedge \neg \feven\bigl(\rounddown{x}\bigr)\bigr)$:]
          by~\eqref{eq:round-to-nearest} we have $\rounddown{x} \leq \roundup{x} = \roundnear{x}$.
\end{description}
\item[$ -\fmax> x:$]
we have the following cases
\begin{description}
\item[$ -2^{\emax}(2 - 2^{-p})<x< -\fmax$:]
by~\eqref{eq:round-to-nearest} we have $\rounddown{x} \leq \roundup{x} = \roundnear{x}$.

\item[$ x\leq  -2^{\emax}(2 - 2^{-p})$:]
by~\eqref{eq:round-to-nearest} we have
$\roundnear{x} = \rounddown{x} \leq \roundup{x}$;
\end{description}

\item[$ \fmax< x:$]
we have the following cases
\begin{description}
\item[$ 2^{\emax}(2 - 2^{-p})>x>\fmax$:]
               by~\eqref{eq:round-to-nearest} we have
$\roundnear{x} = \rounddown{x} \leq \roundup{x}$;

\item[$ x\geq  2^{\emax}(2 - 2^{-p})$:]
by~\eqref{eq:round-to-nearest} we have $\rounddown{x} \leq \roundup{x} = \roundnear{x}$.
\end{description}

\end{description}

In order to prove \eqref{round-properties:fourth},
let us compute $-\roundup{-x}$.
There are the following cases:
\begin{description}

\item[$-x > \fmax:$]
this implies that $x < -\fmax$
and, by~\eqref{eq:round-towards-plus-infinity},
$\roundup{-x} = +\infty$;
hence, by~\eqref{eq:round-towards-minus-infinity},
$-\roundup{-x} = -\infty = \rounddown{x}$;

\item[$-x \leq -\fmin\;\lor\;0 < -x \leq \fmax:$]
this implies that $x \geq \fmin  \;\lor\;  -\fmax \geq  x > 0$
and, by~\eqref{eq:round-towards-plus-infinity}, we have
$\roundup{-x} = \min \{\, z \in \Fset \mid z \geq -x \,\}$;
therefore, by~\eqref{eq:round-towards-minus-infinity},
\(
  -\roundup{-x}
    = -\min \{\, z \in \Fset \mid z \geq -x \,\}
    =  \max \{\, z \in \Fset \mid z \leq x \,\}
    = \rounddown{x}
\),

\item[$-\fmin < -x < 0:$]
this implies that $0 < x < \fmin$ and,
by~\eqref{eq:round-towards-plus-infinity}, $\roundup{-x}= -0$;
hence, by~\eqref{eq:round-towards-minus-infinity},
$-\roundup{-x} = +0 = \rounddown{x}$.
\end{description}
\end{delayedproof}

\subsection{Proofs of Results in Section~\ref{sec:rounding-modes-and-rounding-errors}}
\label{sec:proofs-rounding-modes-and-rounding-errors}

\begin{delayedproof}[Rest of the proof of Proposition~\ref{prop:worst-case-rounding-modes}]
We prove the second part of
Proposition~\ref{prop:worst-case-rounding-modes},
regarding rounding mode selectors for inverse propagators.
Before doing so, we need to prove the following result.
Let
\(
  \mathord{\mop}
   \in
     \{ \mathord{\madd}, \mathord{\msub}, \mathord{\mmul}, \mathord{\mdiv} \},
\)
and let $\mathrm{r}$ and $\mathrm{s}$ be two IEEE~754 rounding modes,
such that for any $a, b \in \Fset$,
\[
a \mop_r b \sleq a \mop_s b.
\]
Moreover, let $x, z \in \Fset$,
and let $\bar{y}_\mathrm{s}$ be the minimum
$y_\mathrm{s} \in \Fset$ such that
$x = y_\mathrm{s} \mop_\mathrm{s} z$.
Then, for any $y_\mathrm{r} \in \Fset$ such that
$x = y_\mathrm{r} \mop_\mathrm{r} z$ we have
\begin{align*}
  \bar{y}_\mathrm{s} \mop_\mathrm{r} z
  &\sleq
  \bar{y}_\mathrm{s} \mop_\mathrm{s} z \\
  &=
  x \\
  &=
  y_\mathrm{r} \mop_\mathrm{r} z.
\end{align*}
This leads us to write
\[
  [
    \bar{y}_\mathrm{s} \circ z
  ]_\mathrm{r}
  \sleq
  [
    y_\mathrm{r} \circ z
  ]_\mathrm{r}
\]
which, due to the isotonicity of all IEEE~754 rounding modes,
implies
\[
  \bar{y}_\mathrm{s} \circ z
  \sleq
  y_\mathrm{r} \circ z.
\]
Finally, if operator `$\circ$' is isotone we have
\[
  \bar{y}_\mathrm{s}
  \sleq
  y_\mathrm{r},
\]
which implies that $\bar{y}_\mathrm{s}$
is the minimum $y \in \Fset$ such that
$x = y \mop_r z$ or $x = y \mop_s z$.
On the other hand, if `$\circ$' is antitone we have
\[
  \bar{y}_\mathrm{s}
  \sgeq
  y_\mathrm{r},
\]
and $\bar{y}_\mathrm{s}$
is the maximum $y \in \Fset$ such that
$x = y \mop_r z$ or $x = y \mop_s z$.
An analogous result can be proved regarding
the upper bound for $y$ in case the operator is isotone,
and regarding the lower bound for $y$ in case it is antitone.

The above claim allows us to prove the following.
Assume first that $\mop$ is isotone with respect to $y$ in $x=y\mop z$. Let
$\hat{y}_{\rup}$
be the minimum $y_{\rup} \in \Fset$ such that
$x = y_{\rup} \mop_\rup z = \roundup{y_{\uparrow} \mop z}$, let
$\hat{y}_\rnear$
be the minimum $y_\rnear \in \Fset$ such that
$x = y_{\rnear} \mop_\rnear z = \roundnear{y_\rnear \mop z}$ and, finally,
let
$\hat{y}_\rdown$
be the minimum $y_\rdown \in \Fset$ such that
$x = y_{\rdown} \mop_\rdown z = \rounddown{y_\rdown \mop z}$.
We will prove that
\[
  \hat{y}_\rup \sleq \hat{y}_\rnear \sleq \hat{y}_\rdown.
\]
Since we assumed that $\mop$ is isotone with respect to $y$ in $x = y \mop z$,
the rounding mode that gives the minimal $y$ solution of $x = [y \mop z]_\mathrm{r}$
is the one that yields a bigger (w.r.t. $\sleq$ order) floating point number,
as we proved before.
We must now separately treat the following cases:
\begin{description}
\item[$y\mop z\not=0:$]
  By~\eqref{round-properties:third}, we have
  $\rounddown{y \mop z} \leq \roundnear{y \mop z} \leq \roundup{y \mop z}$.
  Since in this case
  $y \mop z \not= 0$, we have that
  $\rounddown{y \mop z} \sleq \roundnear{y \mop z} \sleq \roundup{y \mop z}$.
  This implies $\hat{y}_\rup \sleq \hat{y}_\rnear \sleq \hat{y}_\rdown$.
\item[$y\mop z=0:$]
  In this case, $\rounddown{y \mop z} \sleq \roundnear{y \mop z} = \roundup{y \mop z}$.
  This implies $\hat{y}_\rup \sleq \hat{y}_\rnear \sleq \hat{y}_\rdown$.
\end{description}
Moreover, let
$\tilde{y}_{\rup}$
be the maximum $y_{\rup} \in \Fset$ such that
$x = y_{\rup} \mop_\rup z = \roundup{y_{\uparrow} \mop z}$, let
$\tilde{y}_\rnear$
be the maximum $y_\rnear \in \Fset$ such that
$x = y_{\rnear} \mop_\rnear z = \roundnear{y_\rnear \mop z}$ and, finally,
let
$\tilde{y}_\rdown$
be the maximum $y_\rdown \in \Fset$ such that
$x = y_{\rdown } \mop_\rdown z = \rounddown{y_\rdown \mop z}$.
We will prove the fact that
\[
  \tilde{y}_\rup \sleq  \tilde{y}_\rnear \sleq \tilde{y}_\rdown.
\]
Since we assumed that $\mop$ is isotone with respect to $y$ in $x=y\mop z$,
the rounding mode that gives a maximum $y$ solution of $x = y \mop z_\mathrm{r}$
is the one that gives a smaller (w.r.t. $\sleq$ order) floating point number.
We must now deal with the following cases:
\begin{description}
\item[$y\mop z\not=0:$]
  By~\eqref{round-properties:third}, we have
  $\rounddown{y \mop z} \leq \roundnear{y \mop z} \leq \roundup{y \mop z}$.
  Since in this case $y \mop z \not= 0$,
  we have $\rounddown{y \mop z} \sleq \roundnear{y \mop z} \sleq \roundup{y \mop z}$.
  This implies $\tilde{y}_\rup \sleq  \tilde{y}_\rnear \sleq \tilde{y}_\rdown$.
\item[$y\mop z=0:$]
  In this case $\rounddown{y \mop z} \sleq \roundnear{y \mop z} = \roundup{y \mop z}$.
  This implies $\tilde{y}_\rup \sleq \tilde{y}_\rnear \sleq \tilde{y}_\rdown$.
\end{description}
The inequalities
$\hat{y}_\rup \sleq \hat{y}_\rnear \sleq \hat{y}_\rdown$ and
$\tilde{y}_\rup \sleq \tilde{y}_\rnear \sleq \tilde{y}_\rdown$
allow us to claim that the rounding mode selectors
$\hat{r}_\ell(S, \mop, b)$ and $\hat{r}_u(S, b)$
are correct when $\mop$ is isotone with respect to $y$.
In a similar way it is possible to prove that,
in case $\mop$ is antitone with respect to argument $y$,
the above-mentioned rounding mode selectors can be exchanged:
$\hat{r}_u(S, b)$ can be used to obtain the lower bound for $y$,
while $\hat{r}_\ell(S, \mop, b)$ can be used to obtain the upper bound.

Note that, in general, the roundTowardZero rounding mode is equivalent
to roundTowardPositive if the result of the rounded operation is negative,
and to roundTowardNegative if it is positive.
The only case in which this is not true is when the result is $+0$
and the operation is a sum or a subtraction:
this value can come from the rounding toward negative infinity
of a strictly positive exact result,
or the sum of $+0$ and $-0$, which behaves like roundTowardPositive,
yielding $+0$.
This case must be treated separately,
and it is significant only in $\hat{r}_\ell(S, \mop, b)$,
which is used when seeking for the lowest possible value of the variable to be refined
that yields $+0$.

Definition~\ref{def:rounding-mode-selectors-inverse} also contains
selectors that can choose between rounding mode selectors
$\hat{r}_\ell(S, b)$ and $\hat{r}_u(S, b)$
by distinguishing whether the operator is isotone or antitone
with respect to the operand $y$ to be derived by propagation;
they take the result of the operation $b$ and the known operand $a$
into account.
In particular,
$\bar{r}_\ell^\ell(S, b, \mop, a)$,
$\bar{r}_u^\ell(S, b, \mop, a)$
choose the appropriate selector for the leftmost operand, and
$\bar{r}_\ell^r(S, b, \mop, a)$,
$\bar{r}_u^r(S, b, \mop, a)$
are valid for the rightmost one.
\end{delayedproof}

\begin{delayedproof}[of Proposition~\ref{prop:min-max-x-ftwiceerrnearneg-ftwiceerrnearpos}]
We first  prove~\eqref{eq:min-x-plus-ftwiceerrnearneg}.
By Definition~\ref{def:rounding-error-functions},
we have the following cases:
\begin{description}

\item[$x_\ell = -\fmax$:]
Then,
\begin{align*}
x_\ell + \ftwiceerrnearneg(x_\ell)/2
&= -\fmax + \bigl(-\fmax - \fsucc(-\fmax)\bigr) / 2 \\
&= -2^{\emax} (2 - 2^{1-p}) + \bigl(-2^{\emax} (2 - 2^{1-p}) + 2^{\emax} (2 - 2^{1-p}- 2^{1-p}) \bigr)/2 \\
&= -2^{\emax} (2 - 2^{1-p} +1 - 2^{-p} -1+ 2^{1-p}) \\
&= -2^{\emax}(2- 2^{-p})
\end{align*}

On the other hand, consider any $x$ such that $x_\ell < x \leq x_u$. Since  $x \in \Fset$,
this implies that $\fsucc(-\fmax) \leq x \leq x_u$.
In this case
\[
  x + \ftwiceerrnearneg(x) / 2 = (x + \fpred(x)) / 2.
\]
Since `$\fpred$' is monotone, the minimum can be found when $x = \fsucc(-\fmax)$.
In this case, we have that
\begin{align*}
 (x + \fpred(x))/2
&=  (\fsucc(-\fmax) - \fmax)/2\\
&=\bigl(-2^{\emax} (2 -  2^{1-p}- 2^{1-p})  - 2^{\emax} (2-  2^{1-p})\bigr)/2\\
&=  \bigl(- 2^{\emax}(2 -2^{1-p}- 2^{1-p}+2- 2^{1-p})\bigr)/2\\
&= -2^{\emax}(2- 3 \cdot 2^{-p})\\
&>   x_\ell +  \ftwiceerrnearneg(x_\ell)/2\\
& =  -2^{\emax}(2- 2^{-p}).
\end{align*}
Hence we can conclude that
\(\min_{x_\ell \leq x \leq x_u} \bigl(x + \ftwiceerrnearneg(x)/2\bigr)
    = x_\ell + \ftwiceerrnearneg(x_\ell)/2\).

\item[$x_\ell > -\fmax$:] In this case
\[
  x + \ftwiceerrnearneg(x) / 2 = (x + \fpred(x)) / 2.
\]
Since `$\fpred$' is monotone,
\(
\min_{x_\ell \leq x \leq x_u} \bigl(x + \ftwiceerrnearneg(x_\ell)/2\bigr)
    = x_\ell + \ftwiceerrnearneg(x_\ell)/2
\).
\end{description}

We now prove~\eqref{eq:max-x-plus-ftwiceerrnearpos}.
By Definition~\ref{def:rounding-error-functions},
we have the following cases:
\begin{description}

\item[$x_u = \fmax$:]
Then,
\begin{align*}
x_u +  \ftwiceerrnearpos(x_u)/2
  &= \fmax + \bigl(\fmax - \fpred(\fmax)\bigr) / 2   \\
  &= 2^{\emax} (2 - 2^{1-p}) + \bigl(2^{\emax} (2 - 2^{1-p}) - 2^{\emax} (2 - 2^{1-p} - 2^{1-p})\bigr) / 2 \\
  &= 2^{\emax} (2 - 2^{1-p} + 1 - 2^{-p} - 1 + 2^{1-p}) \\
  &= 2^{\emax}(2 - 2^{-p}).
\end{align*}
Now, consider any $x$ such that $x_\ell \leq x < x_u$.
Since $x \in \Fset$, this implies that $x_\ell \leq x\leq \fpred(\fmax)$.
In this case
\[
  x + \ftwiceerrnearpos(x)/2 = (x + \fsucc(x))/2.
\]
Since `$\fsucc$' is monotone, the maximum can be found when $x= \fpred(\fmax)$.
In this case, we have that
\begin{align*}
 (x + \fsucc(x))/2
&=  (\fpred(\fmax) + \fmax)/2\\
&=\bigl(2^{\emax} (2 -  2^{1-p}- 2^{1-p})  + 2^{\emax} (2-  2^{1-p})\bigr)/2\\
&=  \bigl( 2^{\emax}(2 - 2^{1-p}- 2^{1-p}+2- 2^{1-p})\bigr)/2\\
&= 2^{\emax}(2- 3\cdot 2^{-p})\\
&>  x_u +  \ftwiceerrnearpos(x_u)/2\\
&=  2^{\emax}(2- 2^{-p}).
\end{align*}

Hence we can conclude that
\(\max_{x_\ell \leq x \leq x_u} \bigl(x + \ftwiceerrnearpos(x)/2\bigr)
    = x_u + \ftwiceerrnearpos(x_u)/2.
\)

\item[$x_u< \fmax$:] In this case
\[
  x + \ftwiceerrnearpos(x)/2 = (x + \fsucc(x))/2.
\]
Since `$\fsucc$' is monotone,
\(\max_{x_\ell \leq x \leq x_u} \bigl(x + \ftwiceerrnearpos(x)/2\bigr)
    = x_u + \ftwiceerrnearpos(x_u)/2.
\)
\end{description}
\end{delayedproof}

We now introduce and prove Proposition~\ref{prop:rounding-error-functions},
which contains properties of the rounding error functions
that are only needed in the proof of Proposition~\ref{prop:real-approx-of-fp-constraints}.
\begin{proposition}
\label{prop:rounding-error-functions}
For each $r \in \Rset \setdiff \{ 0 \}$ we have
\begin{align}
\label{eq:rounding-error-functions-down}
  0
    &\leq
      r - \rounddown{r}
        <
          \ferrdown\bigl(\rounddown{r}\bigr) \\
\label{eq:rounding-error-functions-up}
  \ferrup\bigl(\rounddown{r}\bigr)
    &<
      r - \roundup{r}
        \leq
          0            \\
\label{eq:rounding-error-functions-near}
  \ftwiceerrnearneg\bigl(\roundnear{r}\bigr)/2
    &\leq
      r - \roundnear{r}
        \leq
          \ftwiceerrnearpos\bigl(\roundnear{r}\bigr)/2,
\end{align}
where the two inequalities of~\eqref{eq:rounding-error-functions-near}
are strict if $\roundnear{r}$ is odd.
\end{proposition}
\begin{proof}
Suppose $r \in \Rset$ was rounded down to $x \in \Fset$.  Then the
error that was committed, $r - x$, is a nonnegative extended real that
is strictly bounded from above by $\ferrdown(x) = \fsucc(x) - x$, that
is, $0 \leq r - x < \fsucc(x) - x$, for otherwise we would have
$r \geq \fsucc(x)$ or $r < x$ and, in both cases $r$ would not have
been rounded down to $x$.
Note that $\ferrdown(\fmax) = +\infty$, coherently with the fact that
the error is unbounded from above in this case.

Dually, if $r \in \Rset$ was rounded up to $x \in \Fset$
the error that was committed, $r - x$, is a nonpositive extended
real that is strictly bounded from below by $\ferrup(x) = \fpred(x) - x$,
that is, $\fpred(x) - x < r - x \leq 0$ since, clearly, $\fpred(x) < r \leq x$.
Note that $\ferrup(-\fmax) = -\infty$, coherently with the fact that
the error is unbounded from below in this case.

Suppose now that $r \in \Rset$ was rounded-to-nearest to $x \in \Fset$.
Then the error that was committed, $r - x$, is such that
$\ftwiceerrnearneg(x)/2 \leq r - x \leq \ftwiceerrnearpos(x)/2$,
where the two inequalities are strict if $x$ is odd.

In fact, if $x \notin \{ -\infty, -\fmax \}$,
then $\ftwiceerrnearneg(x)/2 = \bigl(\fpred(x) - x\bigr)/2 \leq r - x$,
for otherwise $r$ would be closer to $\fpred(x)$.
If $x = -\infty$, then $\ftwiceerrnearneg(x)/2 = +\infty$ and
$r - x = +\infty$, so $\ftwiceerrnearneg(x)/2 \leq r - x$ holds.
If $x = -\fmax$, then
\begin{align*}
  \ftwiceerrnearneg(x)/2
    &= \bigl(-\fmax - \fsucc(-\fmax)\bigr)/2 \\
    &= \bigl(-2^{\emax}(2 - 2^{1-p}) + 2^{\emax}(2 - 2^{1-p} - 2^{1-p})\bigr)/2 \\
    &= \bigl(-2^{\emax}(2 - 2^{1-p} - 2 + 2^{1-p} + 2^{1-p})\bigr)/2 \\
    &= -2^{\emax}2^{1-p}/2 \\
    &= -2^{\emax+1-p}/2 \\
    &= -2^{\emax-p}
\end{align*}
and thus, considering that $-\fmax$ is odd,
$\ftwiceerrnearneg(x)/2 < r - x$ is equivalent to
\begin{align*}
  \ftwiceerrnearneg(x)/2 + x
    &= -\bigl(2^{\emax-p} + 2^{\emax}(2 - 2^{1-p})\bigr) \\
    &= -2^\emax (2^{-p} + 2 - 2^{1-p}) \\
    &= -2^\emax (2 - 2^{-p}) \\
    &< r,
\end{align*}
which must hold, for otherwise $r$ would have been rounded to $-\infty$
\cite[Section~4.3.1]{IEEE-754-2008}.

Suppose now $x \notin \{ +\infty, \fmax \}$:
then $\ftwiceerrnearpos(x)/2 = \bigl(\fsucc(x) - x\bigr)/2 \geq r - x$,
for otherwise $r$ would be closer to $\fsucc(x)$.
If $x = +\infty$, then $\ftwiceerrnearpos(x)/2 = -\infty$ and
$r - x = -\infty$, and thus $\ftwiceerrnearpos(x)/2 \geq r - x$ holds.
If $x = \fmax$,
then
\begin{align*}
  \ftwiceerrnearpos(x)/2
    &= \bigl(\fmax - \fpred(\fmax)\bigr)/2 \\
    &= \bigl(2^{\emax}(2 - 2^{1-p}) - 2^{\emax}(2 - 2^{1-p} - 2^{1-p})\bigr)/2 \\
    &= \bigl(2^{\emax}(2 - 2^{1-p} - 2 + 2^{1-p} + 2^{1-p})\bigr)/2 \\
    &= 2^{\emax}2^{1-p}/2 \\
    &= 2^{\emax+1-p}/2 \\
    &= 2^{\emax-p}
\end{align*}
and thus, considering that $\fmax$ is odd,
$\ftwiceerrnearpos(x)/2 > r - x$ is equivalent to
\begin{align*}
  \ftwiceerrnearpos(x)/2 + x
    &= \bigl(2^{\emax-p} + 2^{\emax}(2 - 2^{1-p})\bigr) \\
    &= 2^\emax (2^{-p} + 2  - 2^{1-p}) \\
    &= 2^\emax (2 - 2^{-p}) \\
    &> r,
\end{align*}
which must hold, for otherwise $r$ would have been rounded to $+\infty$.
\end{proof}

\begin{delayedproof}[of Proposition~\ref{prop:real-approx-of-fp-constraints}]
In order to prove~\eqref{eq:real-approx-of-fp-constraints:1},
first observe that  $x \sleq y \mop_\rdown z$ implies that  $x \leq y \mop_\rdown z$.
Assume first that $y \mop_\rdown z \in  \Rset_+ \union \Rset_-$.
In this case, $y \mop_\rdown z = \rounddown{ y \circ z} $.
By inequality~\eqref{round-properties:first} of Proposition~\ref{prop:round-properties},
$y \mop_\rdown z =   \rounddown{ y \circ z}\leq  y \circ z $.
Therefore, $x \leq y \mop_\rdown z = \rounddown{ y \circ z} \leq y \circ z $.
Then,
assume that  $ y \mop_\rdown z =+\infty$. In this case, since the rounding towards minus infinity never rounds to  $+\infty$,
it follows that  $y \mop_\rdown z= y\circ z$.  Hence, $ x \leq y \circ z = + \infty$, holds.
Assume now that  $ y \mop_\rdown z =-\infty$. In this case it must be that  $x= -\infty$ then $ x \leq y \circ z$, holds.
Finally, assume that  $ y \mop_\rdown z =+ 0$ or  $ y \mop_\rdown z = -0$. In any case $x \sleq +0$ that implies
$x \leq 0$. On the other hand,
we have two cases, $y \circ z \neq 0$ or $y \circ z =0$.
For the first case,
 by Definition~\ref{def:rounding-functions},
 $0 \leq y \circ z < \fmin $, then  $ x \leq y \circ z$, holds.
For the second case, since $x \leq 0$  then $x \leq y \circ z$.

 In order to prove~\eqref{eq:real-approx-of-fp-constraints:2},
 as before observe that   $x \sleq y \mop_\rup z$ implies that  $x \leq y \mop_\rup z$. Note that
$ x + \ferrup(x) = \fpred(x)$.
So we are left to prove  $ \fpred(x) < y \circ z$.
Assume now that $ 0 < y \circ z \leq \fmax$ or  $ x \leq -\fmin$.
Moreover, note that  it cannot be the case that    $\fpred(x) \geq     y \circ z$, otherwise, by  Definition~\ref{def:rounding-functions},
$y \mop_\rup z \leq \fpred(x)$ and, therefore,  $x \leq y \mop_\rup z$ would not hold.
Then, in this case, we can conclude $\fpred(x)<   y \circ z$.  Now, assume that $- \fmin < y \circ z <0$. In this case
$y \mop_\rup z= -0$. Hence, $x \leq 0$. By Definition~\ref{def:floating-point-predecessors-and-successors},
$\fpred(x) \leq - \fmin$. Hence, $\fpred(x) \leq y \circ z $, holds. Next,  assume  $y \circ z > \fmax$. In this case
$y \mop_\rup z=  \infty$. Hence,  $x \leq \infty$. By Definition~\ref{def:floating-point-predecessors-and-successors},
$\fpred(x) \leq  \fmax$. Hence $\fpred(x) < y \circ z $, holds. Next  assume  $y \circ  z= 0$.  In this case
$y \mop_\rup z= +0$ or $y \mop_\rup z= -0$. Hence, $x \leq 0$.
By Definition~\ref{def:floating-point-predecessors-and-successors},
$\fpred(x) \leq -\fmin$. Hence $\fpred(x) < y \circ z $, holds. Finally assume $y \circ  z= \infty$.
In this case $y \mop_\rup z= \infty$.  Hence $x \sleq \infty$ and therefore  $x \leq \infty$.
By Definition~\ref{def:floating-point-predecessors-and-successors},
$\fpred(x) \leq \fmax$. Hence $\fpred(x) < y \circ z $, holds.

In order to prove~\eqref{eq:real-approx-of-fp-constraints:3},
 as the previous two cases,  note that  $x \sleq y \mop_\rnear z$ implies that  $x \leq y \mop_\rnear z$.
 First observe that    for $x\neq -\infty$,
$ x + \ftwiceerrnearneg(x)/2 < x $.
Indeed, assume first that $x \neq -\fmax$, then, by Definition~\ref{def:rounding-error-functions},
 $  \ftwiceerrnearneg(x)=   x - \fsucc(x) $. Hence $ x + \ftwiceerrnearneg(x)/2= x +  (x- \fsucc(x))/2= (3x -\fsucc(x))/2$.
Since $x < \fsucc(x)$,  we can conclude that  $ x + \ftwiceerrnearneg(x)/2 < x$.
Assume now that  $x =  -\fmax$.
 By Definition~\ref{def:rounding-error-functions},
 $  \ftwiceerrnearneg(x)=   \fpred(x) - x $. Hence $ x + \ftwiceerrnearneg(x)/2= x +  (\fpred(x)-x)/2= (x + \fpred(x))/2$.
Since $x > \fpred(x)$,  we can conclude that  $ x + \ftwiceerrnearneg(x)/2 < x$.

Now, by  Definition~\ref{def:rounding-functions}, we have to consider the following cases for $ x \mop_\rnear y \in  \Rset_+ \union \Rset_-$:
\begin{description}
\item[]$y \mop_\rnear z=  \rounddown{ y \circ z} $.
In this case, by inequality~\eqref{round-properties:first} of
Proposition~\ref{prop:round-properties},
$x + \ftwiceerrnearneg(x) /2 < x \leq   y \mop_\rnear z=  \rounddown{ y \circ z} \leq  y \circ z$.
Therefore,
$ x + \ftwiceerrnearneg(x)/2 <  y \circ z$, holds.

\item[]$y \mop_\rnear z=  \roundup{ y \circ z} $. Assume first that  $x < y \mop_\rnear z$. In this case,  by Definition~\ref{def:rounding-functions}, since $x\in \Fset$ and
$x < y \mop_\rnear z$,  it must be the case  that
 $x < y \circ z$.  Then, we can conclude that  $x + \ftwiceerrnearneg(x)/2  < x <  y \circ z$. Therefore,
$ x + \ftwiceerrnearneg(x)/2 < y \circ z$, holds.
Assume now  that  $x =y \mop_\rnear z$ and $\feven(x)$.  In this case, by Proposition~\ref{prop:rounding-error-functions}, we have that
$ \ftwiceerrnearneg\bigl(\roundnear{y \circ z}\bigr)/2 \leq
      (y \circ z) - \roundnear{y \circ z }$. Since, in this case $x=y \mop_\rnear z$, we obtain  $ \ftwiceerrnearneg(x)/2 \leq
      (y \circ z) - x$. Hence, $x+   \ftwiceerrnearneg(x)/2 \leq
      y \circ z$. If $\fodd(x)$,   by Proposition~\ref{prop:rounding-error-functions}, we have that
$ \ftwiceerrnearneg\bigl(\roundnear{y \circ z}\bigr)/2 <
      (y \circ z) - \roundnear{y \circ z }$. Hence, $x+   \ftwiceerrnearneg(x)/2 <
      y \circ z$.

\end{description}
Consider now the case  that
$y \mop_\rnear z= +0$ or $y \mop_\rnear z= -0$.  If $y \circ z\neq 0$, then $y \mop_\rnear z= \rounddown{ y \circ z} $ or
$y \mop_\rnear z= \roundup{ y \circ z} $. In this case we can reason as above.  Assume then  that   $y \circ z=  0$.  Since
 $x \sleq +0$  or $x \sleq -0$ implies that $x \leq 0$. Therefore,  we can conclude that  $x + \ftwiceerrnearneg(x) < x \leq 0$ holds.
 Assume now that  $y \mop_\rnear z= +\infty$. If $y \circ z\neq \infty$ then $y \mop_\rnear z= \roundup{ y \circ z} $. In this case we can reason as above.
 On the other hand if    $y \circ z=  +\infty$ then   $x + \ftwiceerrnearneg(x) \leq \infty$ holds.

In order to prove~\eqref{eq:real-approx-of-fp-constraints:4},
remember that $x \sgeq y \mop_\rdown z$ implies that  $x \geq y \mop_\rdown z$. Note that
$ x + \ferrdown(x) = \fsucc(x)$.
So we are left to prove  $ \fsucc(x) > y \circ z$.
Assume now that $ -\fmax < y \circ z < 0$ or  $ \fmin < y \circ  z \leq \fmax $.
Note that  it cannot be the case that    $\fsucc(x) \leq     y \circ z$, otherwise, by  Definition~\ref{def:rounding-functions},
$y \mop_\rdown z \geq \fsucc(x)$ and $x \geq y \mop_\rdown z$ would not hold.
Then, in this case, we can conclude that $\fsucc(x)>   y \circ z$.  Next, assume that $0 < y \circ z <\fmin$. In this case
$y \mop_\rdown z= +0$. Hence, $x \geq 0$. By Definition~\ref{def:floating-point-predecessors-and-successors},
$\fsucc(x) \geq \fmin$. Hence $\fsucc(x) \geq y \circ z $, holds. Next,  assume  $y \circ z <- \fmax$. In this case
$y \mop_\rdown z= - \infty$. Hence $x \geq -\infty$. By Definition~\ref{def:floating-point-predecessors-and-successors},
$\fsucc(x) \geq - \fmax$. Hence $\fsucc(x) > y \circ z $, holds. Next  assume  $y \circ  z= 0$.  In this case
$y \mop_\rdown z= +0$ or $y \mop_\rdown z= -0$. In any case, $x \geq 0$.
By Definition~\ref{def:floating-point-predecessors-and-successors},
$\fsucc(x) \geq \fmin$. Hence $\fsucc(x) > y \circ z $, holds. Finally assume $y \circ  z= -\infty$.
In this case $y \mop_\rdown z= -\infty$.  Hence,  since $x \sgeq -\infty$,  $x \geq -\infty$.
By Definition~\ref{def:floating-point-predecessors-and-successors},
$\fsucc(x) \geq - \fmax$. Hence, $\fsucc(x) > y \circ z $, holds.

In order to prove~\eqref{eq:real-approx-of-fp-constraints:5}, as before,  observe that  $x \sgeq y \mop_\rup z$ implies that  $x \geq y \mop_\rup z$.
Assume first that $ y \mop_\rup z \in  \Rset_+ \union \Rset_-$.
In this case,    $ y \mop_\rup z =   \roundup{ y \circ z} $.
By \eqref{round-properties:first} from Proposition~\ref{prop:round-properties},
$y \mop_\rup z = \roundup{ y \circ z}\geq  y \circ z $.
Then, assume that  $ y \mop_\rup z = -\infty$.
In this case, since the rounding towards plus infinity never rounds to  $-\infty$,
it follows that  $y \mop_\rup z= y\circ z$.
Hence, $ x \geq y \circ z = - \infty$, holds.
Assume now that  $ y \mop_\rup z = +\infty$.
In this case, $x = +\infty$ then $x \geq y \circ z$, holds.
Finally, assume that $y \mop_\rup z = + 0$ or $y \mop_\rup z = -0$.
In any case $x \sgeq -0$ that implies $x \geq 0$.
On the other hand, we have two cases, $y \circ z \neq 0$ or $y \circ z =0$.
For the first case, by Definition~\ref{def:rounding-functions},
$-\fmin < y \circ z < 0 $, then  $ x \geq y \circ z$, holds.
For the second case, since $x \geq 0$ then $x \geq y \circ z$.

In order to prove~\eqref{eq:real-approx-of-fp-constraints:6},
note that   $x \sgeq y \mop_\rnear z$ implies that  $x \geq y \mop_\rnear z$. First observe that    for $x\neq +\infty$,
$ x + \ftwiceerrnearpos(x)/2 > x $.
Indeed, assume first that $x \neq \fmax$, then, by Definition~\ref{def:rounding-error-functions},
 $  \ftwiceerrnearpos(x)=   x - \fpred(x) $. Hence $ x + \ftwiceerrnearpos(x)/2= x +  (x - \fpred(x))/2= (3x - \fpred(x))/2$.
Since $x > \fpred(x)$,  we can conclude that  $ x + \ftwiceerrnearpos(x)/2>x$. Assume now that  $x =  \fmax$.
 By Definition~\ref{def:rounding-error-functions},
 $  \ftwiceerrnearpos(x)=   \fsucc(x)-x $. Hence $ x + \ftwiceerrnearpos(x)/2= x +  (\fsucc(x)-x)/2= (x + \fsucc(x))/2$.
Since $x < \fsucc(x)$,  we can conclude that  $ x + \ftwiceerrnearpos(x)/2 > x$.

By Definition~\ref{def:rounding-functions}, we have to consider
the following cases for $x \mop_\rnear y \in \Rset_+ \union \Rset_-$:
\begin{description}
\item[]$y \mop_\rnear z = \roundup{y \circ z}$.
In this case, by inequality~\eqref{round-properties:first}
of Proposition~\ref{prop:round-properties},
$x + \ftwiceerrnearpos(x)/2 > x \geq y \mop_\rnear z = \roundup{y \circ z} \geq  y \circ z$.
Therefore,
$x + \ftwiceerrnearpos(x)/2 > y \circ z$, holds.

\item[]$y \mop_\rnear z=  \rounddown{ y \circ z} $.
Assume first that $x > y \mop_\rnear z$.
In this case, by Definition~\ref{def:rounding-functions},
since $x \in \Fset$ and $x > y \mop_\rnear z$, it must be the case that $x > y \circ z$.
Hence, as in the previous case, by inequality~\eqref{round-properties:first}
Proposition~\ref{prop:round-properties}, $x + \ftwiceerrnearpos(x)/2  > x >  y \circ z$.
Therefore, $x + \ftwiceerrnearpos(x)/2 > y \circ z$, holds.
Assume now  that  $x =y \mop_\rnear z$ and $\feven(x)$.  In this case, by Proposition~\ref{prop:rounding-error-functions}, we have that
$ \ftwiceerrnearpos\bigl(\roundnear{y \circ z}\bigr)/2 \geq
      (y \circ z) - \roundnear{y \circ z }$. Since, in this case $x=y \mop_\rnear z$, we obtain  $ \ftwiceerrnearpos(x)/2 \geq
      (y \circ z) - x$. Hence, $x+   \ftwiceerrnearpos(x)/2 \geq
      y \circ z$. If $\fodd(x)$,   by Proposition~\ref{prop:rounding-error-functions}, we have that
$ \ftwiceerrnearpos\bigl(\roundnear{y \circ z}\bigr)/2 >
      y \circ z - \roundnear{y \circ z }$. Hence, $x+   \ftwiceerrnearpos(x)/2 >
      y \circ z$.

\end{description}
Consider now the case  that
$y \mop_\rnear z= +0$ or $y \mop_\rnear z= -0$.  If $y \circ z\neq 0$, then $y \mop_\rnear z= \rounddown{ y \circ z} $ or
$y \mop_\rnear z= \roundup{ y \circ z} $. In this case we can reason as above.  Assume now that   $y \circ z=  0$.  Since
 $x \sgeq +0$  or $x \sgeq -0$ implies that $x \geq 0$, we can conclude that  $x + \ftwiceerrnearpos(x)/2 > x \geq 0$ holds.
 Assume now that  $y \mop_\rnear z= -\infty$. If $y \circ z\neq- \infty$ then $y \mop_\rnear z= \rounddown{ y \circ z} $. In this case
we can
 reason as above.
 On the other hand if    $y \circ z=  -\infty$ then   $x + \ftwiceerrnearpos(x)/2 \geq -\infty$ holds.
\end{delayedproof}

\begin{delayedproof}[of Proposition~\ref{prop:fp-approx-of-real-constraints}]
We first  prove~\eqref{prop:fp-approx-of-real-constraints:1}.
By inequality \eqref{round-properties:first} from
Proposition~\ref{prop:round-properties}, $e \geq  \rounddown{e}$.
Hence, $ x \geq  \rounddown{e}$. Since by hypothesis, $e \in E_\Fset$ is  an expression that evaluates on $\extRset$ to a nonzero value,
 we have three cases:
\begin{description}
\item[$\rounddown{e}\neq 0$ and $x\neq0$:] In this case  $ x \geq  \rounddown{e}$ implies  $ x \sgeq  \rounddown{e}$.

\item[$\rounddown{e}= +0$:] In this case,  $0< e< \fmin$. Then, it must be the case that  $x > 0$. Therefore     $ x \sgeq  \rounddown{e}$
holds.
\item[$x=0$:] In this case $x$ must be strictly greater than  $e$ since $e \in E_\Fset$  evaluates to a nonzero value. Therefore,
$ e<0$. Hence, by Definition~\ref{def:rounding-functions},
$\rounddown{e} \leq -\fmin$.
Then $x \sgeq \rounddown{e}$ holds.
\end{description}
In all cases, we have that $x \sgeq \rounddown {e}$.
By Definition~\ref{evaluation functions}, we conclude that $x \sgeq \evaldown{e}$.

We now prove~\eqref{prop:fp-approx-of-real-constraints:2}.
By inequality \eqref{round-properties:first} from Proposition~\ref{prop:round-properties},
as in the previous case, $e \geq  \rounddown{e}$.
Hence, $ x >  \rounddown{e}$.
Since by hypothesis $e \in E_\Fset$ is an expression that evaluates on $\extRset$ to a nonzero value,
we have three cases:
\begin{description}
\item[$\rounddown{e}\neq 0$ and $x\neq0$:] In this case  $ x >  \rounddown{e}$ implies  $ x \sgt  \rounddown{e}$.

\item[$\rounddown{e}= +0$:] In this case,  $0< e< \fmin$. Hence, $x > 0$. Therefore     $ x \sgt  \rounddown{e}$
holds.
\item[$x=0$:]  In this case $x$ must be strictly greater than  $e$ since $e \in E_\Fset$  evaluates to a nonzero value. Therefore,
$ e<0$. Hence, by Definition~\ref{def:rounding-functions},
$\rounddown{e} \leq -\fmin$. Then   $ x   \sgt  \rounddown {e}$ holds.

\end{description}
In all cases, we have that $ x   \sgt  \rounddown {e}$.
 By Definition~\ref{evaluation functions}, we conclude that   $ x  \sgt \evaldown{e}$.
Then, by Definition~\ref{def:floating-point-predecessors-and-successors}, we have the following cases on $\evaldown{e}$:
\begin{description}
\item[$\evaldown{e} = \fmax$:] In this case $\fsucc(\evaldown{e})=+\infty$. Since  $ x  \sgt \evaldown{e}$, this implies that
$x= +\infty$. Then  $ x  \sgeq  \fsucc(\evaldown{e})$, holds.
\item[ $-\fmax \leq \evaldown{e} < -\fmin $ or $\fmin \leq \evaldown{e}   < \fmax$:]  In this case $\fsucc(\evaldown{e})=
  \min \{\, y \in \Fset \mid y > \evaldown{e}  \,\}$. Since $x >  \evaldown{e}$, $x\in  \{\, y \in \Fset \mid y > \evaldown{e}  \,\}$. Hence, $x   \sgeq  \fsucc(\evaldown{e})$, holds.
\item[$\evaldown{e} = +0$ or $\evaldown{e} = -0$:]  In this case $\fsucc(\evaldown{e})=\fmin$. Since  $ x  > \evaldown{e}$, this implies that $x\geq \fmin$. Hence, $x   \sgeq  \fsucc(\evaldown{e})$, holds.
\item [$\evaldown{e} = -\fmin$:] In this case $\fsucc(\evaldown{e})=-0$. Since  $ x  > \evaldown{e}=-\fmin$, $x \sgeq -0$.
 Hence, $x   \sgeq  \fsucc(\evaldown{e})$, holds.
\item[$\evaldown{e} = -\infty$:] In this case $\fsucc(\evaldown{e})=-\fmax$. Since  $ x  > \evaldown{e}=-\infty$, $x \sgeq -\fmax$.
 Hence, $x   \sgeq  \fsucc(\evaldown{e})$, holds.
\end{description}

We now prove~\eqref{prop:fp-approx-of-real-constraints:3}.
By inequality~\eqref{round-properties:first} from Proposition~\ref{prop:round-properties},
$e \leq \roundup{e}$.
Hence, like before, $x \leq \roundup{e}$.
Since by hypothesis $e \in E_\Fset$ is an expression that evaluates on $\extRset$ to a nonzero value,
we have three cases:
\begin{description}
\item[$\roundup{e}\neq 0$ and $x\neq 0$:] In this case  $ x \leq  \roundup{e}$ implies  $ x \sleq  \roundup{e}$.
\item[$\roundup{e}= -0$:] In this case,  $-\fmin< e< 0$. Hence, $x < 0$. Therefore     $ x \sleq  \roundup{e}$
holds.
\item[$x=0$:] In this case it must be the case that  $x$ is strictly smaller than  $e$, since  $e \in E_\Fset$  evaluates to a nonzero value.
 Therefore,
$ e>0$. Hence, by Definition~\ref{def:rounding-functions},
 $\roundup{e} \geq \fmin$. Then   $ x  \sleq \roundup {e}$ holds.
\end{description}

In any case, $x \sleq \roundup {e}$ holds.
By Definition~\ref{evaluation functions}, we conclude that $x \sleq \evalup{e}$.

Next we prove~\eqref{prop:fp-approx-of-real-constraints:4}.
By, again, inequality~\eqref{round-properties:first} from Proposition~\ref{prop:round-properties},
$e \leq  \roundup{e}$.
Hence, $x < \roundup{e}$.
Since by hypothesis $e \in E_\Fset$ is an expression that evaluates on $\extRset$ to a nonzero value,
we have three cases:
\begin{description}
\item[$\roundup{e}\neq 0$ and $x\neq 0$:] In this case  $ x <  \roundup{e}$ implies  $ x \slt  \roundup{e}$.
\item[$\roundup{e}= -0$:] In this case,  $-\fmin< e< 0$. Hence, $x < 0$. Therefore     $ x \slt \roundup{e}$
holds.
\item[$x=0$:] In this case it must be the case that  $x$ is strictly smaller than  $e$, since  $e \in E_\Fset$  evaluates to a nonzero value.
 Therefore,
$ e>0$. Hence, by Definition~\ref{def:rounding-functions},
 $\roundup{e} \geq \fmin$. Then   $ x   \slt  \roundup {e}$ holds.
\end{description}

 In any case, $ x   \slt \roundup {e}$ holds.
 By Definition~\ref{evaluation functions}, we conclude that   $ x  \slt \evalup{e}$.
By Definition~\ref{def:floating-point-predecessors-and-successors}, we have the following cases on $\evalup{e}$:

\begin{description}
\item[$\evalup{e} = -\fmax$:] In this case $\fpred(\evalup{e})=-\infty$. Since  $ x  \slt \evalup{e}$,  this implies that
$x= -\infty$. Then  $ x  \sleq  \fpred(\evalup{e})$, holds.
\item[ $\fmin <\evalup{e} \leq \fmax $ or $-\fmax < \evalup{e}   \leq -\fmin$:]  In this case $\fpred(\evalup{e})=
  \max \{\, y \in \Fset \mid y < \evalup{e}  \,\}$. Since $x <  \evalup{e}$, $x\in  \{\, y \in \Fset \mid y < \evalup{e}  \,\}$. Hence, $x   \sleq  \fpred(\evalup{e})$, holds.
\item[$\evalup{e} = +0$ or $\evalup{e} = -0$:]  In this case $\fpred(\evalup{e})=-\fmin$. Since  $ x  < \evalup{e}$, this implies that $x\leq -\fmin$. Hence, $x   \sleq  \fpred(\evalup{e})$, holds.
\item [$\evalup{e} = \fmin$ :] In this case $\fpred(\evalup{e})=+0$. Since  $ x  < \evalup{e}=\fmin$, $x \sleq +0$.
 Hence, $x   \sleq  \fpred(\evalup{e})$, holds.
\item[$\evaldown{e} = +\infty$ :] In this case $\fpred(\evalup{e})=\fmax$. Since  $ x  < \evalup{e}=\infty$, $x \sleq \fmax$.
 Hence, $x   \sleq  \fpred(\evalup{e})$, holds.
\end{description}

In order to prove~\eqref{case:x-geq-e-implies-x-sgeq-evalup-e} we first want to prove that  $x \sgeq \roundup{e}$.
To this aim consider the following cases for $e$:

\begin{description}
\item[$e>\fmax$:] In this case  $ \roundup{e}=+ \infty$. On the hand, $x\geq e >\fmax$. Since $x\in \Fset$  implies that
$x=+\infty$.  Hence $x \sgeq \roundup{e}$.

\item[$e\leq -\fmin$ or $0<e\leq \fmax$:] In this case  $ \roundup{e}=\min \{\, z \in \Fset \mid z \geq e  \,\}$.
Since $x \geq e$, $x\in  \{\, z \in \Fset \mid z \geq  e  \,\}$. Hence, $x  \geq  \roundup{e}$, holds and also $x \sgeq \roundup{e}$.

\item[$-\fmin< e <0$:] In this case  $ \roundup{e}=-0$.
Since $x \geq e$ and $x\in \Fset$,  $x  \sgeq   -0$, holds.

\item[$e=-\infty$:] In this case    $ \roundup{e}=-\infty$ and  $x  \sgeq -\infty   $ holds.

\end{description}
Since by hypothesis  $\roundup{e}=\evalup{e}$, we can conclude that   $x  \sgeq \evalup{e}  $ holds.

In order to prove~\eqref{case:x-leq-e-implies-x-sleq-evaldown-e} we first want to prove that  $x \sleq \rounddown{e}$.
To this aim consider the following cases for $e$:

\begin{description}
\item[$e < -\fmax$:] In this case  $ \rounddown{e}= -\infty$. On the hand, $x\leq e <- \fmax$. Since $x\in \Fset$  implies that
$x=-\infty$.  Hence $x \sleq \rounddown{e}$.

\item[$e\geq \fmin$ or $-\fmax \leq e < 0$:] In this case  $ \rounddown{e}=\max \{\, z \in \Fset \mid z \leq e  \,\}$.
Since $x \leq e$, $x\in  \{\, z \in \Fset \mid z \leq  e  \,\}$. Hence, $x  \leq  \rounddown{e}$, holds and also $x \sleq \rounddown{e}$.

\item[$0< e <\fmin$:] In this case  $ \rounddown{e}=+0$.
Since $x \leq e$ and $x\in \Fset$,  $x  \sleq   +0$, holds.

\item[$e=+\infty$:] In this case    $ \rounddown{e}=\infty$ and  $x  \sleq +\infty   $ holds.

\end{description}
Since by hypothesis  $\rounddown{e}=\evaldown{e}$, we can conclude that   $x  \sleq \evaldown{e}  $ holds.
\end{delayedproof}


\subsection{Proofs of Results in Section~\ref{sec:filtering-algorithms}}
\label{sec:proofs-filtering}

\begin{delayedproof}[of Theorem~\ref{teo:direct-projection-division}]
Given the constraint $x = y \mdiv_S z$ with
$x \in X = [x_\ell, x_u]$,
$y \in Y = [y_\ell, y_u]$ and
$z \in Z = [z_\ell, z_u]$,
Algorithm~\ref{algo5} computes a new refining interval $X'$ for variable $x$.
Note that $X' = [x'_\ell, x'_u] \inters X$, which assures us that $X' \sseq X$.

As for the proof of Theorem \ref{teo:indirect-projection-mult},
it is easy to verify that $y_L$ and $w_L$ (resp., $y_U$ and $w_U$) computed using
function $\tau$ of Figure~\ref{fig:the-tau-function},
are the boundaries of $Y$ and $W$ upon which $x$ touches its
minimum (resp., maximum).
Moreover, remember that by Proposition~\ref{prop:worst-case-rounding-modes},
following the same reasoning of the proofs of the previous theorems,
we can focus on finding a lower bound for $y_L \mdiv_{r_\ell} w_L $
and an upper bound for $y_U \mdiv_{r_u} w_U$.

We will now comment only on the most critical entries of function $\dirdivl$
of Figure~\ref{fig:direct-projection-division}:
let us briefly discuss the cases in which $y_L = -\infty$ and $w_L = \pm \infty$.
\begin{description}
\item[$w_L=-\infty.$]
  In this case, by function $\tau$ of Figure~\ref{fig:the-tau-function}
  (see the first three cases), we have $y_L = y_u = -\infty$,
  while either $w_L = w_\ell$ or $w_L = w_u$.
  Since by the IEEE~754 Standard \cite{IEEE-754-2008}
  dividing $\pm \infty$ by $\pm \infty$ is an invalid operation,
  we are left to consider the case $w_L = w_\ell$.
  In this case, recall that by the IEEE~754 Standard \cite{IEEE-754-2008},
  dividing $-\infty$ by a finite negative number yields
  $+\infty$. Hence, we can conclude $x_\ell = +\infty$.
\item[$w_L = +\infty.$]
  By function $\tau$ of Figure~\ref{fig:the-tau-function} (see the fourth and last case),
  we have $y_L = y_\ell = -\infty$, while $w_L = w_\ell = +\infty$.
  Hence, $x_\ell = -0$, since dividing a negative finite number by $+\infty$ gives $-0$.
\end{description}
A similar reasoning applies for the cases $y_L = +\infty$, $w_L= \pm \infty$.
Dually, the only critical entries of function $\dirdivu$
of Figure~\ref{fig:direct-projection-division} are those in which
$y_U = \pm \infty$ and $w_U = \pm \infty$ and can be handled analogously.

We are left to prove that
$\forall X'' \subset X, \exists r \in S, y \in Y, z \in Z \itc y \mdiv_r z \not\in X''$.
Let us focus on the lower bound $x^+_\ell$ proving that,
if $[x^+_\ell, x^+_\ell] \neq \emptyset$, then there exist $r \in S, y \in Y, z\in Z$
such that $y \mdiv_r z = x^+_\ell$.
Consider the particular values $y_L$, $z_\ell = w_L$ and $r_\ell$
that correspond to $x^+_\ell$ in Algorithm~\ref{algo5},
i.e. $y_L$ and $w_L$ and $r_\ell$ are such that $\dirdivl(y_L,w_L,r_\ell) = x^+_\ell$.
By Algorithm~\ref{algo5}, such $y_L$ and $w_L$ must exist.
First consider the cases in which
$y_L \not\in (\Rset_- \cup \Rset_+)$ or $w_L \not\in (\Rset_- \cup \Rset_+)$.
A brute-force verification was successfully conducted, in this cases,
to prove that $y_L \mdiv_{r_\ell} w_L = x^+_\ell$.
For the cases in which $y_L \in (\Rset_- \cup \Rset_+)$ and $w_L \in (\Rset_- \cup \Rset_+)$
we have, by definition of $\dirdivl$ of Figure~\ref{fig:direct-projection-division},
that $x^+_\ell = y_L \mdiv_{r_\ell} w_L$.
Remember that, by Proposition~\ref{prop:worst-case-rounding-modes},
there exists $r \in S$ such that
$y_L \mdiv_\mathrm{r_\ell} w_L = y_L \mdiv_\mathrm{r} w_L$.
Since $y_L \in Y$ and $w_L \in Z$, we can conclude that
$x^+_\ell \not\in X''$ implies that $y_L \mdiv_\mathrm{r} w_L \not\in X''$,
for any $X'' \subseteq X'$.
An analogous reasoning applies to $x^-_\ell$, to $x^+_u$ and $x^-_u$.
This allows us to prove the optimality claim.
\end{delayedproof}

\begin{delayedproof}[of Theorem~\ref{teo:first-indirect-projection-division}]
Given the constraint $x = y \mdiv_S z$ with
$x \in X = [x_\ell, x_u]$,
$y \in Y = [y_\ell, y_u]$ and
$z \in Z = [z_\ell, z_u]$,
Algorithm~\ref{algo6} computes a new, refining interval $Y'$
for variable $y$.
It returns either
$Y' \assign (Y \inters [y^-_\ell, y^-_u]) \biguplus (Y \inters [y^+_\ell, y^+_u])$
or $Y' = \emptyset$: hence, in both cases, we are sure that $Y' \sseq Y$.

By Proposition~\ref{prop:worst-case-rounding-modes},
we can focus on finding a lower bound for $y\in Y$
by exploiting the constraint $y \mdiv_{\bar{r}_\ell} z = x$
and an upper bound for $y$
by exploiting the constraint $y \mdiv_{\bar{r}_u} z = x$.

In order to compute correct bounds for $y$,
Algorithm~\ref{algo7} first splits the interval of $z$ into
the sign-homogeneous intervals $Z_-$ and $Z_+$,
since knowing the sign of $z$ is crucial to determine correct bounds for $y$.
Hence, for $W = Z_-$ (and, analogously, for $W = Z_+$), it calls
function $\sigma$ of Figure~\ref{fig:the-sigma-function} to determine
the appropriate extrema of intervals $X$ and $W$ to be used to compute
the new lower and upper bounds for $y$.
As we did in the proof of Theorem~\ref{teo:direct-projection-mult},
it is easy to verify that $x_L$ and $w_L$ (resp., $x_U$ and $w_U$),
computed using function $\sigma$ of Figure~\ref{fig:the-sigma-function},
are the boundaries of $X$ and $W$ upon which $y$ touches its
minimum (resp., maximum).
Functions $\invfirstdivl$ of Figure~\ref{fig:first-indirect-projection-division-yl}
and $\invfirstdivu$ of Figure~\ref{fig:first-indirect-projection-division-yu}
are then used to find the new bounds for $y$.
The so obtained intervals for $y$ will be eventually joined using convex union
to obtain the refining interval for $y$.

We will now prove the non-trivial parts of the definitions of
functions $\invfirstdivl$ and $\invfirstdivu$.
Concerning the case analysis of $\invfirstdivl$
(Fig~\ref{fig:first-indirect-projection-division-yl}) marked as $a_4$,
the result changes depending on the selected rounding mode:
\begin{description}
\item[$\bar{r}_\ell = \rup:$]
  we clearly must have $y = +\infty$, according to the IEEE~754
  Standard \cite{IEEE-754-2008};
\item[$\bar{r}_\ell = \rdown:$]
  it must be $y / w_L < -\fmax$ and thus,
  since $w_L$ is negative, $y > -\fmax \cdot w_L$ and,
  by~\eqref{prop:fp-approx-of-real-constraints:2}
  of Proposition~\ref{prop:fp-approx-of-real-constraints},
  $y \sgeq \fsucc(-\fmax \mmul_\rdown w_L)$.
\item[$\bar{r}_\ell = \rnear:$]
  since $\fodd(\fmax)$, for $w_L = -\infty$ we need $y$ to be greater than
  or equal to $\bigl(-\fmax + \ftwiceerrnearneg(-\fmax)/2\bigr) \cdot w_L$.
  If
  \(
    \evalup{\bigl(-\fmax + \ftwiceerrnearneg(-\fmax)/2\bigr) \cdot w_L}
    = \roundup{\bigl(-\fmax + \ftwiceerrnearneg(-\fmax)/2\bigr) \cdot w_L}
  \),
  by~\eqref{case:x-geq-e-implies-x-sgeq-evalup-e}
  of Proposition~\ref{prop:fp-approx-of-real-constraints}, we can conclude
  $y \sgeq \evalup{\bigl(-\fmax + \ftwiceerrnearneg(-\fmax)/2\bigr) \cdot w_L}$.
  On the other hand, if
  \(
    \evalup{\bigl(-\fmax + \ftwiceerrnearneg(-\fmax)/2\bigr) \cdot w_L}
    \neq \roundup{\bigl(-\fmax + \ftwiceerrnearneg(-\fmax)/2\bigr) \cdot w_L}
  \),
  then we can only apply~\eqref{prop:fp-approx-of-real-constraints:1}
  of Proposition~\ref{prop:fp-approx-of-real-constraints}, obtaining
  $y \sgeq \evaldown{\bigl(-\fmax + \ftwiceerrnearneg(-\fmax)/2\bigr) \cdot w_L}$.
\end{description}

The case analysis of $\invfirstdivl$ (Fig~\ref{fig:first-indirect-projection-division-yl})
marked as $a_5$ can be explained as follows:
\begin{description}
\item[$\bar{r}_\ell = \rdown:$]
  we must have $y = +\infty$, according to the IEEE~754
  Standard \cite{IEEE-754-2008};
\item[$\bar{r}_\ell = \rup:$]
  inequality $y / w_L > \fmax$ must hold and thus, since $w_L$ is positive,
  $y > \fmax \cdot w_L$ and, by~\eqref{prop:fp-approx-of-real-constraints:2}
  of Proposition~\ref{prop:fp-approx-of-real-constraints},
  $y \sgeq \fsucc(\fmax \mmul_\rdown w_L)$.
\item[$\bar{r}_\ell = \rnear:$]
  since $\fodd(\fmax)$, for $x_L = +\infty$ we need $y$ to be greater than
  or equal to $\bigl(\fmax + \ftwiceerrnearpos(\fmax) / 2\bigl) \cdot w_L$. If
  \(
    \evalup{\bigl(\fmax + \ftwiceerrnearpos(\fmax) / 2\bigl) \cdot w_L}
    = \roundup{\bigl(\fmax + \ftwiceerrnearpos(\fmax) / 2\bigl) \cdot w_L}
  \),
  by~\eqref{case:x-geq-e-implies-x-sgeq-evalup-e}
  of Proposition~\ref{prop:fp-approx-of-real-constraints},
  we can  conclude
  $y \sgeq \evalup{\bigl(\fmax + \ftwiceerrnearpos(\fmax) / 2\bigl) \cdot w_L}$.
  On the other hand, if
  \(
    \evalup{\bigl(\fmax + \ftwiceerrnearpos(\fmax) / 2\bigl) \cdot w_L}
    \neq \roundup{\bigl(\fmax + \ftwiceerrnearpos(\fmax) / 2\bigl) \cdot w_L}
  \)
  then, we can only apply~\eqref{prop:fp-approx-of-real-constraints:1}
  of Proposition~\ref{prop:fp-approx-of-real-constraints}, obtaining
  $y \sgeq \evaldown{\bigl(\fmax + \ftwiceerrnearpos(\fmax) / 2\bigl) \cdot w_L}$.
\end{description}

The explanation for the case analysis of $\invfirstdivl$
(Fig~\ref{fig:first-indirect-projection-division-yl}) marked as $a_6$
is the following:
\begin{description}
\item[$\bar{r}_\ell = \rup:$]
  the lowest value of $y$ that yields $x_L = +0$ with $w_L \in \Rset_-$
  is clearly $y = -0$;
\item[$\bar{r}_\ell = \rdown:$]
  inequality $y / w_L < \fmin$ should hold and thus, since $w_L$ is negative,
  $y > \fmin \cdot w_L$ and, by~\eqref{prop:fp-approx-of-real-constraints:2}
  of Proposition~\ref{prop:fp-approx-of-real-constraints},
  $y \sgeq \fsucc(\fmin \mmul_\rdown w_L)$.
\item[$\bar{r}_\ell = \rnear:$]
  since $\fodd(\fmin)$, for $x_L = +0$ we need $y$ to be greater than
  or equal to $(\fmin \cdot w_L)/2$.
  Since in this case
  \(
    \evalup{(\fmin \cdot w_L)/2}
    = \roundup{(\fmin \cdot w_L)/2}
    = (\fmin \mmul_\rup w_L)/2
  \),
  by~\eqref{case:x-geq-e-implies-x-sgeq-evalup-e}
  of Proposition~\ref{prop:fp-approx-of-real-constraints},
  we can conclude $y \sgeq (\fmin \mmul_\rup w_L)/2$.
\end{description}

Concerning the case analysis of $\invfirstdivl$
(Fig~\ref{fig:first-indirect-projection-division-yl}) marked as $a_7$,
we must distinguish between the following cases:
\begin{description}
\item[$\bar{r}_\ell = \rdown:$]
  considering $x_L = -0$ and $w_L \in \Rset_+$,
  we clearly must have $y = -0$;
\item[$\bar{r}_\ell = \rup:$]
  it should be $y / w_L > -\fmin$ and thus, since $w_L$ is positive,
  $y > -\fmin \cdot w_L$ and, by~\eqref{prop:fp-approx-of-real-constraints:2}
  of Proposition~\ref{prop:fp-approx-of-real-constraints},
  $y \sgeq \fsucc(-\fmin \mmul_\rdown w_L)$.
\item[$\bar{r}_\ell = \rnear:$]
  since $\fodd(\fmin)$, for $x_L = -0$ we need
  $y$ be to greater than or equal to $\bigl(-\fmin \cdot w_L\bigr)/2$.
  Since in this case
  \(
    \evalup{\bigl(-\fmin \cdot w_L\bigr)/2}
    = \roundup{\bigl(-\fmin \cdot w_L\bigr)/2}
    = (-\fmin \mmul_\rup w_L)/2
  \),
  by~\eqref{case:x-geq-e-implies-x-sgeq-evalup-e}
  of Proposition~\ref{prop:fp-approx-of-real-constraints},
  we can conclude $y \sgeq (-\fmin \mmul_\rup w_L)/2$.
\end{description}

Similar arguments can be used to prove the case analyses of $\invfirstdivu$
of Fig~\ref{fig:first-indirect-projection-division-yu} marked as
$a_9$, $a_{10}$, $a_{11}$ and $a_{12}$.

We will now analyze the case analyses of $\invfirstdivl$
of Fig~\ref{fig:first-indirect-projection-division-yl}
marked as $a^-_3$ and $a^+_3$, and the ones of  $\invfirstdivu$
of Fig~\ref{fig:first-indirect-projection-division-yu}
marked as $a^-_8$ and $a^+_8$.
We can assume, of course,
$X = [x_\ell, x_u]$,
$Y = [y_\ell, y_u]$ and
$Z = [w_\ell, w_u]$, where $x_\ell, x_u, w_\ell, w_u \in \Fset \inters \Rset$,
$x_\ell \leq x_u$,
$w_\ell \leq w_u$ and
$\sgn(w_\ell) = \sgn(w_u)$.
Exploiting $x \sleq y \mdiv z$ and $x \sgeq y \mdiv z$,
by Proposition~\ref{prop:real-approx-of-fp-constraints}, we have
\begin{align}
    y / z
  &\begin{cases}
      \mathord{} \geq x,
        &\text{if $\bar{r}_\ell = \rdown$;} \\
      \mathord{} > x + \ferrup(x) = \fpred(x),
        &\text{if $\bar{r}_\ell = \rup$;} \\
      \mathord{} \geq x + \ftwiceerrnearneg(x)/2,
        &\text{if $\bar{r}_\ell = \rnear$ and $\feven(x)$;} \\
      \mathord{} > x + \ftwiceerrnearneg(x)/2,
        &\text{if $\bar{r}_\ell = \rnear$ and $\fodd(x)$.}
    \end{cases} \label{eq:division-firstandsecond-indirect-projection-I}\\
    y / z
    &\begin{cases}
      \mathord{} < x + \ferrdown(x) = \fsucc(x),
        &\text{if $\bar{r}_u = \rdown$;} \\
      \mathord{} \leq x,
        &\text{if $\bar{r}_u = \rup$;} \\
      \mathord{} \leq x + \ftwiceerrnearpos(x)/2,
        &\text{if $\bar{r}_u = \rnear$ and $\feven(x)$;} \\
      \mathord{} < x + \ftwiceerrnearpos(x)/2,
        &\text{if $\bar{r}_u = \rnear$ and $\fodd(x)$.}
    \end{cases} \label{eq:division-firstandsecond-indirect-projection-II}
\end{align}
Since the case $z = 0$ is handled separately
by $\invfirstdivl$ of Fig~\ref{fig:first-indirect-projection-division-yl}
and by $\invfirstdivu$ of Fig~\ref{fig:first-indirect-projection-division-yu},
we can assume $z \neq 0$.
Thanks to the split of $Z$ into a positive and a negative part,
the sign of $z$ is determinate.
In the following, we will prove the case analyses
marked as $a_3^+$ and $a_8^+$, hence assuming $z > 0$.
From the previous case analysis we can derive
\begin{align}
  y
   &\begin{cases}
      \mathord{} \geq x\cdot z,
        &\text{if $\bar{r}_\ell = \rdown$;} \\
      \mathord{} > \fpred(x) \cdot z,
        &\text{if $\bar{r}_\ell = \rup$} \\
      \mathord{} \geq \bigl(x + \ftwiceerrnearneg(x)/2\bigr) \cdot z,
        &\text{if $\bar{r}_\ell = \rnear$ and $\feven(x)$;} \\
      \mathord{} > \bigl( x + \ftwiceerrnearneg(x)/2 \bigr) \cdot z,
        &\text{if $\bar{r}_\ell = \rnear$ and $\fodd(x)$;}
    \end{cases} \label{eq:division-firstandsecond-indirect-projection-III} \\
    y
   &\begin{cases}
      \mathord{} < \fsucc(x)\cdot z,
        &\text{if $\bar{r}_u = \rdown$;} \\
       \mathord{} \leq x \cdot z,
        &\text{if $\bar{r}_u = \rup$;} \\
      \mathord{} \leq \bigl(x + \ftwiceerrnearpos(x)/2 \bigr) \cdot z,
        &\text{if $\bar{r}_u = \rnear$ and $\feven(x)$;} \\
      \mathord{} < \bigl(x + \ftwiceerrnearpos(x)/2 \bigr) \cdot z,
        &\text{if $\bar{r}_u = \rnear$ and $\fodd(x)$.}
    \end{cases} \label{eq:division-firstandsecond-indirect-projection-IV}
\end{align}

Note that the members of the product are independent.
Therefore, we can find the minimum of the product by minimizing
each member of the product. Since we are analyzing the case in which
$W = Z_+$, let $(x_L, x_U, w_L, w_U)$ as defined in function $\sigma$
of Figure~\ref{fig:the-sigma-function},
replacing the role of $y$ with $z$ and the role of $z$ with $x$.
Hence, by
Proposition~\ref{prop:min-max-x-ftwiceerrnearneg-ftwiceerrnearpos} and
the monotonicity of `$\fpred$' and `$\fsucc$' we obtain
\begin{align}
  y
   &\begin{cases}
      \mathord{} \geq x_L \cdot w_L,
        &\text{if $\bar{r}_\ell = \rdown$;} \\
      \mathord{} > \fpred(x_L) \cdot w_L,
        &\text{if $\bar{r}_\ell = \rup$} \\
      \mathord{} \geq \bigl(x_L + \ftwiceerrnearneg(x_L)/2\bigr) \cdot w_L,
        &\text{if $\bar{r}_\ell = \rnear$ and $\feven(x)$;} \\
      \mathord{} > \bigl(x_L + \ftwiceerrnearneg(x_L)/2 \bigr) \cdot w_L,
        &\text{if $\bar{r}_\ell = \rnear$ and $\fodd(x)$;}
    \end{cases} \\
    y
   &\begin{cases}
      \mathord{} < \fsucc(x_U) \cdot w_U,
        &\text{if $\bar{r}_u = \rdown$;} \\
       \mathord{} \leq x_U \cdot w_U,
        &\text{if $\bar{r}_u = \rup$;} \\
\label{eq:div-1-refine-real-upper_bounds}
      \mathord{} \leq \bigl(x_U + \ftwiceerrnearpos(x_U)/2 \bigr) \cdot w_U,
        &\text{if $\bar{r}_u = \rnear$ and $\feven(x)$;} \\
      \mathord{} < \bigl(x_U + \ftwiceerrnearpos(x_U)/2 \bigr) \cdot w_U,
        &\text{if $\bar{r}_u = \rnear$ and $\fodd(x)$.}
    \end{cases}
\end{align}
We can now exploit Proposition~\ref{prop:fp-approx-of-real-constraints}
and obtain:
\begin{align}
\label{eq:div-inv-1-num-rup-rdown-lower}
  y'_\ell &\defeq
    \begin{cases}
      x_L \mmul_\rup w_L,
        &\text{if $\bar{r}_\ell = \rdown$;} \\
      \fsucc\bigl(\fpred(x_L) \mmul_\rdown w_L\bigr),
        &\text{if $\bar{r}_\ell = \rup$;}
    \end{cases} \\
\label{eq:div-inv-1-num-rup-rdown-upper}
     y'_u &\defeq
    \begin{cases}
      \fpred\bigl(\fsucc(x_U) \mmul_\rup w_U\bigr),
        &\text{if $\bar{r}_u = \rdown$;} \\
      x_U \mmul_\rdown w_U,
        &\text{if $\bar{r}_u = \rup$.}
    \end{cases}
\end{align}
Indeed, if $\bar{r}_\ell = \rup$ and $x_L \neq 0$,
then part~\eqref{case:x-geq-e-implies-x-sgeq-evalup-e} of
Proposition~\ref{prop:fp-approx-of-real-constraints}
applies and we have $y \sgeq x_L \mmul_\rup w_L$.
On the other hand, if $x_L = 0$, since by hypothesis $z > 0$ implies $w_L > 0$,
according to IEEE~754 \cite[Section~6.3]{IEEE-754-2008},
we have $x_L \mmul_\rup w_L = \sgn(x_L) \cdot 0$ and, indeed, for each non-NaN,
nonzero and finite $w \in \Fset \inters [+0, +\infty]$,
$\sgn(x_L) \cdot 0$ is the least value for $y$ that satisfies
$\sgn(x_L) \cdot 0 = y \mdiv_\rdown w$.

Analogously,
if $\bar{r}_\ell = \rup$ and $x_L \neq \fmin$,
then Proposition~\ref{prop:fp-approx-of-real-constraints}
applies and we have $\fsucc\bigl(\fpred(x_L) \mmul_\rdown w_L\bigr)$.
On the other hand, if $x_L = \fmin$, in this case,
$\fsucc\bigl(\fpred(x_L) \mmul_\rdown w_L\bigr) = \fmin$
which is consistent with the fact that, for each non-NaN, nonzero
and finite $w \in \Fset \inters [+0, +\infty]$, $\fmin$
is the lowest value for $y$ that satisfies $\fmin = y \mdiv_\rup w$.

A symmetric argument justifies~\eqref{eq:div-inv-1-num-rup-rdown-upper}.

As before,
we need to approximate the values of the expressions
$e^+_\ell = \bigl(x_L + \ftwiceerrnearneg(x_L)/2\bigr) \cdot w_L$ and
$e^+_u = \bigl(x_U + \ftwiceerrnearpos(x_U)/2 \bigr) \cdot w_U$.
We leave this as an implementation choice, thus taking into account the case
$\evalup{e^+_\ell} = \roundup{e^+_\ell}$ and $\evaldown{e^+_u} = \rounddown{e^+_u}$
as well as
$\evalup{e^+_\ell} > \roundup{e^+_\ell}$ and $\evaldown{e^+_u} < \rounddown{e^+_u}$.
Therefore, when
$\evaldown{e^+_u} < \rounddown{e^+_u}$ by~\eqref{eq:div-1-refine-real-upper_bounds}
and~\eqref{prop:fp-approx-of-real-constraints:3}
of Proposition~\ref{prop:fp-approx-of-real-constraints}
we obtain $y \sleq  \evalup{e^+_u}$, while, when
$\evaldown{e^+_\ell} > \rounddown{e^+_\ell}$ by~\eqref{eq:div-1-refine-real-upper_bounds}
and~\eqref{prop:fp-approx-of-real-constraints:1}
of Proposition~\ref{prop:fp-approx-of-real-constraints}
we obtain $y \sgeq  \evaldown{e^+_\ell}$.

Thus, for the case in which $\bar{r}_\ell = \rnear$,
since $e^+_u \neq 0$ and $e^+_\ell \neq 0$,
by Proposition~\ref{prop:fp-approx-of-real-constraints}, we have
\begin{align}
\label{eq:div1-inv-num-rnear-lower}
  y'_\ell &\defeq
    \begin{cases}
     \evalup{e^+_\ell},
        &\text{if $\feven(x_L)$ and $\evalup{e^+_\ell} = \roundup{e^+_\ell}$;} \\
          \evaldown{e^+_\ell},
        &\text{if $\feven(x_L)$ and $\evalup{e^+_\ell} \neq \roundup{e^+_\ell}$;} \\
      \fsucc\bigl(\evaldown{e^+_\ell}\bigr),
        &\text{otherwise;}
    \end{cases} \\
\intertext{%
whereas, for the case in which $\bar{r}_u = \rnear$,
}
\label{eq:div1-inv-num-rnear-upper}
  y'_u &\defeq
    \begin{cases}
      \evaldown{e^+_u},
        &\text{if $\feven(x_U)$ and $\evaldown{e^+_u} = \rounddown{e^+_u}$;} \\
          \evalup{e^+_u},
        &\text{if $\feven(x_U)$ and $\evaldown{e^+_u} \neq \rounddown{e^+_u}$;} \\
      \fpred\bigl(\evalup{e^+_u}\bigr),
        &\text{otherwise.}
    \end{cases}
\end{align}

An analogous reasoning, but with $z < 0$, allows us to obtain
the case analyses marked as $a_3^-$ and $a_8^-$.
\end{delayedproof}

\begin{delayedproof}[of Theorem~\ref{teo:second-indirect-projection-div}]
Given the constraint $x = y \mdiv_S z$ with
$x \in X = [x_\ell, x_u]$,
$y \in Y = [y_\ell, y_u]$ and
$z \in Z = [z_\ell, z_u]$,
Algorithm~\ref{algo7} finds a new, refined interval $Z'$ for variable $z$.

Since it assigns either
$Z' \assign (Z \inters [z^-_\ell, z^-_u]) \biguplus (Z \inters [z^+_\ell, z^+_u])$
or $Z'=\emptyset$, in both cases we are sure that $Z' \sseq Z$.
By Proposition~\ref{prop:worst-case-rounding-modes}, as in the previous proofs,
we can focus on finding a lower bound for $z \in Z$
by exploiting the constraint $y \mdiv_{\bar{r}_\ell} z = x$
and an upper bound for $z$
by exploiting the constraint $y \mdiv_{\bar{r}_u} z = x$.

We first need to split interval $X$ into
the sign-homogeneous intervals $X_-$ and $X_+$,
because knowing the sign of $x$ is crucial
for determining correct bounds for $z$.
Hence, for $V = X_-$ (and, analogously, for $V = X_+$)
function $\tau$ of Figure~\ref{fig:the-tau-function}
determines the appropriate interval extrema of $Y$ and $V$
to be used to compute the new lower and upper bounds for $z$.
As in the previous proofs
(see, for example, proof of Theorem~\ref{teo:indirect-projection-mult}),
it is easy to verify that $y_L$ and $v_L$ (resp., $y_U$ and $v_U$)
computed using function $\tau$ of Figure~\ref{fig:the-tau-function}
are the boundaries of $Y$ and $V$ upon which $z$ touches its
minimum (resp., maximum).
Functions $\invsecdivl$ of Figure~\ref{fig:second-indirect-projection-division-zl}
and $\invsecdivu$ of Figure~\ref{fig:second-indirect-projection-division-zu}
are then used to find the new bounds for $z$.
The so obtained intervals for $z$ will be then joined with convex union
in order to obtain the refining interval for $z$.

We will prove the most important parts of the definitions
of $\invsecdivl$ (Figure~\ref{fig:second-indirect-projection-division-zl})
and $\invsecdivu$ (Figure~\ref{fig:second-indirect-projection-division-zu})
only, starting with the case analysis marked as $a_4$.
Depending on the rounding mode in effect, the following arguments are given:
\begin{description}
\item[$\bar{r}_\ell = \rdown:$]
  in this case, the only possible way to obtain $-0$
  as the result of the division is having $z = +\infty$ (with $y \in \Rset_-$);
\item[$\bar{r}_\ell = \rup:$]
  it should be $y_L / z > -\fmin$ and thus, since $y_L$ and $x_L$ are negative,
  we can conclude that $z$ is positive.
  Thus, $y_L > -\fmin \cdot z$ implies $y_L / -\fmin < z$,
  and by~\eqref{prop:fp-approx-of-real-constraints:2}
  of Proposition~\ref{prop:fp-approx-of-real-constraints},
  $z \sgeq \fsucc(z_L \mdiv_\rdown -\fmin)$.
\item[$\bar{r}_\ell = \rnear:$]
  since $\fodd(-\fmin)$, for $v_L = -0$ we need
  \(
    y_L / z
    \geq (-\fmin + \ftwiceerrnearpos(-\fmin)/2)
    = (-\fmin + \fmin / 2)
    = -\fmin / 2
  \).
  As before, since $y_L$ and $v_L$ are negative,
  we can conclude that $z$ is positive:
  hence $y_L \geq (-\fmin / 2) \cdot z$.
  Therefore, $z \geq y_L / (-\fmin/2) = z \geq (y_L / -\fmin) \cdot 2$.
  Since in this case
  \(
    \evalup{(y_L / -\fmin) \cdot 2}
    = \roundup{(y_L / -\fmin) \cdot 2}
    = (y_L \mdiv_\rup -\fmin) \cdot 2
  \),
  by~\eqref{case:x-geq-e-implies-x-sgeq-evalup-e}
  of Proposition~\ref{prop:fp-approx-of-real-constraints},
  we can conclude $y \sgeq (y_L \mdiv_\rup -\fmin) \cdot 2$.
\end{description}

As for the case analysis of $\invsecdivl$
(Figure~\ref{fig:second-indirect-projection-division-zl}) marked as $a_5$,
we must distinguish between the following cases:
\begin{description}
\item[$\bar{r}_\ell = \rup:$]
  we must have $z = +\infty$ in order to obtain $x = +0$;
\item[$\bar{r}_\ell = \rdown:$]
  inequality $y_L / z < \fmin$ must hold and thus,
  since positive $y_L$ and $v_L$ imply a positive $z$,
  $z > y_L / \fmin$ and,
  by~\eqref{prop:fp-approx-of-real-constraints:2}
  of Proposition~\ref{prop:fp-approx-of-real-constraints},
  $z \sgeq \fsucc(y_L \mdiv_\rdown \fmin)$.
\item[$\bar{r}_\ell = \rnear:$]
  since $\fodd(\fmin)$,
  for $v_L = +0$ we need $y_L / z \leq \fmin/2$.
  As $z$ is positive in this case, $(y_L / \fmin) \cdot 2 \leq z$.
  Since
  \(
    \evalup{(y_L / \fmin) \cdot 2}
    = \roundup{(y_L / \fmin) \cdot 2}
    = (y_L \mdiv_\rup \fmin) \cdot 2
  \),
  by~\eqref{case:x-geq-e-implies-x-sgeq-evalup-e}
  of Proposition~\ref{prop:fp-approx-of-real-constraints},
  we can conclude $y \sgeq (y_L \mdiv_\rup \fmin) \cdot 2$.
\end{description}

Concerning the case analysis of $\invsecdivl$
(Fig~\ref{fig:second-indirect-projection-division-zl}) marked as $a_6$,
we must distinguish between the following cases:
\begin{description}
\item[$\bar{r}_\ell = \rdown:$]
  the lowest value of $z$ that gives $x = +\infty$
  with $y \in \Rset_-$ is $z = -0$;
\item[$\bar{r}_\ell = \rup:$]
  inequality $ y_L / z > \fmax$ must hold;
  since $y_L$ is negative and $v_L$ is positive, $z$ must be negative,
  and therefore $y_L < \fmax \cdot z$.
  Hence, $y_L / \fmax < z$.
  By~\eqref{prop:fp-approx-of-real-constraints:2}
  of Proposition~\ref{prop:fp-approx-of-real-constraints},
  we obtain $z \sgeq \fsucc(y_L \mdiv_\rdown \fmax)$.
\item[$\bar{r}_\ell = \rnear:$]
  since $\fodd(\fmax)$, for $v_L = +\infty$ we need
  $y_L / z \geq (\fmax +  \ftwiceerrnearpos(\fmax)/2)$.
  As before, since $w_L$ is negative and $v_L$ is positive,
  we can conclude that $z$ is negative, and, therefore,
   $y_L \leq (\fmax + \ftwiceerrnearpos(\fmax)/2) \cdot z$ holds.
   As a consequence, $y_L / (\fmax + \ftwiceerrnearpos(\fmax)/2) \leq z$.
   If
   \(
     \evalup{y_L / (\fmax + \ftwiceerrnearpos(\fmax)/2)}
     = \roundup{y_L / (\fmax + \ftwiceerrnearpos(\fmax)/2)}
   \),
   by~\eqref{case:x-geq-e-implies-x-sgeq-evalup-e}
   of Proposition~\ref{prop:fp-approx-of-real-constraints},
   we can conclude $z \sgeq \evalup{y_L / (\fmax + \ftwiceerrnearpos(\fmax)/2)}$.
   On the other hand, if
   \(
     \evalup{y_L / (\fmax + \ftwiceerrnearpos(\fmax)/2)}
     \neq \roundup{y_L / (\fmax + \ftwiceerrnearpos(\fmax)/2)}
   \) then,
   we can only apply~\eqref{prop:fp-approx-of-real-constraints:1}
   of Proposition~\ref{prop:fp-approx-of-real-constraints},
   obtaining $z \sgeq \evaldown{y_L / (\fmax + \ftwiceerrnearpos(\fmax)/2)}$.
\end{description}

Regarding the case analysis of $\invsecdivl$
(Fig~\ref{fig:second-indirect-projection-division-zl}) marked as $a_7$,
we have the following cases:
\begin{description}
\item[$\bar{r}_\ell = \rup:$]
  the lowest value of $z$ that yields $x = -\infty$ with $y \in \Rset_+$ is $z = -0$;
\item[$\bar{r}_\ell = \rdown:$]
  inequality $y_L / z < -\fmax$ must hold and thus,
  since a positive $y_L$ and a negative $v_L$ imply that the sign of $z$ is negative,
  $y_L > -\fmax \cdot z$.
  Hence, $y_L / -\fmax < z$.
  By~\eqref{prop:fp-approx-of-real-constraints:2}
  of Proposition~\ref{prop:fp-approx-of-real-constraints},
  $z \sgeq \fsucc(y_L \mdiv_\rdown -\fmax)$.
\item[$\bar{r}_\ell = \rnear:$]
  since $\fodd(-\fmax)$, for $v_L = -\infty$ we need
  $y_L / z \leq -\fmax + \ftwiceerrnearneg(-\fmax)/2$.
  Since $z$ in this case is negative, we obtain the inequality
  $z \geq y_L / (-\fmax + \ftwiceerrnearneg(-\fmax)/2)$.
  If
  \(
    \evalup{y_L / \bigl(-\fmax + \ftwiceerrnearneg(-\fmax)/2\bigr)}
    = \roundup{y_L / \bigl(-\fmax + \ftwiceerrnearneg(-\fmax)/2\bigr)}
  \),
  by~\eqref{case:x-geq-e-implies-x-sgeq-evalup-e}
  of Proposition~\ref{prop:fp-approx-of-real-constraints},
  we can conclude
  $y \sgeq \evalup{y_L / \bigl(-\fmax + \ftwiceerrnearneg(-\fmax)/2\bigr)}$.
  On the other hand, if
  \(
    \evalup{y_L / \bigl(-\fmax + \ftwiceerrnearneg(-\fmax)/2\bigr)}
    \neq \roundup{y_L / \bigl(-\fmax + \ftwiceerrnearneg(-\fmax)/2\bigr)}
  \),
  then we can only apply~\eqref{prop:fp-approx-of-real-constraints:1}
  of Proposition~\ref{prop:fp-approx-of-real-constraints},
  obtaining $y \sgeq \evaldown{y_L / \bigl(-\fmax + \ftwiceerrnearneg(-\fmax)/2\bigr)}$.
\end{description}

Similar arguments can be used to prove the case analyses
of function $\invsecdivu$ of Figure~\ref{fig:second-indirect-projection-division-zu}
marked as $a_9$, $a_{10}$, $a_{11}$ and $a_{12}$.

We will now analyze the case analyses of $\invsecdivl$
of Figure~\ref{fig:second-indirect-projection-division-zl}
marked as $a^-_3$ and $a^+_3$, and the ones of $\invsecdivu$
of Figure~\ref{fig:first-indirect-projection-division-yu} marked as $a^-_8$ and $a^+_8$.
In this proof, we can assume
$y_L, v_L \in \Rset_- \cup \Rset_+$, $y_U, v_U \in \Rset_- \cup \Rset_+$
and $\sgn(v_L) = \sgn(v_U)$.
First, note that the argument that leads
to~\eqref{eq:division-firstandsecond-indirect-projection-III}
and~\eqref{eq:division-firstandsecond-indirect-projection-IV}
starting from $x \sleq y \mdiv z$ and $x \sgeq y \mdiv z$
is in common with the proof of Theorem~\ref{teo:first-indirect-projection-division}.

Provided that interval $X$ is split into intervals $X_+$ and $X_-$,
it is worth discussing the reasons why it is not necessary to partition also $Y$
directly in Algorithm~\ref{algo7}.
Assume $Y = [-a, b]$ with $a, b > 0$ and consider the partition of $Y$
into two sign-homogeneus intervals $Y \cap [-\infty, -0]$
and $Y \cap [+0, +\infty]$, as usual.
Note that the values $-0 \in Y \cap [-\infty, -0] = [-a, -0]$
and the values $+0 \in Y \cap [+0, +\infty] = [+0, b]$
can never be the boundaries of $Y$ upon which $z$ touches its minimum (resp., maximum).
This is because $y$ will be the numerator of fractions
(see expressions~\eqref{eq:division-second-indirect-projection-I}
and~\eqref{eq:division-second-indirect-projection-II}).
Moreover, by the definition of functions
$\invsecdivl$ of Fig~\ref{fig:second-indirect-projection-division-zl}
and $\invsecdivu$ of Fig~\ref{fig:second-indirect-projection-division-zu},
it easy to verify that the partition of $Y$ would not prevent
the interval computed for $y$ from being equal to the empty set.
That is, if $\invsecdivl(y_L, v_L, \bar{r}_\ell) = \uns$ or
$\invsecdivu(y_U, v_U, \bar{r}_u) = \uns$,
then partitioning also $Y$ into sign-homogeneus intervals
and then applying the procedure of Algorithm~\ref{algo7}
to the two distinct intervals results again into an empty refining interval for $z$.

Hence, to improve efficiency, Algorithm~\ref{algo7} does not split
interval $Y$ into sign-homogeneous intervals.
However, in this proof it is necessary to partition $Y$
into intervals $Y_-$ and $Y_+$ in order to determine
the correct formulas for lower and upper bounds for $z$.
In the following, for the sake of simplicity,
we will analyze the special case $X_+$ and $Y = Y_+$,
so that $Y$ does not need to be split because it is already a sign-homogeneous interval.
The remaining cases in which $Y$ is sign-homogeneous as well as
those in which it is not can be derived analogously.
To sum up, in this case we assume $x \geq 0$ and $y \geq 0$, and therefore $z > 0$.

Now, we need to prove the cases marked as $a_3^+$ and $a_8^+$.
The case analysis of~\eqref{eq:division-firstandsecond-indirect-projection-I}
and~\eqref{eq:division-firstandsecond-indirect-projection-II}
yields~\eqref{eq:division-firstandsecond-indirect-projection-III}
and~\eqref{eq:division-firstandsecond-indirect-projection-IV}.
Remember that the case $x = \pm 0$ is handled separately by functions
$\invsecdivl$ of Figure~\ref{fig:second-indirect-projection-division-zl}
and $\invsecdivu$ of Figure~\ref{fig:second-indirect-projection-division-zu},
hence assuming $x > 0$, we obtain
\begin{align}
  z
   &\begin{cases}
      \mathord{} \leq y/x,
        &\text{if $\bar{r}_u = \rdown$;} \\
      \mathord{} < y / \fpred(x),
        &\text{if $\bar{r}_u = \rup$ and $x \neq \fmin$;} \\
      \mathord{} \leq \fmax,
        &\text{if $\bar{r}_u = \rup$ and $x = \fmin$;} \\
      \mathord{} \leq y / \bigl(x + \ftwiceerrnearneg(x)/2 \bigr),
        &\text{if $\bar{r}_u = \rnear$ and $\feven(x)$;} \\
      \mathord{} < y / \bigl(x + \ftwiceerrnearneg(x)/2 \bigr),
        &\text{if $\bar{r}_u = \rnear$ and $\fodd(x)$;}
  \end{cases}
     \label{eq:division-second-indirect-projection-I} \\
         z
   &\begin{cases}
      \mathord{} > y / \fsucc(x),
        &\text{if $\bar{r}_\ell = \rdown$ and $x \neq -\fmin$;} \\
          \mathord{} \geq -\fmax,
        &\text{if $\bar{r}_\ell = \rdown$ and $x = -\fmin$;} \\
       \mathord{} \geq y / x ,
        &\text{if $\bar{r}_\ell = \rup$;} \\
      \mathord{} \geq y / \bigl(x + \ftwiceerrnearpos(x)/2 \bigr),
        &\text{if $\bar{r}_\ell = \rnear$ and $\feven(x)$;} \\
      \mathord{} > y / \bigl(x + \ftwiceerrnearpos(x)/2 \bigr),
        &\text{if $\bar{r}_\ell = \rnear$ and $\fodd(x)$.}
    \end{cases}
    \label{eq:division-second-indirect-projection-II}
\end{align}

Since the members of the divisions are
independent, we can find the minimum of said divisions by minimizing
each one of their members.
Let $(y_L, y_U, v_L, v_U)$ be as returned by function $\tau$ of
Figure~\ref{fig:the-tau-function}.
By Proposition~\ref{prop:min-max-x-ftwiceerrnearneg-ftwiceerrnearpos} and
the monotonicity of `$\fpred$' and `$\fsucc$' we obtain
\begin{align}
  z
   &\begin{cases}
      \mathord{} \leq y_U/v_U,
        &\text{if $\bar{r}_u = \rdown$;} \\
      \mathord{} < y_U / \fpred(v_U),
        &\text{if $\bar{r}_u = \rup$ and $v_U \neq \fmin$;} \\
        \mathord{} \leq \fmax,
        &\text{if $\bar{r}_u = \rup$ and $v_U = \fmin$;} \\
      \mathord{} \leq y_U /\bigl(v_U + \ftwiceerrnearneg(v_U)/2\bigr),
        &\text{if $\bar{r}_u = \rnear$ and $\feven(v_U)$;} \\
      \mathord{} < y_U / \bigl( v_U + \ftwiceerrnearneg(v_U)/2 \bigr),
        &\text{if $\bar{r}_u = \rnear$ and $\fodd(v_U)$;}
    \end{cases} \\
    z
   &\begin{cases}
      \mathord{} > y_L / \fsucc(v_L),
        &\text{if $\bar{r}_\ell = \rdown$ and $v_L\neq-\fmin$;} \\
          \mathord{} \geq -\fmax,
        &\text{if $\bar{r}_\ell = \rdown$ and $v_L=-\fmin$;}  \\
       \mathord{} \geq y_L / v_L ,
        &\text{if $\bar{r}_\ell = \rup$;} \\
      \mathord{} \geq y_L / \bigl(v_L + \ftwiceerrnearpos(v_L)/2 \bigr),
        &\text{if $\bar{r}_\ell = \rnear$ and $\feven(v_L)$;}\\
      \mathord{} > y_L / \bigl(v_L + \ftwiceerrnearpos(v_L)/2 \bigr),
        &\text{if $\bar{r}_\ell = \rnear$ and $\fodd(v_L)$.}
    \end{cases}
\end{align}

We can now exploit Proposition~\ref{prop:fp-approx-of-real-constraints}
and obtain:
\begin{align}
\label{eq:div-inv-2-num-rup-rdown-lower}
  z'_\ell &\defeq
    \begin{cases}
      y_L \mdiv_\rup v_L,
        &\text{if $\bar{r}_\ell = \rup$;} \\
      \fsucc\bigl(y_L \mdiv_\rdown \fsucc(v_L) \bigr),
        &\text{if $\bar{r}_\ell = \rdown$ and $v_L\neq-\fmin$;}
    \end{cases} \\
 z'_u &\defeq
    \begin{cases}
      \fpred\bigl(y_U \mdiv_\rup \fpred(v_U) \bigr),
        &\text{if $\bar{r}_u = \rup$ and $v_U\neq\fmin$;} \\
      y_U \mdiv_\rdown v_U,
        &\text{if $\bar{r}_u = \rdown$.}
        \label{eq:div-inv-2-num-rup-rdown-upper}
    \end{cases}
\end{align}
Since $y_L \neq 0$, then $ y_L / \fsucc(v_L) \neq 0$.
Hence, Proposition~\ref{prop:fp-approx-of-real-constraints}
applies and we have $z \sgeq y_L \mdiv_\rup v_L$ if $\bar{r}_\ell = \rup$ and
$z \sgeq \fsucc\bigl(y_L / \fsucc(v_L) \bigr)$
if $\bar{r}_\ell = \rdown$ and $v_L \neq -\fmin$.
Analogously,
since $y_U \neq 0$, then $y_U / \fpred(v_L) \neq 0$.
Hence, by Proposition~\ref{prop:fp-approx-of-real-constraints}
we obtain~\eqref{eq:div-inv-2-num-rup-rdown-upper}.

Note that, since division by zero is not defined on real numbers,
we had to separately address the case
$\bar{r}_u = \rup$ and $x = \fmin$
in \eqref{eq:division-second-indirect-projection-I},
and the case $\bar{r}_\ell = \rdown$ and $x = -\fmin$
in \eqref{eq:division-second-indirect-projection-II}.
Division by zero is, however, defined on IEEE~754 floating-point numbers.
Indeed, if we evaluate the second case of
\eqref{eq:div-inv-2-num-rup-rdown-lower} with $v_L = -\fmin$,
we obtain
$\fsucc(y_L \mdiv_\rdown \fsucc(-\fmin)) = -\fmax$,
which happens to be the correct value for $z'_\ell$,
provided $y_L > 0$.
The same happens for \eqref{eq:div-inv-2-num-rup-rdown-upper}.
Therefore, there is no need for a separate treatment when variable $x$
takes the values $\pm\fmin$.

As before,
we need to approximate the values of the expressions
$e^+_u \defeq y_U / \bigl(v_U + \ftwiceerrnearneg(v_U)/2 \bigr)$
and
$e^+_\ell \defeq y_L / \bigl(v_L + \ftwiceerrnearpos(v_L)/2 \bigr)$.
Thus, when
$\evaldown{e^+_u} < \rounddown{e^+_u}$
by~\eqref{eq:div-1-refine-real-upper_bounds}
and~\eqref{prop:fp-approx-of-real-constraints:3}
of Proposition~\ref{prop:fp-approx-of-real-constraints}
we obtain $y \sleq \evalup{e^+_u}$, while, when
$\evaldown{e^+_\ell} > \rounddown{e^+_\ell}$
by~\eqref{eq:div-1-refine-real-upper_bounds}
and~\eqref{prop:fp-approx-of-real-constraints:1}
of Proposition~\ref{prop:fp-approx-of-real-constraints}
we obtain $y \sgeq \evaldown{e^+_\ell}$.
Thus, for the case where $\bar{r}_\ell = \rnear$,
since $e^+_u \neq 0$ and $e^+_\ell \neq 0$,
by Proposition~\ref{prop:fp-approx-of-real-constraints}, we have
\begin{align}
\label{eq:div2-inv-num-rnear-lower}
  y'_\ell &\defeq
    \begin{cases}
     \evalup{e^+_\ell},
        &\text{if $\feven(v_L)$ and $\evalup{e^+_\ell} = \roundup{e^+_\ell}$;} \\
     \evaldown{e^+_\ell},
        &\text{if $\feven(v_L)$ and $\evalup{e^+_\ell} \neq \roundup{e^+_\ell}$;} \\
     \fsucc\bigl(\evaldown{e^+_\ell}\bigr),
        &\text{otherwise;}
    \end{cases} \\
\intertext{%
whereas, for the case in which $\bar{r}_u = \rnear$,
}
\label{eq:div2-inv-num-rnear-upper}
  y'_u &\defeq
    \begin{cases}
      \evaldown{e^+_u},
        &\text{if $\feven(v_U)$ and $\evaldown{e^+_u} = \rounddown{e^+_u}$;} \\
      \evalup{e^+_u},
        &\text{if $\feven(v_U)$ and $\evaldown{e^+_u} \neq \rounddown{e^+_u}$;} \\
      \fpred\bigl(\evalup{e^+_u}\bigr),
        &\text{otherwise.}
    \end{cases}
\end{align}

An analogous reasoning allows us to prove the case analyses
marked as $a_3^-$ and $a_8^-$.
\end{delayedproof}

\begin{delayedproof}[of Theorem~\ref{teo:direct-projection-mult}]
Given the constraint $x = y \mmul_S z$ with
$x \in X = [x_\ell, x_u]$,
$y \in Y = [y_\ell, y_u]$ and
$z \in Z = [z_\ell, z_u]$,
then $X' = [x'_\ell, x'_u] \inters X$.
Hence, we are sure that $X' \sseq X$.

It should be immediate to verify that
function $\sigma$ of Figure~\ref{fig:the-sigma-function},
related to the case $\sgn(y_\ell) = \sgn(y_u)$,
chooses the appropriate interval extrema $y_L, y_U, z_L, z_U$,
necessary for computing bounds for $x$.
Indeed, note that such choice is completely driven by the sign of the resulting product.
Analogously, the correct interval extrema $y_L, y_U, z_L, z_U$
related to the case $\sgn(z_\ell) = \sgn(z_u)$
can be determined by applying function $\sigma$ of Figure~\ref{fig:the-sigma-function},
but swapping the role of $y$ and $z$.
Hence, if the sign of $y$ or of $z$ is constant
(see the second part of Algorithm~\ref{algo3})
function $\sigma$ of Figure~\ref{fig:the-sigma-function}
finds the appropriate extrema for $y$ and $z$ to compute the bound for $x$.

Concerning the cases $\sgn(y_\ell) = \sgn(z_\ell) = -1$
and $\sgn(y_u) = \sgn(z_u) = 1$ (first part of Algorithm~\ref{algo3}),
note that we have only two possibilities for the interval extrema $y_L$ and $z_L$,
that are $y_\ell$ and $z_u$ or $y_u$ and $z_\ell$.
Since the product of $y_L$ and $z_L$ will have a negative sign in both cases,
the right extrema for determining the lower bound $x'_\ell$
have to be chosen by selecting the smallest product of $y_L$ and $z_L$.
Analogously, for $y_U$ and $z_U$ there are two possibilities:
$y_\ell$ and $z_\ell$ or $y_u$ and $z_u$.
Since the product of $y_U$ and $z_U$ will have a positive sign in both cases,
the appropriate extrema for determining the upper bound $x'_u$
have to be chosen as the biggest product of $y_U$ and $z_U$.

Remember that by Proposition~\ref{prop:worst-case-rounding-modes},
following the same reasoning as in the previous proofs,
it suffices to find a lower bound for $y_L \mmul_{r_\ell} z_L$
and an upper bound for $y_U \mmul_{r_u} z_U$.

We now comment on some critical case analyses
of function $\dirmull$ of Figure~\ref{fig:direct-projection-multiplication}.
Consider, for example, when $y_L = \pm \infty$ and $z_L = \pm 0$.
In particular, we analyze the case in which $y_L = -\infty$ and $z_L = \pm 0$.
Note that $y_L = -\infty$ implies $y_\ell = -\infty$.
Assume, first, that $z_L = +0$.
Recall that by the IEEE~754 Standard \cite{IEEE-754-2008}
$\pm \infty \mmul \pm 0$ is an invalid operation.
However, since $y_\ell = -\infty$, we have two cases:
\begin{description}
\item[$y_u \geq -\fmax:$]
  note that, in this case, $-\fmax \mmul +0 = -0$;
\item[$y_u = -\infty:$]
  in this case, $z_L$ must correspond to $z_u$
  (see the last three cases of function $\sigma$).
  Since $-\infty \mmul z$ for $z < 0$ results in $+\infty$,
  we can conclude that $-0$ is a correct lower bound for $x$.
\end{description}

A similar reasoning applies for the cases $y_L = +\infty$, $z_L = \pm 0$.
Dually, the only critical entries of function
$\dirmulu$ of Figure~\ref{fig:direct-projection-multiplication}
are those in which $y_U = \pm\infty$ and $z_U = \pm 0$.
In these cases we can reason in a similar way, too.

We are left to prove that
$\forall X'' \subset X \itc \exists r \in S, y \in Y, z \in Z \st y \mmul_r z \not\in X''$.
Let us focus on the lower bound $x'_\ell$,
proving that there exist values $r \in S, y \in Y, z \in Z$
such that $y \mmul_r z = x'_\ell$.
Consider the particular values of $y_L$, $z_L$ and $r_\ell$
that correspond to the value of $x'_\ell$ chosen by Algorithm~\ref{algo3},
that is $y_L$, $z_L$ and $r_\ell$ are such that $\dirmull(y_L, z_L, r_\ell) = x'_\ell$.
By Algorithm~\ref{algo3}, such values of $y_L$ and $z_L$ must exist.
First, consider the cases in which $y_L \not\in (\Rset_- \cup \Rset_+)$
or $z_L \not\in (\Rset_- \cup \Rset_+)$.
In these cases, a brute-force verification was successfully conducted
to verify that $y \mmul_{r_\ell} z = x'_\ell$.
For the cases in which $y_L \in (\Rset_- \cup \Rset_+)$
and $z_L \in (\Rset_- \cup \Rset_+)$ we have,
by definition of $\dirmull$ of Figure~\ref{fig:direct-projection-multiplication},
that $x'_\ell = y_L \mmul_{r_\ell} z_L$.
Remember that, by Proposition~\ref{prop:worst-case-rounding-modes},
there exist $r' \in S$ such that
$y_L \mmul_\mathrm{r_\ell} z_L = y_L \mmul_\mathrm{r'} z_L$.
Since $y_L \in Y$ and $z_L \in Z$, we can conclude that
$x'_\ell \not\in X''$ implies that $y'_L \mmul_\mathrm{r'} z_L \not\in X''$.
An analogous reasoning allows us to conclude that $\exists r \in S$
for which the following holds:
$x'_u \not\in X''$ implies $y_U \mmul_\mathrm{r} z_U \not\in X''$.
\end{delayedproof}

\begin{delayedproof}[of Theorem~\ref{teo:indirect-projection-mult}]
Given the constraint $x = y \mmul_S z$ with
$x \in X = [x_\ell, x_u]$,
$y \in Y = [y_\ell, y_u]$ and
$z \in Z = [z_\ell, z_u]$,
Algorithm~\ref{algo4} computes $Y'$,
a new and refined interval for variable $y$.

First, note that either
$Y' \assign (Y \inters [y^-_\ell, y^-_u]) \biguplus (Y \inters[y^+_\ell, y^+_u])$
or $Y' = \emptyset$, hence, in both cases, we are sure that $Y' \sseq Y$ holds.

By Proposition~\ref{prop:worst-case-rounding-modes},
we can focus on finding a lower bound for $y \in Y$
by exploiting the constraint $y \mmul_{\bar{r}_\ell} z=x$
and an upper bound for $y \in Y$
by exploiting the constraint $y \mmul_{\bar{r}_u} z=x$.

Now, in order to compute correct bounds for $y$,
we first need to split the interval of $z$
into the sign-homogeneous intervals $Z_-$ and $Z_+$,
because it is crucial to be sure of the sign of $z$.
As a consequence, for $W=Y_-$ (and, analogously, for $W=Y_+$)
function $\tau$ of Figure~\ref{fig:the-tau-function} picks
the appropriate interval extrema of $X$ and $W$ to be used
to compute the new lower and upper bounds for $y$.
It is easy to verify that the values of $x_L$ and $w_L$ (resp., $x_U$ and $w_U$)
computed using function $\tau$ of Figure~\ref{fig:the-tau-function}
are the boundaries of $X$ and $W$ upon which $y$ touches its
minimum (resp., maximum).
Functions $\invmull$ of Figure~\ref{fig:indirect-projection-multiplication-yl}
and $\invmulu$ of Figure~\ref{fig:indirect-projection-multiplication-yu}
are then employed to find the new bounds for $y$.
The so obtained intervals for $y$ are then joined using
convex union between intervals, in order to obtain the refining interval for $y$.

Observe that functions $\invmull$ of Fig~\ref{fig:indirect-projection-multiplication-yl}
and $\invmulu$ of Fig~\ref{fig:indirect-projection-multiplication-yu}
are dual to each other:
every row/column of one table can be found in the other table
reversed and changed of sign.
This is due to the fact that, for each $r \in R$ and each
$D \sseq \Fset \times \Fset$, we have
\begin{align*}
  &\min \bigl\{\,
         y \in \Fset
       \bigm|
         (x, z) \in D, x = y \mmul_r z
       \,\bigr\} \\
  =
  &-\max \bigl\{\,
          y \in \Fset
        \bigm|
          (x, z) \in D, -x = y \mmul_r z
        \,\bigr\} \\
  =
  &-\max \bigl\{\,
          y \in \Fset
        \bigm|
          (x, z) \in D, x = y \mmul_r -z
        \,\bigr\}.
\end{align*}

Concerning the case analysis of $\invmull$
marked as $a_4$ of Fig~\ref{fig:indirect-projection-multiplication-yl},
we must consider the following cases:
\begin{description}
\item[$\bar{r}_\ell = \rup:$]
  we clearly must have $y = +\infty$ in this case;
\item[$\bar{r}_\ell = \rdown:$]
  inequality $y \cdot w_L < - \fmax$ must hold and thus, since $w_L$ is negative,
  $y > -\fmax / w_L$ and, by~\eqref{prop:fp-approx-of-real-constraints:2}
  of Proposition~\ref{prop:fp-approx-of-real-constraints},
  $y \sgeq \fsucc(-\fmax \mdiv_\rdown w_L)$.
\item[$\bar{r}_\ell = \rnear:$]
  since $\fodd(\fmax)$, for $x_L = -\infty$ we need $y$ to be greater
  or equal than $\bigl(-\fmax + \ftwiceerrnearneg(-\fmax)/2\bigr) / w_L$. If
  \(
    \evalup{\bigl(-\fmax + \ftwiceerrnearneg(-\fmax)/2\bigr) / w_L}
    = \roundup{-\fmax + \ftwiceerrnearneg(-\fmax)/2\bigr) / w_L}
  \),
  by~\eqref{case:x-geq-e-implies-x-sgeq-evalup-e}
  of Proposition~\ref{prop:fp-approx-of-real-constraints},
  we can conclude $y \sgeq \evalup{\bigl(-\fmax + \ftwiceerrnearneg(-\fmax)/2\bigr)/w_L}$.
  On the other hand, if
  \(
    \evalup{\bigl(-\fmax + \ftwiceerrnearneg(-\fmax)/2\bigr) / w_L}
    \neq \roundup{-\fmax + \ftwiceerrnearneg(-\fmax)/2\bigr) / w_L}
  \),
  then we can only apply~\eqref{prop:fp-approx-of-real-constraints:1}
  of Proposition~\ref{prop:fp-approx-of-real-constraints},
  obtaining $y \sgeq \evaldown{\bigl(-\fmax + \ftwiceerrnearneg(-\fmax)/2\bigr)/w_L}$.
\end{description}

Regarding the case analysis of $\invmull$
marked a  $a_5$ of Fig~\ref{fig:indirect-projection-multiplication-yl},
we have the following cases:
\begin{description}
\item[$\bar{r}_\ell = \rdown:$]
  in this case, we must have $y = -0$;
\item[$\bar{r}_\ell = \rup:$]
  inequality $y \cdot w_L > - \fmin$ must hold and thus, since $w_L$ is positive,
  $y > -\fmin / w_L$ and, by~\eqref{prop:fp-approx-of-real-constraints:4}
  of Proposition~\ref{prop:fp-approx-of-real-constraints},
  $y \sgeq \fsucc(-\fmin \mdiv_\rdown w_L)$.
\item[$\bar{r}_\ell = \rnear:$]
  since $\fodd(\fmin)$, for $x_L=-0$ we need $y$ to be greater or equal than
  $-\fmin  / (2 \cdot w_L)$.
  Since in this case
  \(
    \evalup{-\fmin / (2 \cdot w_L)}
    = \roundup{-\fmin / (2 \cdot w_L)}
    = (-\fmin) \mdiv_\rup (2 \cdot w_L)
  \),
  by~\eqref{case:x-geq-e-implies-x-sgeq-evalup-e}
  of Proposition~\ref{prop:fp-approx-of-real-constraints},
  we can conclude $y \sgeq -\fmin \mdiv_\rup (2 \cdot w_L)$.
\end{description}

As for the case analysis of $\invmull$
marked as $a_6$ of Figure~\ref{fig:indirect-projection-multiplication-yl},
the following cases must be studied:
\begin{description}
\item[$\bar{r}_\ell = \rup:$]
  we must have $y = +0$ in this case;
\item[$\bar{r}_\ell = \rdown:$]
  it should be $y \cdot w_L < \fmin$ and thus, since $w_L$ is negative,
  $y > \fmin / w_L$ and, by~\eqref{prop:fp-approx-of-real-constraints:4}
  of Proposition~\ref{prop:fp-approx-of-real-constraints},
  $y \sgeq \fsucc(-\fmin \mdiv_\rdown w_L)$.
\item[$\bar{r}_\ell = \rnear:$]
  since $\fodd(\fmin)$, for $x_L = -0$ we need $y$ to be greater than
  or equal to $\bigl(\fmin / (2 \cdot w_L)\bigr)$.
  Since in this case
  \(
    \evalup{\fmin / (2 \cdot w_L)}
    = \roundup{\fmin / (2 \cdot w_L)}
    = \fmin \mdiv_\rup (2 \cdot w_L)
  \),
  by~\eqref{case:x-geq-e-implies-x-sgeq-evalup-e}
  of Proposition~\ref{prop:fp-approx-of-real-constraints}
  we can conclude $y \sgeq \fmin \mdiv_\rup (2 \cdot w_L)$.
\end{description}

Finally, for the case analysis of $\invmull$
marked as $a_7$ of Fig~\ref{fig:indirect-projection-multiplication-yl},
the following cases must be considered:
\begin{description}
\item[$\bar{r}_\ell = \rdown:$]
  in this case we must have $y = +\infty$;
\item[$\bar{r}_\ell = \rup:$]
  it should be $y \cdot w_L > -\fmax$ and thus, since $w_L$ is positive,
  $y > \fmax / w_L$ and, by~\eqref{prop:fp-approx-of-real-constraints:2}
  of Proposition~\ref{prop:fp-approx-of-real-constraints},
  $y \sgeq \fsucc(\fmax \mdiv_\rdown w_L)$.
\item[$\bar{r}_\ell = \rnear:$]
  since $\fodd(\fmax)$, for $x_L = +\infty$ we need $y$ to be greater than
  or equal to $\bigl(\fmax + \ftwiceerrnearpos(\fmax)/2\bigr) / w_L$. If
  \(
    \evalup{\bigl(\fmax + \ftwiceerrnearpos(\fmax)/2\bigr) / w_L}
    = \roundup{\fmax + \ftwiceerrnearpos(\fmax)/2\bigr) / w_L}
  \),
  by~\eqref{case:x-geq-e-implies-x-sgeq-evalup-e}
  of Proposition~\ref{prop:fp-approx-of-real-constraints},
  we can conclude $y \sgeq \evalup{\bigl(\fmax + \ftwiceerrnearpos(\fmax)/2\bigr)/w_L}$.
  On the other hand, if
  \(
    \evalup{\bigl(\fmax + \ftwiceerrnearpos(\fmax)/2\bigr) / w_L}
    \neq \roundup{\fmax + \ftwiceerrnearpos(\fmax)/2\bigr) / w_L}
  \),
  then we can only apply~\eqref{prop:fp-approx-of-real-constraints:1}
  of Proposition~\ref{prop:fp-approx-of-real-constraints}, obtaining
  $y \sgeq \evaldown{\bigl(\fmax + \ftwiceerrnearpos(\fmax)/2\bigr)/w_L}$.
\end{description}

Similar arguments can be used to prove the case analyses of $\invmulu$
of Figure~\ref{fig:indirect-projection-multiplication-yu}
marked as $a_9$, $a_{10}$, $a_{11}$ and $a_{12}$.

We now analyze the case analyses of $\invmull$
of Fig~\ref{fig:indirect-projection-multiplication-yl}
marked as $a^-_3$ and $a^+_3$ and the ones
of $\invmulu$ of Fig~\ref{fig:indirect-projection-multiplication-yu}
marked as $a^-_8$ and $a^+_8$, for which we can assume
$x_L, w_L \in \Fset \inters \Rset$ and $x_U, w_U \in \Fset \inters \Rset$,
and $\sgn(w_\ell) = \sgn(w_u)$.
Exploiting $x \sleq y \mmul z$ and $x \sgeq y \mmul z$,
by Proposition~\ref{prop:real-approx-of-fp-constraints} we have
\begin{align}
    y \cdot z
  &\begin{cases}
      \mathord{} \geq x,
        &\text{if $\bar{r}_\ell = \rdown$;} \\
      \mathord{} > x + \ferrup(x) = \fpred(x),
        &\text{if $\bar{r}_\ell = \rup$;} \\
      \mathord{} \geq x + \ftwiceerrnearneg(x)/2,
        &\text{if $\bar{r}_\ell = \rnear$ and $\feven(x)$;} \\
      \mathord{} > x + \ftwiceerrnearneg(x)/2,
        &\text{if $\bar{r}_\ell = \rnear$ and $\fodd(x)$.}
    \end{cases} \\
    y \cdot z
    &\begin{cases}
      \mathord{} < x + \ferrdown(x) = \fsucc(x),
        &\text{if $\bar{r}_u = \rdown$;} \\
      \mathord{} \leq x,
        &\text{if $\bar{r}_u = \rup$;} \\
      \mathord{} \leq x + \ftwiceerrnearpos(x)/2,
        &\text{if $\bar{r}_u = \rnear$ and $\feven(x)$;} \\
      \mathord{} < x + \ftwiceerrnearpos(x)/2,
        &\text{if $\bar{r}_u = \rnear$ and $\fodd(x)$.}
    \end{cases}
\end{align}
Since the case $z = 0$ is handled separately
by $\invmull$ of Figure~\ref{fig:indirect-projection-multiplication-yl}
and by $\invmulu$ of Figure~\ref{fig:indirect-projection-multiplication-yu},
we can assume $z \neq 0$.
Thanks to the splitting of $Z$ into a positive and a negative part,
the sign of $z$ is determined.
In the following, we will prove the case analyses marked as $a_3^+$ and $a_8^+$.
Hence, assuming $z > 0$, the previous case analysis gives us
\begin{align}
  y
   &\begin{cases}
      \mathord{} \geq x/z,
        &\text{if $\bar{r}_\ell = \rdown$;} \\
      \mathord{} > \fpred(x)/z,
        &\text{if $\bar{r}_\ell = \rup$} \\
      \mathord{} \geq \bigl(x + \ftwiceerrnearneg(x)/2\bigr) / z,
        &\text{if $\bar{r}_\ell = \rnear$ and $\feven(x)$;} \\
      \mathord{} > \bigl( x + \ftwiceerrnearneg(x)/2 \bigr)/ z,
        &\text{if $\bar{r}_\ell = \rnear$ and $\fodd(x)$;}
    \end{cases} \\
    y
   &\begin{cases}
      \mathord{} < \fsucc(x)/z,
        &\text{if $\bar{r}_u = \rdown$;} \\
       \mathord{} \leq x / z,
        &\text{if $\bar{r}_u = \rup$;} \\
      \mathord{} \leq \bigl(x + \ftwiceerrnearpos(x)/2 \bigr)/ z,
        &\text{if $\bar{r}_u = \rnear$ and $\feven(x)$;} \\
      \mathord{} < \bigl(x + \ftwiceerrnearpos(x)/2 \bigr)/ z,
        &\text{if $\bar{r}_u = \rnear$ and $\fodd(x)$.}
    \end{cases}
\end{align}

Note that the numerator and the denominator of the previous fractions are
independent. Therefore, we can find the minimum of the fractions by minimizing
the numerator and maximizing the denominator. Since we are analyzing the case
in which $W = Z_+$, let $(x_L, w_L, x_U, w_U)$
as the result of function $\tau$ of Figure~\ref{fig:the-tau-function}.
Hence, by Proposition~\ref{prop:min-max-x-ftwiceerrnearneg-ftwiceerrnearpos} and
the monotonicity of `$\fpred$' and `$\fsucc$ we obtain
\begin{align}
  y
   &\begin{cases}
      \mathord{} \geq x_L/w_L,
        &\text{if $\bar{r}_\ell = \rdown$;} \\
      \mathord{} > \fpred(x_L)/w_L,
        &\text{if $\bar{r}_\ell = \rup$} \\
      \mathord{} \geq \bigl(x_L + \ftwiceerrnearneg(x_L)/2\bigr) / w_L,
        &\text{if $\bar{r}_\ell = \rnear$ and $\feven(x)$;} \\
      \mathord{} > \bigl( x_L + \ftwiceerrnearneg(x_L)/2 \bigr)/ w_L,
        &\text{if $\bar{r}_\ell = \rnear$ and $\fodd(x)$;}
    \end{cases} \\
    y
   &\begin{cases}
      \mathord{} < \fsucc(x_U)/w_U,
        &\text{if $\bar{r}_u = \rdown$;} \\
       \mathord{} \leq x_U / w_U,
        &\text{if $\bar{r}_u = \rup$;} \\
\label{eq:mult-refine-real-upper_bounds}
      \mathord{} \leq \bigl(x_U + \ftwiceerrnearpos(x_U)/2 \bigr)/ w_U,
        &\text{if $\bar{r}_u = \rnear$ and $\feven(x)$;} \\
      \mathord{} < \bigl(x_U + \ftwiceerrnearpos(x_U)/2 \bigr) /w_U,
        &\text{if $\bar{r}_u = \rnear$ and $\fodd(x)$.}
    \end{cases}
\end{align}

We can now exploit Proposition~\ref{prop:fp-approx-of-real-constraints}
and obtain:
\begin{align}
\label{eq:mult-inv-num-rup-rdown-lower}
  y'_\ell &\defeq
    \begin{cases}
      x_L \mdiv_\rup w_L,
        &\text{if $\bar{r}_\ell = \rdown$;} \\
      \fsucc\bigl(\fpred(x_L) \mdiv_\rdown w_L\bigr),
        &\text{if $\bar{r}_\ell = \rup$;}
    \end{cases} \\
\label{eq:mult-inv-num-rup-rdown-upper}
     y'_u &\defeq
    \begin{cases}
      \fpred\bigl(\fsucc(x_U) \mdiv_\rup w_U\bigr),
        &\text{if $\bar{r}_u = \rdown$;} \\
      x_U \mdiv_\rdown w_U,
        &\text{if $\bar{r}_u = \rup$.}
    \end{cases}
 \end{align}
Indeed, if $x_L \neq 0$, then Proposition~\ref{prop:fp-approx-of-real-constraints}
applies and we have $y \sgeq x_L \mdiv_\rup w_L$.
On the other hand, if $x_L = 0$, since by hypothesis $z > 0$ implies  $w_L > 0$,
according to IEEE~754 \cite[Section~6.3]{IEEE-754-2008},
we have $(x_L \mdiv_\rup w_L) = \sgn(x_L) \cdot 0$ and, indeed,
for each non-NaN, nonzero and finite
$w \in \Fset \inters [+0, +\infty]$, $\sgn(x_L) \cdot 0$
is the least value for $y$ that satisfies
$\sgn(x_L) \cdot 0 = y \mmul_\rdown w$.

Analogously, if $x_L \neq \fmin$,
then Proposition~\ref{prop:fp-approx-of-real-constraints}
applies and we have $\fsucc\bigl(\fpred(x_L) \mdiv_\rdown w_L\bigr)$.
On the other hand, if $x_L = \fmin$,
$\fsucc\bigl(\fpred(x_L) \mdiv_\rdown w_L\bigr) = \fmin$,
which is consistent with the fact that, for each non-NaN, nonzero
and finite $w \in \Fset \inters [+0, +\infty]$,
$\fmin$ is the lowest value of $y$ that satisfies
$\fmin = y \mmul_\rup w$.

A symmetric argument justifies~\eqref{eq:mult-inv-num-rup-rdown-upper}.

As before, we will consider both the cases
$\evalup{e^+_\ell} = \roundup{e^+_\ell}$ and $\evaldown{e^+_u} = \rounddown{e^+_u}$
as well as
$\evalup{e^+_\ell} > \roundup{e^+_\ell}$ and $\evaldown{e^+_u} < \rounddown{e^+_u}$.
Thus, when $\evaldown{e^+_u} < \rounddown{e^+_u}$
by~\eqref{eq:mult-refine-real-upper_bounds}
and~\eqref{prop:fp-approx-of-real-constraints:3}
of Proposition~\ref{prop:fp-approx-of-real-constraints}
we obtain $y \sleq \evalup{e^+_u}$.
Instead, when $\evaldown{e^+_\ell} > \rounddown{e^+_\ell}$,
by~\eqref{eq:mult-refine-real-upper_bounds}
and~\eqref{prop:fp-approx-of-real-constraints:1}
of Proposition~\ref{prop:fp-approx-of-real-constraints}
we obtain $y \sgeq  \evaldown{e^+_\ell}$.
In conclusion, for the case in which $\bar{r}_\ell = \rnear$,
since $e_u \neq 0$ and $e_\ell \neq 0$,
by Proposition~\ref{prop:fp-approx-of-real-constraints}, we have
\begin{align}
\label{eq:mult-inv-num-rnear-lower}
  y'_\ell &\defeq
    \begin{cases}
     \evalup{e^+_\ell},
        &\text{if $\feven(x_L)$ and $\evalup{e^+_\ell} = \roundup{e^+_\ell}$;} \\
          \evaldown{e^+_\ell},
        &\text{if $\feven(x_L)$ and $\evalup{e^+_\ell} \neq \roundup{e^+_\ell}$;} \\
      \fsucc\bigl(\evaldown{e^+_\ell}\bigr),
        &\text{otherwise;}
    \end{cases} \\
\intertext{%
whereas, for the case in which $\bar{r}_u = \rnear$,
}
\label{eq:mult-inv-num-rnear-upper}
  y'_u &\defeq
    \begin{cases}
      \evaldown{e^+_u},
        &\text{if $\feven(x_U)$ and $\evaldown{e^+_u} = \rounddown{e^+_u}$;} \\
          \evalup{e^+_u},
        &\text{if $\feven(x_U)$ and $\evaldown{e^+_u} \neq \rounddown{e^+_u}$;} \\
      \fpred\bigl(\evalup{e^+_u}\bigr),
        &\text{otherwise.}
    \end{cases}
\end{align}

An analogous reasoning with $z < 0$ allows us to obtain
the case analyses marked as $a_3^-$ and $a_8^-$.
\end{delayedproof}

\end{document}